\documentclass[11pt,letterpaper]{article}
\usepackage[utf8]{inputenc}			
\usepackage{makeidx}
\usepackage[pdftex]{graphicx}
\usepackage{mathtools,amsfonts}
\usepackage{array}
\usepackage{multirow, hhline}
\usepackage{paralist}
\usepackage{bm,xfrac,xspace,setspace, xr-hyper}
\usepackage{enumitem}
\usepackage[textfont={small,it}, labelfont={small,bf}]{caption}
\usepackage[usenames,dvipsnames,table]{xcolor}
\usepackage{wrapfig}
\usepackage[english]{babel}		
\usepackage[margin=1in]{geometry}	
\usepackage{lipsum}			
\usepackage{comment}
\usepackage{chngpage}			
\usepackage{amsfonts,amssymb,amsmath, amsthm}
\usepackage{amsopn}
\usepackage{color}
\usepackage{nicefrac}
\usepackage{tikz}
\usetikzlibrary{arrows}
\usetikzlibrary{calc}

\usepackage[noend]{algpseudocode}
\usepackage[Algorithm]{algorithm}
\usepackage{subcaption}

\usepackage{tabto}
\usepackage{csquotes} 
\usepackage{multirow}
\usepackage{amsthm}
\usepackage[colorlinks,urlcolor=black,citecolor=black,linkcolor=black,menucolor=black]{hyperref}
\newtheorem{theorem}{Theorem}
\newtheorem*{theorem*}{Theorem}

\newtheorem{corollary}[theorem]{Corollary}
\newtheorem{definition}[]{Definition}
\newtheorem{lemma}[theorem]{Lemma}

\newtheorem{claim}[theorem]{Claim}

\theoremstyle{definition}

\newtheorem{remark}[theorem]{Remark}

\newtheorem*{problem*}{Problem} 

\newtheorem*{fact*}{Fact}
\newtheorem{observation}[theorem]{Observation}

\newtheorem{assumption}[theorem]{Assumption}

\newcommand{\ip}[1]{\langle #1\rangle}

\newcommand{\R}{\mathbb R}

\newcommand{\pp}{{\bm p}}

\DeclareMathOperator{\argmin}{arg\,min }
\DeclareMathOperator{\NSW}{NSW}
\DeclareMathOperator{\logd}{RelDist}

\makeatletter
\newcommand\footnoteref[1]{\protected@xdef\@thefnmark{\ref{#1}}\@footnotemark}
\makeatother
\title{Polynomial Time Algorithms to Find an Approximate Competitive Equilibrium for Chores} 
\usepackage{authblk}
\author[1]{Shant Boodaghians\footnote{Supported by NSF CAREER award 1750436}}
\author[2]{Bhaskar Ray Chaudhury}
\author[1]{Ruta Mehta\footnote{Supported by NSF CAREER award 1750436}}
\affil[1]{\textit{University of Illinois Urbana-Champaign,} \texttt{\{boodagh2,rutamehta\}@illinois.edu}}
\affil[2]{\textit{Max Planck Institute for Informatics, Saarland Informatics Campus,} \texttt{braycha@mpi-inf.mpg.de}}

\date{} 
\renewcommand{\d}{D}
\newcommand{\vd}{{\overrightarrow{\d}}}
\newcommand{\OPT}{\mathit{OPT}}
\newcommand{\eCEEI}{{$\varepsilon$-CEEI}\xspace}

\renewcommand{\epsilon}{\varepsilon}

\begin{document}
\pagestyle{empty}

\maketitle
\thispagestyle{empty}

\begin{abstract}
{\em Competitive equilibrium with equal income (CEEI)} is considered one of the best mechanisms to allocate a set of items among agents fairly and efficiently. In this paper, we study the computation of CEEI when items are chores that are disliked (negatively valued) by agents, under 1-homogeneous and concave utility functions which includes linear functions as a subcase. It is well-known that, even with linear utilities, the set of CEEI may be non-convex and disconnected, and the problem is PPAD-hard in the more general exchange model. In contrast to these negative results, we design FPTAS: A polynomial-time algorithm to compute $\varepsilon$-approximate CEEI where the running-time depends polynomially on $\frac{1}{\varepsilon}$.

Our algorithm relies on the recent characterization due to Bogomolnaia et al.~(2017) of the CEEI set as exactly the KKT points of a non-convex minimization problem that have all coordinates non-zero. Due to this {\em non-zero} constraint, na\"ive gradient-based methods fail to find the desired local minima as they are attracted towards zero. We develop an {\em exterior-point method} that alternates between guessing {\em non-zero} KKT points and maximizing the objective along supporting hyperplanes at these points. We show that this procedure must converge quickly to an approximate KKT point which then can be mapped to an approximate CEEI; this exterior point method may be of independent interest. 

When utility functions are linear, we give explicit procedures for finding the exact iterates, and as a result show that a stronger form of approximate CEEI can be found in polynomial time. Finally, we note that our algorithm extends to the setting of un-equal incomes (CE), and to mixed manna with linear utilities where each agent may like (positively value) some items and dislike (negatively value) others. 

\end{abstract}
\newpage

\clearpage
\setcounter{page}{1}
\pagestyle{plain}

\section{Introduction}
Allocating a set of items among agents in a non-wasteful {\em (efficient)} and agreeable {\em (fair)} manner is an age old problem extensively explored within economics, social choice, and computer science. 
An allocation based on {\em competitive equilibria (CE)} has emerged as one of the best mechanisms for this problem due its remarkable fairness and efficiency guarantees~\cite{AD,Varian74,BogomolnaiaMSY17}. 
The existence and computation of competitive equilibria has seen much work when all the items are goods, 
\textit{i.e.}\ liked (positively valued) by agents. 
However, when items are chores, \textit{i.e.}\ disliked (negatively valued) by agents, 
the problem is relatively less explored even though it is as relevant in every day life; 
for example dividing teaching load among faculty, job shifts among workers, and daily household chores among tenants. 

In this paper, we study the problem of computing competitive equilibria with {\em equal income} (CEEI) \cite{Varian74,BogomolnaiaMSY17} for chore division, where a set of $m$ {\em divisible} chores has to be allocated among a set of agents. 
Agents receive {\em payments} for doing chores, and are required to earn a minimum amount, and under {\em equal income}, these amounts are the same.\footnote{The earning requirement of an agent can also be thought of as her importance/weight compared to others, and thereby under {\em equal income} all agents have the same weight.}
A competitive equilibrium (CE) for chores consists of a payment 
per-unit for each chore, and an allocation of chores to agents such that every agent gets her {\em optimal bundle}, {\em i.e.,} the disutility-minimizing bundle subject to fulfilling her earning requirement. 
Typically, agent preferences are represented by a monotone and concave {\em utility function} \cite{AD,BogomolnaiaMSY17}, that is negative and decreasing in case of chores. 
Equivalently, we consider {\em disutility functions}, namely $\d_i:\R^m_+ \rightarrow \R_+$ for agent $i$, 
that is monotone increasing and convex. 
We assume disutility functions to be 1-homogeneous as otherwise the problem is known to be intractable \cite{CThard,CGMMhard}. 
We note that 1-homogeneous functions form a rich class that includes the well-studied linear and CES functions as special cases.

The computational complexity of CE is well-understood when items are goods, {\em e.g.,} ~\cite{DevanurPSV08,ChenDDT09,chen2017complexity,VaziraniY11,ColeDGJMVY17,Rubinstein18} (see Section \ref{sec:relWork} for a detailed discussion): 
for 1-homoge\-neous utilities, the famous Eisenberg-Gale~\cite{EG} convex programming formulation and its dual are known to give equilibrium allocation and prices respectively. As a consequence the set of CE is convex, and the ellipsoid and/or interior point methods would find an approximate CE in polynomial-time, assuming utility functions are {\em well-behaved}. 
When utility functions are further restricted to be linear, there are many (strongly) polynomial time combinatorial algorithms known~\cite{DevanurPSV08, Orlin10}, even for the more general {\em exchange model} where agents want to {\em exchange} items they own to optimize their utilities~\cite{DuanM15, DuanGM16, GargV19}.

Although goods and chores problems seem similar, results for chores are surprisingly contrasting: 
Even in the restricted case of linear disutilities, the set of CEEI can be non-convex and disconnected~\cite{BogomolnaiaMSY17, BogomolnaiaMSY19}, 
and in the exchange model computing a CE is PPAD-hard~\cite{CGMMhard}.
No polynomial time algorithms are known to find CEEI with chores, except for when number of agents or number of chores is a constant~\cite{BranzeiS19, GargM20}.\footnote{These algorithms are based on enumeration from a cleverly designed set of candidates.
Similar approaches are known for goods manna when the number of items or agents is a constant \cite{DevanurK08, GargMSV15}, while the general case is PPAD-hard even to approximate \cite{CThard, Rubinstein18}} 
We note that the combinatorial approaches known for the goods case \cite{DevanurPSV08, Orlin10, Vegh12} seem to fail due to disconnectedness of the CEEI set (see Remark \ref{rem:comb} for further explanation). 
In light of these results, computing exact CEEI may turn out to be hard even with linear disutilities, 
but what about an approximate CEEI?

We resolve the above question by designing an FPTAS for the more general class of 1-homogeneous disutilities. Specifically, we design 
an algorithm to find $\epsilon$-approximate CEEI in time polynomial in $\frac{1}{\epsilon}$ and bit-size of the input instance parameters. 
We remark that many of the above bottlenecks exist even when we focus on approximate CEEI. In particular, the set of approximate CEEI can be non-convex and disconnected. And the fundamental bottleneck in generalizing the combinatorial algorithm explained in Remark~\ref{rem:comb} still persists. 
Despite these challenges, we are able to design an FPTAS to find an approximate CEEI, and extend it to more general valuations than linear which includes CES valuation functions.

Our algorithm crucially builds on the characterization of Bogomolnaia et al.~\cite{BogomolnaiaMSY17}, which states that the set of CEEI is exactly the {\em strictly positive} local-minima (KKT points) of a non-convex formulation, namely minimize
the {\em product of disutilities (equivalently $\sum_i \log{d_i}$)} over the space of {\em feasible} disutility vectors. 
The set of feasible disutility vectors may not be convex, but they can be made convex by allowing overallocation.
Unfortunately, standard interior-point methods for finding local optimum, such as gradient descent, will fail at ensuring the strict positivity constraint, since the gradient of the objective is attracted towards the minimum disutility coordinate. This difficulty is not alleviated by barrier function methods either. A possible fix is to introduce additional constraints to avoid zeros, but then we loose the CEEI characterization.

The above issues would not arise if we {\em maximize} $\sum_i \log{d_i}$ instead of minimizing it. Motivated from this observation, 
we design an {\em exterior-point method} that tries to maximize the objective {\em outside} of the feasible region, starting from an outside point that is below the lower-hull. 
However, we are faced with two crucial difficulties: $(i)$ now the {\em outside region} is truly non-convex, 
and $(ii)$ we must ensure that we do find a desired local minimum from the inside. 

Our exterior-point method handles the above issues 
by repeatedly guessing candidate solutions, and checking if they are local minima for the problem inside the feasible region by verifying if the gradient is parallel to some supporting hyperplane.
If not, it goes on to try another such candidate, while ensuring it is always increasing along the objective function. 
Thus, the objective acts as a potential function, and we can bound convergence rates by the size of objective improvement at each step. 
This method may be of independent interest. 
We terminate search when the supporting hyperplane direction is {\em approximately} equal to the gradient, in a multiplicative sense, and argue that such an approximate KKT point suffices to guarantee an approximate CEEI.

The crucial step in each iteration of this procedure is to find the nearest feasible point in the disutility space, which allows us to find a boundary point along with a supporting hyperplane at it. 
When disutility functions are linear, we argue that both distance minimization and supporting hyperplane computation can be solved exactly, leading to a stronger form of approximate CEEI.

For the case of general 1-homogeneous and convex disutility functions,
the nearest feasible point must be found by interior point methods. 
We assume black-box access to the disutility functions' value and partial derivatives.
This approximate nearest-point computation introduces errors in the local optimum and the supporting hyperplane both, that are tricky to handle. 
We show how to handle these extra errors by modifying the algorithm, 
and argue that a slight weakening of approximately competitive equilibria can still be guaranteed. 
As expected, these guarantees, including successful application of the interior point method, rely on the disutility functions being ``well-behaved'', and the running time of the algorithm depends logarithmically on continuity parameters of the disutility functions, 
namely, the Lipschitz constants for lower-bounding and upper-bounding the partial derivatives.

\noindent{\em Extensions.} Finally, we argue how our algorithm easily extends to the setting of un-equal income (CE) when max to min income/weight ratios is polynomially bounded.  
Another natural extension we consider is to mixed manna, where each agent may like some items and dislike others. Again, using the characterization of~\cite{BogomolnaiaMSY17}, every instance can be put into one of the three categories, namely positive, negative, and null. 
We argue that the instance in the positive category can be solved using the Eisenberg-Gale convex program~\cite{EG}, and those in null have a trivial solution. For instances in the negative category, we discuss how our algorithm can be extended with simple modifications. 
\medskip

\noindent{\em Linear Disutilities with Infinities.}
We note that, \cite{CGMMhard} that shows PPAD-hardness for the linear exchange model allows an agent to have infinite disutility for some chores indicating they do not have skills to do the chore in a reasonable amount of time. 
Our algorithm extends to this model as well, since their sufficiency conditions to ensure existence of equilibrium dictates that every component of the bipartite graph between agents and chores with finite disutility edges should be a complete bipartite graph. 
They show that even CEEI may not exist without this condition, and checking if it exists is NP-hard. 
Under this condition, it suffices to find CEEI for each of the connected component separately where there are no agent-chore pairs with infinite disutility. 
\medskip

In order to convey the main ideas cleanly we mainly focus on CEEI with chores in what follows, and discuss the extensions to CE and mixed manna at the end of the paper.

\subsection{Model and Our Results}\label{sec:model}
In the chore division problem, a set of $m$ {\em divisible} chores $[m]:=\{1,\,\dotsc,\,m\}$ is to be allocated to a set of $n$ agents $[n]:=\{1,\,\dotsc,\,n\}$. It is without loss of generality to assume that exactly one unit of each chore needs to be allocated. 
Agent $i$'s preferences (over chores) is represented by a non-negative, non-decreasing, and convex disutility function $\d_i \colon \mathbb{R}^m_{\geq 0} \rightarrow \mathbb{R}_{\ge 0}$.\footnote{Typically, agents' preferences for chores are represented by non-positive, non-increasing, and concave utility functions since agents dislike chores \cite{BogomolnaiaMSY17}. 
By taking the negation of these utility functions we get non-negative, non-decreasing, convex disutility functions that agents want to minimize.} 
We denote by $x_{ij}$ the fraction of item $j$ that is allocated to agent $i$, and we denote $\bm x_i := (x_{i1},\,\dotsc,\, x_{im})$. 
We assume that $\d_i$'s are 1-homogeneous, \textit{i.e.}
\begin{equation}\label{eq:intro-hom}
\d_i(a\cdot \bm x_i) = a\cdot \d_i(\bm x_i)\quad \text{ for all $\bm x_i$, ~ and all $a\ge 0$.}
\end{equation}

If $\d_i(\cdot)$ is linear, then it is represented by $\d_i(\bm x_i) = \sum_{j \in [m]} \d_{ij} \cdot x_{ij}$ where $\d_{ij}\in (0, \infty)$ is the disutility of agent $i$ per unit of chore $j$.\footnote{If for some $(i,j)$ pair $D_{ij}=0$ then chore $j$ can be freely allocated to agent $i$, and can be removed. Infinite disutilities can be handled as discussed in the introduction.}
Equivalently, we write $\d_i(\bm x_i) = \ip{ \bm \d_i, \bm x_i }$ where 
$\bm \d_i = (\d_{i1}, \d_{i2}, \dots , \d_{im})$. We also use $\vd (\bm x)$ to denote the disutility vector $(D_1(\bm x_1), D_2(\bm x_2), \dots, D_n(\bm x_n))$.

\paragraph{Competitive equilibrium with equal income (CEEI)}
At a CE with chores, payments are linear, and the $j$-th chore pays $p_j$ per unit of the chore assigned. 
Let $\pp =(p_1,\,\dots,\,p_m)$ denote the vector of payments, and then the payment to agent $i$ is $\ip{\pp,\bm x_i}$.
Each agent seeks to minimize their disutility subject to being paid at least 1 unit. 
We note that, under {\em equal income}, the exact value being paid is immaterial so long as all agents get paid the same amount.
Prices $\pp$ and allocation $\bm x = (\bm x_1, \bm x_2, \dots , \bm x_n)$ are said to be at CEEI if all the chores are fully allocated when every agent consumes her least-disliked bundle with payment at least 1, {\em i.e.,} an optimal bundle. 
Formally~\cite{Varian74,BogomolnaiaMSY17}

	\begin{itemize}[label={(\arabic*)}]
		\item[$(E1)$]\label{EP1} {\em (equal payments)} for all agents $i$ and $i'$ we have $\langle \bm x_i, \bm p \rangle =  \langle \bm x_{i'}, \bm p \rangle$, and
		\item[$(E2)$]\label{EP2} {\em (optimal bundle)} for all $i \in [n]$, we have  $\d_i(\bm x_i) \leq \d_i(\bm y_i)$ for all $\bm y$ s.t. $\langle \bm y_i, \bm p \rangle \geq \langle \bm x_i, \bm p \rangle$, and 
		\item[$(E3)$]\label{EP3} {\em (feasible allocation)} for all $j \in [m]$, we have $\sum_{i \in [n]} x_{ij} =1$.
	\end{itemize} 

It is known that the set of CEEI may be nonconvex, or even disconnected \cite{BogomolnaiaMSY17}.
In light of this fact, and the PPAD-hardness of CE in the linear-exchange model \cite{CGMMhard}, 
we turn our attention to approximately competitive equilibria. 
We formalize the notion of \eCEEI as follows:

\begin{definition}\label{def:CE}
	Prices $\pp$ and allocation $\bm x$ are termed a \eCEEI for an $\varepsilon\ge 0$, if and only if
	\begin{enumerate}[label={(\arabic*) }]
		\item \label{P1} for all agents $i$ and $i'$, we have $(1-\varepsilon) \cdot \langle \bm x_i, \bm p \rangle \leq \langle \bm x_{i'}, \bm p \rangle$, and  
		\item \label{P2} for all $i \in [n]$, we have and $(1-\varepsilon) \cdot d_i(\bm x_i) \leq d_i(\bm y_i)$ for all $\bm y$ such that $\langle \bm y_i, \bm p \rangle \geq \langle \bm x_i, \bm p \rangle$, and  
		\item \label{P3} for all $j \in [m]$, we have $1 - \varepsilon \leq \sum_{i \in [n]} x_{ij} \leq 1+\varepsilon$. 
	\end{enumerate} 
\end{definition} 

It is well known that CEEI satisfy well-sought-after fairness and efficiency notions of {\em envy-freeness} and {\em Pareto-optimality} respectively. 
An allocation $\bm x$ is said to be envy-free (EF) if every agent prefers their own bundle over that of any other agent. 
And it is said to be Pareto-optimal (PO) if no other allocation Pareto-dominates it, {\em i.e.,} 
there is no feasible allocation $\bm y$ such that $\d_i(\bm y_i)\leq \d_i(\bm x_i)$ for all $i$, and for some agent $k$, $\d_i(\bm y_i)<\d_i(\bm x_i)$.
In Appendix \ref{app:ef-po} we show that an \eCEEI allocation approximately guarantees these properties.

Our main contribution in this paper is an FPTAS -- a polynomial time algorithm to find an \eCEEI where the running time depends polynomially on $\frac{1}{\epsilon}$; proved formally in Section \ref{gen:CE}. Informally, lets say that function $\d_i$ is {\em $L$-well-behaved} if it satisfies Assumption \ref{as:lip} regarding it's derivatives.  
\begin{theorem*}
	Given black-box access to 1-homogeneous and convex disutilities $\d_{1},\,\dotsc,\,\d_{n}$ 
	that are $L$-well-behaved, and also to their partial derivatives, 
	Algorithm~\ref{alg2:KKT}, finds an $\varepsilon$-CEEI in time polynomial in $n$, $m$, $1/\varepsilon$, and $\log(L)$.
\end{theorem*}

We note that our result holds under a weaker assumption than of Assumption~\ref{as:lip};  discussed briefly in Remark~\ref{extendingtoCES}.
For linear disutilities, we show the following stronger guarantee in Section \ref{sec:CE}.

\begin{theorem*}
Given an instance $I$ with linear disutility functions represented by  $\d_{11},\dots,\d_{nm}>0$, and an $\epsilon>0$, a stronger $\epsilon$-CEEI can be computed in time $\textup{poly}\left(n,m,\log\left(\frac{\max_{ij} \d_{ij}}{\min_{ij} \d_{ij}}\right),\frac{1}{\epsilon}\right)$ where no error is incurred in the last two conditions, {\em i.e.,} $(\bm x,\bm p)$ that satisfies \ref{P1}, $(E2)$, and $(E3)$.
\end{theorem*}

More importantly, our algorithm is an {\em exterior point method} that builds on tools from continuous optimization to find an {\em approximate KKT point}, which may be of independent interest. 
Next we give an overview of this method and it's analysis.

\subsection{Overview of the Algorithm and Analysis}\label{sec:overview}
Our algorithm builds on the following characterization of CEEI due to Bogomolnaia et al.~\cite{BogomolnaiaMSY17}:  
Analogous to the convex program of Eisenberg and Gale~\cite{EG}, the CEEI in the case of bads are characterized as local minima to the product of disutilities. However, this optimization program is over {\em disutility space}, rather than {\em allocation space}. 

Formally, we let $\mathcal F$ denote the set of feasible allocations, namely 
\begin{subequations}\label{def:F}
\begin{equation}	
	\mathcal F:=\left\{\bm x\in \mathbb R^{nm}\,\middle|
		\, \textstyle\sum_{i}x_{ij}=1\ \ \forall\,j,\ x_{ij}\geq 0\ \ \forall\,i,j\right\}\ .
\end{equation}
The disutility space $\mathcal D$ will be the set of all disutility profiles which can be attained over $\mathcal F$, or 
\begin{equation}	
	\mathcal D:=\left\{\bm d=(d_1,\,\dotsc,\, d_n)\in \mathbb R^{n}\,\middle|
		\, \exists\,\bm x=(\bm x_1,\,\dotsc,\,\bm x_n)\in \mathcal F:\ \d_i(\bm x_i) = d_i\ \forall\,i\right\}\ .
\end{equation}
\end{subequations}
In all that follows, we will distinguish between disutilities as functions and as variables by the upper- and lower-case symbols respectively.
When disutility functions are linear, $\mathcal D$ is a polytope. 
However, for more general convex disutility functions, $\mathcal D$ may not be a convex set. 
We will remedy this by working instead with the extended feasible region $\mathcal D+\mathbb R_{\geq 0}^n$, the Minkowski sum, which we show is convex (Claim \ref{claim:convex}). 
This is the set of all disutility profiles which are {\em at least as bad} as some feasible profile, {\em i.e.,} disutility profiles attainable at over-allocations of the chores.

The characterization of Bogomolnaia et al.~\cite{BogomolnaiaMSY17} states that any disutility profile $\bm d$ which is a local minimum (KKT point) to the following non-convex minimization program is the disutility profile of {\em some} CEEI, and the prices and allocation of this CEEI can be found by understanding $\bm d$ in allocation-space.
\begin{equation}
	\nonumber
	\min_{\bm d\in \mathcal D}\ \textstyle\prod_{i=1}^n d_i\quad \text{s.t.}\ d_i>0\ \forall\, i\ .
\end{equation}
Note that minimizing over the set $\mathcal D$ is equivalent to minimizing over the extended set $\mathcal D+\mathbb R_{\geq 0}^n$. 
And the KKT points to this program are equivalent to the KKT points for the minimization of the logarithm of objective, $\mathcal L(\bm d):=\sum_{i=0}^n\log(d_i)$. Hence the above program can equivalently stated as,
\begin{equation}\label{eq:min-log}
	\min_{\bm d\in \mathcal D+\mathbb R_{\geq 0}^n}\ \textstyle \mathcal L(\bm d) = \sum_{i=0}^n\log(d_i) \quad \text{s.t.}\ d_i>0\ \forall\, i\ .
\end{equation}

\paragraph{Primary difficulty.}
The open constraints $d_i>0$ are both fundamental to the above characterization, and the source of the main difficulty of the problem. 
Any disutility profile $\bm d$ with a zero coordinate are trivial optima to these minimization problems, 
but are economically meaningless since no fairness or efficiency properties can be guaranteed. 
Furthermore, any na\"ive interior-point attempt at finding local minima are attracted by these constraints: 
the gradient of the objective $\mathcal L$ at $\bm d$ is inversely proportional to $\bm d$ componentwise, since $\frac{\partial}{\partial d_i}\mathcal L = 1/d_i$. 
This has the effect of accelerating gradient descent towards the $d_i\ge 0$ constraint for the smallest $d_i$ value. 
See Figure~\ref{fig:gd-fails} (red) for an illustration.
This effect is robust to barrier methods at the boundaries, and thus gradient-following methods are not helpful in this task. 
The same problem afflicts attempts at strengthening the constraint to $d_i>\eta$ for some small $\eta$, as the dual variables for these constraint break the CEEI characterization. 
We circumvent this issue by designing an iterative {\em exterior-point method} that always increases the log-sum $\mathcal L$. 

\begin{figure}[t]
\begin{minipage}[c]{0.3\textwidth}
    \includegraphics[scale=1.1]{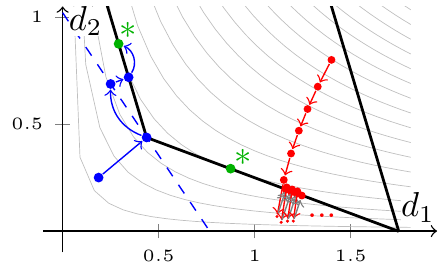}
  \end{minipage}\hfill
  \begin{minipage}[c]{0.67\textwidth}
\caption{A representation of both the exterior point method (in blue), and the pitfalls of gradient descent (in red) in minimizing the objective inside the feasible region. The light gray lines denote the level sets of the objective, orthogonal to the gradient.
\textcolor{blue}{\textbf{Blue}}: For the exterior point method, we start outside the region, find a nearest point and supporting hyperplane, jump to a new exterior point, and repeat until we find one of the two on-face local optima (green). 
\textcolor{red}{\textbf{Red}}: For the gradient descent, we start inside the region and quickly accelerate towards the $d_2 \ge 0$ boundary, which we wish to avoid. }\label{fig:gd-fails}
  \end{minipage}
\end{figure}

\begin{remark}\label{rem:comb}
A natural question is if the combinatorial methods known for computing CE in linear Fisher (exchange) model with goods, \textit{e.g.}, \cite{DevanurPSV08, Orlin10}, extend to chores with linear disutilities? Unfortunately, they do not.  In particular, the non-convexity and disconnectedness of the CEEI set is a primary difficulty in extending any algorithm from the goods setting to the chores setting. For example, these methods rely on the fact that CE allocation and prices changes continuously with the the (money) endowments of the agents \cite{MegiddoV07}, which is not true with chores. A chore division instance may have multiple disconnected equilibria some of which may disappear as we change these parameters, and as a result the said methods may get stuck. This fundamental bottleneck persists even if we restrict ourselves to approximate-CEEI.
\end{remark}

In the rest of this section, we will outline our approach first for linear disutility functions, and then afterwards in the general case. 
In the linear setting, many sub-routines can be solved exactly. Thus, it requires less technical detail to present, and serves as a good intuition for the more involved general case that we address later. 

\subsubsection{Linear Disutilities: Relating Approximate KKT to Approximate CEEI}
The main insights of our result are that (1) approximate KKT points allow for approximately competitive equilibria to be constructed, and (2) approximate KKT points can be found, despite the difficulty described above about ensuring strict positivity constraints.  
We begin here by formally defining the approximate KKT conditions.

Recall, KKT points are local optima where the cone of  normal vectors of the tight constraints contains the function's gradient. 
Equivalently, there exists a supporting hyperplane at the local optimum whose normal vector is parallel to the function's gradient.
Formally, $\bm d$ is a KKT point for the problem $\min_{\bm d\in \mathcal D} \mathcal L(\bm d)$ if there exists a normal vector $\bm a$ such that: $\{\bm y\in \mathbb R^n | \ip{\bm a, \bm y}\geq \ip{\bm a,\bm d}\}$ is a supporting hyperplane for $\mathcal D+\mathbb R^n_{\geq 0}$, and there exists some $c>0$ such that $\bm a = c\cdot \nabla\mathcal L(\bm d)$, \textit{i.e.}\ $a_i=c/d_i$ for all $i$.

Since our procedure is iterative, it will converge in the limit to a KKT point, but only approximately after finitely many iterations.
We show that after a polynomial number of iterations, it finds an {\em approximate KKT point}, defined below, with inverse-polynomial error. 

\begin{definition}[$\gamma$-Approximate KKT]\label{def:apx-KKT-lin}
For $\gamma \ge 1$, we say a point $\bm d$ along with the normal direction $\bm a$ is a $\gamma$-KKT point for problem \eqref{eq:min-log} if 
		\begin{itemize}
		\item[(1)] $\bm d\in \mathcal D+\mathbb R_{\geq 0}^n$,\ \  (2) $\gamma_i^{-1} \leq a_i \cdot d_i \leq \gamma_i$ for all $i$, and 
		(3) $\mathcal D +\mathbb R^n_{\geq 0} \subseteq \{\bm y\in \mathbb R^n | \ip{\bm a, \bm y}\geq \ip{\bm a,\bm d}\}$.
	\end{itemize}
				
	  Informally, each entry of $\bm a$ is a $\gamma$-approximation of $\bm 1/\bm d$, the gradient of $\mathcal L$, and $\bm a$ is normal to a supporting hyperplane for $\mathcal D +\mathbb R^n_{\geq 0}$ at $\bm d$.
Furthermore, we say $\bm d$ is a $\gamma$-KKT point if there exists a vector $\bm a$ such that $(\bm d,\bm a)$ satisfy the above conditions. 
\end{definition}

Recall, when disutilities are linear, $\mathcal D$ is a linear polytope and is therefore convex. Hence, the above definition need not be defined over $\mathcal D+\mathbb R_{\geq 0}^n$, but we introduce it as it will be necessary later.

We outline here the first insight of our result, that approximate local minima give approximate equilibria. 
In the original analysis of the Eisenberg-Gale program~\cite{EG} for goods, and more notably in the proof of ~\cite{BogomolnaiaMSY17} for chores, 
the relationship between competitive equilibria and local maxima hinges on the gradient being inversely proportional to the marginal (dis)utility-per-dollar incurred. Intuitively, the KKT conditions enforce that the payment to each agent (in the chores setting) is perfectly balanced by their disutility incurred, and their payment is equal to that of any other player. 
It can be shown that if some player is paid more, then the KKT conditions are violated. 
Thus, we can conclude condition~\ref{P1} of Definition~\ref{def:CE}, with $\varepsilon=0$. 
Condition~\ref{P2} is argued using the fact that an agent is only allocated her minimum disutility-per-dollar chores, and \ref{P3} is true by definition since the allocation lies in $\mathcal F$.

To extend this argument to the approximate setting, it suffices to observe that when multiplicative error is introduced in the gradient direction, then this argument suffers only multiplicatively. 
A $\gamma$-sized error bound in the gradient direction allows for some player to be paid $\gamma$ less than the unit, and another $\gamma$ more, which allows us to show that $\gamma$-KKT points satisfy condition~\ref{P1} of Definition~\ref{def:CE} with $\varepsilon = (1-\gamma^2)$. 
As above, condition \ref{P2} is argued similarly with the same $\varepsilon$, 
and condition \ref{P3} holds with $\varepsilon=0$, again by feasibility.
Formally, by extending the argument of Bogomolnaia et al.~\cite{BogomolnaiaMSY17}, we show the following.

\def\linearTheoremKKTGivesCE{Let $(\bm d, \bm a)$ be a $(1+\varepsilon)$-KKT point for the problem of minimizing $\mathcal L(\bm d)$ subject to $\bm d\in \mathcal D$, and $\mathcal L(\bm d) > -\infty$. 
	Let $\bm x\in \mathcal F$ be any allocation that realizes $\bm d$, \textit{i.e.}\ $\d_i(\bm x_i) = d_i$ for all $i$. 
	Then there exists payments $\bm p=(p_1,\,\dotsc,\,p_m)$ such that $(\bm x,\bm p)$ form a stronger $2\varepsilon$-CEEI, where no error is incurred in the last two conditions, {\em i.e.,} $(\bm x,\bm p)$ satisfies \ref{P1}, $(E2)$, and $(E3)$. 
	
	Furthermore, when disutilities are linear, the allocation $\bm x$ and payments $\bm p$ can be computed exactly in polynomial time from the disutility profile $\bm d$ and normal vector $\bm a$.}

\begin{theorem*}
	\linearTheoremKKTGivesCE
\end{theorem*}

This theorem is proven in Section \ref{definition-of-vector-division}, Theorem~\ref{corr:approx}, and the first part of it does not require that the disutilities be linear.
However, as we will see below, $\gamma$-KKT points can only be guaranteed when disutilities are linear, and the definitions will need to be modified for the general case. The allocation and prices can be efficiently computed when disutilities are linear because they are the solutions to linear feasibility problems. 
With this theorem in hand, it remains therefore to compute $(1+\varepsilon)$-KKT points, discussed next.

\subsubsection{Linear Disutilities: Exterior Point Methods for Approximate KKT Points.}
Here we discuss our approach to find approximate-KKT point in polynomial time; formal details are presented in Section \ref{sec:lin-kkt}.
As discussed above, it is tempting to hope that interior-point methods will find local minima efficiently, but they will not work in this setting. 
Instead, we will rely on the geometry of the feasible space and objective function to allow us to repeatedly make guesses at KKT points, all the while increasing along the objective $\mathcal L(\bm d) = \sum_i\log(d_i)$, which we treat as a potential function.
This potential will ensure that if we do not find approximate KKT points, then we make significant progress, dependent on the degree of precision $\gamma$ needed.
By bounding the values that the potential can take, this will suffice to show that the procedure is an FPTAS. Refer to Figure \ref{fig:gd-fails} (blue) for a pictorial representation of the algorithm.

Our ``guesses'' at KKT points are made by starting with an exterior, infeasible point $\bm d$, and finding the nearest feasible point $\bm d_*$ to it. 
Formally, $\bm d_*$ is the solution to $\min_{\bm y\in \mathcal D} \Vert\bm y-\bm d\Vert^2_2$. 
Using the fact that it is the nearest point to $\bm d$ in the $\ell_2$ sense, we show that $\bm a = \bm d_*-\bm d$ is normal to a supporting hyperplane for $\mathcal D$ at $\bm d_*$ (Lemma \ref{lem:algo-steps-linear}).
Notice that, so long as $\bm d_*\geq \bm d$ componentwise, then this all still holds when replacing $\mathcal D$ with $\mathcal D+\mathbb R_{\geq 0}^n$.
Furthermore, this will ensure that we are increasing along the potential $\mathcal L$. 

It remains to find the start of the next iterate, while ensuring that we are increasing in the $\mathcal L$ direction. 
Note that we have that the hyperplane $\{\bm y\in \mathbb R^n|\ip{\bm a,\bm y} = \ip{\bm a,\bm d_*} \}$ is supporting for $\mathcal D$, and therefore none of the points on this hyperplane are in the interior of $\mathcal D$. 
Thus, we can choose our next starting point to be the $\mathcal L$-maximizing point on this hyperplane. 
Since $\bm d_*$ is also feasible, this will ensure that we are increasing in the $\mathcal L$ direction, 
and that we are starting from a new exterior, infeasible point.
This $\mathcal L$-maximizer on the hyperplane can be found efficiently, since we have a closed form for it: 
the maximizer on the hyperplane will be the point at which $\nabla \mathcal L$ is proportional to $\bm a$, and we show that it is exactly a rescaling of $(1/a_1,\,\dotsc,\,1/a_n)$ (Claim \ref{claim:inc1}).

Thus, the algorithm is iterative, and each round $k\ge 0$ proceeds as follows:
\begin{enumerate}\setcounter{enumi}{-1}
	\item\label{st0} $\bm d^k$ is the infeasible point ``lying below'' $\mathcal D$ starting the round. 
	$\bm d^0$ is any infeasible point.
	\item\label{st1} Set $\bm d^k_*$ to be the nearest feasible point to $\bm d^k$, \textit{i.e.}\ the solution to $\min_{\bm d\in \mathcal D+\mathbb R^n_{\geq 0}}\Vert \bm d-\bm d^k\Vert_2^2$.
	\item\label{st2} Set $\bm a^k\propto\bm d^k_*-\bm d_k$, rescaled so that $\ip{\bm a^k,\bm d^k_*}=n$.
	\item\label{st3} Define $\bm d^{k+1}$ to be $(1/a_1,\,\dotsc,\,1/a_n)$, the maximizer of $\mathcal L(\bm d)$ subject to $\ip{\bm a^k,\bm d}=n$.
	\item\label{st4} Stop if $\bm d^{k+1}$ is 
	``close enough'' to $\bm d^k_*$, otherwise repeat.
\end{enumerate}

We initialize the procedure at any infeasible point ``lying below'' the feasible region. 
When disutilities are linear, this can be found by noticing that we can lower-bound the disutilities over the feasible region, and picking an allocation which assigns half of the lower bound to each agent (Claim~\ref{claim:d0}). 
The normalization in Step~\ref{st2} ensures that if $\bm a^k$ is approximately parallel to the gradient $\nabla\mathcal L(\bm d^k_*)$, then it is also of the right magnitude. 
The notion of ``close enough'' in Step~\ref{st4} is multiplicative, as it measures increase in the potential function $\sum_{i=1}^n\log(d_i)$.

\paragraph{The Potential Function, and Convergence Rates.}
As discussed above, we wish to use $\mathcal L(\bm d) = \sum_{i=1}^n\log(d_i)$ as a potential function to measure the progress of the algorithm. 
For each iteration $k\geq 0$, we will have $\mathcal L(\bm d^k)\leq \mathcal L(\bm d^k_*)\leq \mathcal L(\bm d^{k+1})$ (Claim \ref{claim:inc1}). 
This first inequality is due to the observation that the nearest feasible point to $\bm d^k$ Pareto dominates it, 
and $\mathcal L$ is monotone increasing in each coordinate.
The second inequality is by construction, as we show that $\bm d^{k+1}$ maximizes $\mathcal L$ on a hyperplane that contains $\bm d^k_*$.

It remains then to argue that progress along $\mathcal L$ is rapid, relative to its range.
We noted above that the stopping condition in Step~\ref{st4} is multiplicative. 
Formally, we stop when the $\ell_1$ norm of the logarithmic difference, \textit{i.e.}\ $\sum_{i=1}^n\left|\log\left((\bm d^k_*)_i/(\bm d^{k+1})_i\right)\right|$ is at most $\varepsilon$. 
When this log-distance is more than $\varepsilon$, we show that the objective $\mathcal L$ increases by at least $\Omega(\varepsilon^2/n^2)$ (Lemma \ref{lem:inc2}).

Conversely, when this log-distance is upper-bounded by $\varepsilon$, we will show that $(\bm d^k_*,\bm a^k)$ form a $(1+\varepsilon)$-KKT point (Lemma \ref{lem:kkt-lin}).
Thus, since it is reasonable to bound $\log(\d_i(\bm x_i))$ over the feasible region, we will be able to bound the maximum number of iterations as a polynomial in $1/\varepsilon$, $n$, and $-\mathcal L(\bm d^0)$.
This allows us to argue that an approximate equilibrium may be found in polynomially many iterations.

\paragraph{Implementing Iterations in Polynomial Time.}
We have argued above that an approximate KKT point, and therefore an approximate equilibrium, can be found in polynomially many iterates of the exterior point method. 
However, it remains to show that each step can be solved efficiently. 

With the exception of Step~\ref{st1} above, the rest of the algorithm is arithmetic, which can be easily performed. 
The minimization problem in Step~\ref{st1} may pose a problem in general, if we expect an exact minimum. 
This is the source of the extra care needed in the general case. 
However, in the case of linear disutilities, we show that the minimization problem is actually a quadratic program with a semidefinite bi-linear form over $\mathcal F$ space (Lemma \ref{lem:iter-poly}), 
and methods for finding exact solutions to such programs have long been known~\cite{quadprog}.

Finally, we note that although each step in our algorithm generates polynomial sized rational numbers wrt it's parameters, 
one needs to be careful about how their bit-sizes grow. This can be taken care of by rounding down the $\bm d^k$ to a nearest rational vector with polynomial bit-size at the end of each iteration. Note that this step will ensure that $\bm d^k$ lies below $\mathcal D$ and we also argue why the bound on the iterations still hold: Since at every iteration $k$ of  our algorithm, the value of each $d^k_i$ can be lower bounded using the value of the potential at $\bm d^k$  and the upper bound on the maximum disutility values in $\bm d^k$,\footnote{Note that we start with a $\bm d^0$ where each agent has a non-negligible disutility,  and at any point in time, the disutilities of the agents in $\bm d^k$ are upper-bounded (as the disutility vector lies below $\mathcal D$), implying that there cannot be a significant increase in the disutility of any agent throughout the algorithm. Also, since the sum of logs of the disutilities $\mathcal L (\cdot )$ is increasing throughout the algorithm, we can conclude that there cannot be a significant decrease in the disutility of any agent throughout the algorithm, implying that the disutilities in $\bm d^k$ are also lower bounded. } such a rounding is possible without hitting the $d_i\ge 0$ boundary, and while ensuring at least $\Omega(\varepsilon^2/n^2)$ increase in the potential $\mathcal L(\cdot )$. 
However, to convey the main important technical ideas, in Section \ref{sec:CE} we focus on bounding number of arithmetic operations. And we note that the analysis of Section \ref{gen:CE} for the general case is robust to such a rounding. 

Putting all the above together, we get an FPTAS to compute stronger approximate CEEI where the last two conditions of \eCEEI are satisfied with $\varepsilon=0$ (Theorem \ref{MAINTHMLINEAR}).

\subsubsection{General 1-Homogeneous Disutilities}
\label{gen:overview}
In general, the disutility functions are 1-homogeneous and convex, and are given as a {\em value oracle} black-box, 
along with a value oracle for their partial derivatives. 
In this section we outline the new issues that arise in extending our algorithm, and their resolutions; see Section \ref{gen:CE} for formal details.

At a high-level the issues are as follows: First, to find the nearest points we need to employ {\em interior point} methods which returns approximate solutions, and in turn we incur error in the hyperplane as well as the gradient. Secondly, 
in order to use the interior point method, we will have to work in the allocation space and can not work with the disutility space directly. 
This causes problems as convex constraints in disutility space need not be convex in allocation space. We elaborate these two issues and also highlight how we overcome them. Finally, we give an overview of the entire algorithm by putting everything together.

\paragraph{Finding Approximate Nearest Point.} Recall, we have defined $\overrightarrow{\d}(\bm x):=\left(\d_1(\bm x_1),\,\dotsc,\,\d_n(\bm x_n)\right)$. The natural program to find the nearest point in $\mathcal D + \mathbb{R}^n_{\geq 0}$ to a point $\bm d $ below $\mathcal D$ (or equivalently outside $\mathcal D + \mathbb{R}^n_{\geq 0}$) requires finding a  $\bm x \in \mathcal F'$ that minimizes $|| \vd (\bm x) - \bm d||_2^2$, where $\mathcal F' := \{ \bm y \in \mathbb{R}^{nm}_{\geq 0} \mid \sum_{i \in [n]} (\bm y)_{ij} \geq 1 \text{ for all } j \in [m] \}$. Unfortunately the objective function is not necessarily convex\footnote{The natural sufficient condition for composition of two convex functions to be convex is if the outer function is monotone in the variables. We do not have this with our current objective function.}. One way to ensure it's convexity is to put additional constraints of the form $D_i(\bm x_i) \geq  \bm d_i$ for all $i \in [n]$. But again, since the disutility functions $D_i(\cdot)$ are convex, these constraints create non-convex feasible region. 

We come up with an alternative formulation for finding the approximate nearest point which is convex. The crucial observation is the fact that given any point $\bm d $ outside $\mathcal D + \mathbb{R}^n_{\geq 0}$ there exists no point $\bm d' \in \mathcal D + \mathbb{R}^n_{\geq 0}$, such that $\bm d$ \emph{Pareto-dominates}(coordinate-wise larger or equal) $\bm d'$, i.e., $\bm d'$ does not belong in the negative orthant centered at $\bm d$. Therefore, a point $\bm d \in \mathbb{R}^n_{\geq 0}$ can Pareto-dominate any point $\bm d' \in \mathcal D + \mathbb{R}^n_{\geq0}$ if and only if $\bm d \in \mathcal D + \mathbb{R}^n_{\geq 0}$ .  We now show how to use this fact to come up with a convex program to find the nearest point in $\mathcal D + \mathbb{R}^n_{\geq 0}$. Our goal is to find a vector $\bm \beta \in \mathbb{R}^n$ of smallest magnitude and a point $\bm d' \in \mathcal D + \mathbb{R}^n_{\geq 0}$ such that the point $\bm d + \bm \beta$ Pareto-dominates $\bm d'$: Note that this is only possible when $\bm d + \bm \beta \in \mathcal D + \mathbb{R}^n_{\geq 0}$. Since $||\bm \beta||^2_2$ is minimum, $\bm d + \bm \beta$ is the nearest point in $\mathcal D + \mathbb{R}^n_{\geq 0}$ to $\bm d$. Formally, 

\begin{equation*}
\begin{array}{ll@{}ll}
\text{minimize}  & \displaystyle \sum\limits_{i \in [n]}^{ } ((\bm \beta)_i)^2 \\
\text{subject to}& \displaystyle  \sum_{i \in [n]}  z_{ij} \geq 1, & &\forall j \in [m]\\ 
&     z_{ij} \geq  0,  & &\forall i \in [n], \forall j \in [m]\\
& D_i(\bm z_i) - (\bm d)_i - (\bm \beta)_i \leq 0, & &\forall i \in [n],
\end{array}
\end{equation*}

It is easy to verify that the above program minimizes a convex function over  a convex domain. The above convex program returns point $\bm z \in \mathcal F'$ such that $\vd (\bm z)$ is the nearest point in $\mathcal D + \mathbb{R}^n_{\geq 0}$ to $\bm d$. Unfortunately, this program cannot be solved exactly in polynomial time and therefore we need to argue about how to extract an approximate-CEEI given an approximate nearest neighbour.

\paragraph{Approximate Supporting Hyperplane and \texorpdfstring{$(\lambda, \gamma, \delta)$}{(l,g,d)}-KKT Points.}  In polynomial time, we can only find an approximate nearest neighbour of a point $\bm d$ in $\mathcal D + \mathbb{R}^n_{\geq 0}$. Therefore, our supporting hyperplanes will also be approximate, and therefore we need to redefine the approximate KKT points that we can compute. 
Let  $\vd (\bm z^*)$ be the nearest point in $\mathcal D + \mathbb{R}^n_{\geq 0}$ to $\bm d$. Then, $\vd (\bm z^*) - \bm d$ is normal to a supporting hyperplane of $\mathcal D + \mathbb{R}^n_{\geq 0}$ at $\vd (\bm z^*)$, i.e.,  $\ip{\vd (\bm z^*) - \bm d, \bm y} = \ip{\vd (\bm z^*) - \bm d, \vd (\bm z^*)}$ is a supporting hyperplane of $\mathcal D + \mathbb{R}^n_{\geq 0}$ at $\vd (\bm z^*)$. Since we have access only to an approximate nearest neighbour of $\bm d$, say $\vd (\bm z')$, we wish to have $\ip{\vd (\bm z') - \bm d, \bm y} = \ip{\vd (\bm z') - \bm d, \vd (\bm z')}$ as an approximate supporting hyperplane, i.e. $\ip{\vd (\bm z') - \bm d, \bm y} \geq \ip{\vd (\bm z') - \bm d, \vd (\bm z')} - \delta$ for all $\bm y \in \mathcal D + \mathbb{R}^n_{\geq 0}$ for a sufficiently small $\delta$. 

With this, we introduce the notion of $(\lambda, \gamma, \delta)$-KKT points. 

\begin{definition}[$(\lambda, \gamma, \delta)$-Approximate KKT]\label{def:apx-KKT}
 We say $(\bm a, \bm d, \bm x)$, i.e., a point $\bm d $ along with the normal direction $\bm a$ and a pre-image $\bm x$  is a $(\lambda, \gamma,\delta)$-KKT point with $\lambda \geq 1$, $\gamma \geq 1$, and $\delta >0$, for the minimization problem on $\mathcal D+\mathbb R^n_{\geq 0}$ if 
		\begin{enumerate}
		\item $x_{ij} \geq 0$ for all $i \in [n]$ and $j \in [m]$, and $\lambda^{-1} \leq \sum_{i \in [n]}  x_{ij} \leq \lambda$ for all $j \in [m]$,
		\item $\bm d = \vd(\bm x)$ and $\gamma_i^{-1} \leq a_i \cdot d_i \leq \gamma_i$ for all $i \in [n]$, and  
		\item and $\mathcal D + \mathbb{R}^n_{\geq 0} \subseteq \{\bm y\in \mathbb R^n | \ip{\bm a, \bm y}\geq \ip{\bm a,\bm d}-\delta = n -\delta \}$.
	\end{enumerate}
 				
 Informally, all chores are almost fully allocated, each entry of $\bm a$ is a $\gamma$-approximation of $\bm 1/\bm d$, the gradient of $\mathcal L$, and $\bm a$ is a $\delta$-approximately-supporting hyperplane for $\mathcal D +\mathbb R^n_{\geq 0}$.
\end{definition}

In Section~\ref{gen:apprx-KKT-apprx-CEEI}, we show that a $(\lambda, \gamma, \delta)$-KKT point with where $\lambda = 1 + \varepsilon/2^{\textup{poly}(n,m)}$, $\gamma = 1 + \varepsilon$ and $\delta = \varepsilon/2^{\textup{poly}(n,m)}$ can be mapped to a $\varepsilon^{1/ 6}$-CEEI. The proof emulates the proof in~\cite{BogomolnaiaMSY17}, and consequently matches the proof in the linear case.

However, some subtle problems 
arise when generalizing the algorithm from the linear case to determine a $(\lambda, \gamma, \delta)$-KKT point. 
Firstly, the convergence of the entire algorithm relies crucially on the fact that the potential $\mathcal L(\bm d)$ never decreases at any point. For this, we require that we have $\vd (\bm z')$ Pareto-dominate $\bm d$. We can ensure this by first computing an arbitrary approximate nearest point $\bm z''$ and then increase the consumption of certain chores in $\bm z''$ to get $\bm z'$ such that $\vd (\bm z')$ Pareto-dominates $\bm d$. Since we know that $\vd (\bm z^*)$ Pareto-dominates $\bm d$, and $|| \vd (\bm z'') - \vd (\bm z^*)||_2$ is small, the increase in consumption of the chores will also be small (Observations~\ref{pareto-domination} and~\ref{appx-nearest-point-dist}).

Secondly, the hyperplane $\ip{\vd (\bm z') - \bm d, \bm y} = \ip{\vd (\bm z') - \bm d, \vd (\bm z')}$ can be a good approximation of the hyperplane $\ip{\vd (\bm z^*) - \bm d, \bm y} = \ip{\vd (\bm z^*) - \bm d, \vd (\bm z^*)}$ (or equivalently $\delta$ is inverse-exponentially small) only if $||\bm d - \vd(\bm z^*)||_2$ is significantly larger than $|| \vd (\bm z^*) - \vd (\bm z')||_2$. Therefore, if at any point in our algorithm, we have $|| \bm d - \vd(\bm z')||_2 \leq M\varepsilon$ for a sufficiently large $M$, where $\varepsilon \geq ||\vd (\bm z') - \vd (\bm z^*) ||_2$,  then we stop and return a pre-image of $\bm d$ (note that as the disutility functions are $1$-homogeneous, this can be done by appropriately scaling the consumption of chores for each agent).

Finally,
and most importantly,
we need to ensure that the approximate supporting hyperplanes do not introduce point with excessive over-allocation.
Let $\ip {\bm a^*, \bm y} = n$ and $\ip{\bm a', \bm y} = n$ represent the hyperplanes $\ip{\vd (\bm z^*) - \bm d, \bm y} = \ip{\vd (\bm z^*) - \bm d, \vd (\bm z^*)}$ and $\ip{\vd (\bm z') - \bm d, \bm y} = \ip{\vd (\bm z') - \bm d, \vd (\bm z')}$ respectively after appropriate scaling, i.e., $\bm a^{\ell} = \big(n/ \ip{\vd (\bm z^{\ell}) - \bm d, \vd (\bm z^{\ell})} \big) \cdot (\vd (\bm z^{\ell}) - \bm d)$ for $\ell \in \{*,'\}$. Since we are dealing with approximate supporting hyperplane\footnote{$\ip{\bm a', \bm y} \geq n - \delta'$ for all $\bm y \in \mathcal D + \mathbb{R}^n_{\geq 0}$, where $\delta' = \frac{\delta \ip{\vd (\bm z^{'}) - \bm d, \vd (\bm z^{'})}}{n}$.}, the point maximizing $\mathcal L$, say $\bm d'$ on $\ip{\bm a', \bm y} = n$, maybe contained in the strict interior of $\mathcal D + \mathbb{R}^n_{\geq 0}$. Also note that in this case, the nearest point in $\mathcal D + \mathbb{R}^n_{\geq 0}$ to $\bm d'$  is $\bm d'$ itself, and therefore the distance between $\bm d'$ and its approximate nearest point in $\mathcal D + \mathbb{R}^n_{\geq 0}$ is significantly smaller than $M \varepsilon$  and our algorithm will return the point $\bm d'$, the normal to the hyperplane  $\bm a'$ and its pre-image, say $\bm x'$ in the very next iteration. Now note that while conditions (2) and (3) in Definition~\ref{def:apx-KKT} are satisfied, condition (1) may not be satisfied. In particular, there could be chores that are significantly over-allocated! 
At first this may seem to be counter-intuitive as the hyperplane $\ip{\bm a', \bm y} = n$ is a good approximation of the exact supporting hyperplane $\ip{\bm a^*, \bm y} = n$, and, the point $\bm d^*$ that maximizes $\mathcal L$ on  $\ip{\bm a^*, \bm y} = n$ lies outside $\mathcal D + \mathbb{R}^n_{\geq 0}$ and as a result no chores are over allocated in a pre-image of $\bm d^*$. However, we show that the disutility profiles of the point $\bm d^*$ maximizing $\mathcal L$ on the hyperplane $\ip{\bm a^*, \bm y} = n$ and the point $\bm d'$ maximizing $\mathcal L$ on the hyperplane $\ip{\bm a', \bm y} = n$ can be very far apart even if $||\bm a' - \bm a^*||_2$ is small\footnote{In fact $||\bm a' - \bm a^*||_2$ will be small as $|| \bm d' - \bm d^*||_2$ is significantly small.}. This is primarily due to the fact that $\bm d' = \big( \tfrac{1}{a'_1}, \tfrac{1}{a'_2}, \dots , \tfrac{1}{a'_n} \big)$ and $\bm d^* = \big(\tfrac{1}{a^*_1}, \tfrac{1}{a^*_2}, \dots , \tfrac{1}{a^*_n} \big)$,  and even though $|a'_i - a^*_i| \leq \varepsilon$ for all $i \in [n]$, $1/a'_i$ and $1/a^*_i$ can be very far apart. 
We circumvent this issue by showing that if there are some chores that are significantly over-allocated in $\bm x'$, then we can find an allocation $\bm x''$ from $\bm x'$ by reducing consumption of the over-allocated chores and re-allocating some of the not-over-allocated chores such that $\vd (\bm x'') \in \mathcal D + \mathbb{R}^n_{\geq 0}$ and $\ip{\bm a, \vd(\bm x'')} < n -\delta'$, which is a contradiction to the fact that $\ip{\bm a, \bm y} = n$ is an approximate supporting hyperplane to $\mathcal D + \mathbb{R}^n_{\geq 0}$. This is where the bulk of the error analysis is required (summarized in Lemmas~\ref{appx-KKT-point-2-lambda-under} and~\ref{appx-KKT-point-2-lambda-over}). 

We now outline the entire procedure.

\paragraph{Putting it Together.}
Similar to the case with linear disutilities, the algorithm is iterative. In each iteration $k \geq 0$, 

\begin{enumerate}\setcounter{enumi}{-1}
	\item\label{gen:st1} $\bm d^k$ is the infeasible point ``lying below'' $\mathcal D$ at the start of round $k$.
	$\bm d^0$ is any infeasible point.
	\item\label{gen:st2} Find $\bm x^k_+$ such that $\vd (\bm x^k_+)$ is an $\varepsilon$-approximate nearest feasible point to $\bm d^k$ in $\mathcal D  +\mathbb{R}^n_{\geq 0}$, s.t. $(\bm d^k_+)_i \geq (\bm d^k)_i$ for all $i \in [n]$ and then \emph{round up} $\bm d^k_+$ to the nearest rational point with polynomial bit size.
	\item \label{gen:st3} If $||\bm d^k_+ -\bm d^k||_2 \leq M \cdot \varepsilon$, then return $(\bm a^{k-1}, \bm d^k, \bm x^k)$ where $\bm x^k$ is a pre-image of $\bm d^k$ obtained by rescaling $\bm x^k_+$ appropriately, i.e., $(\bm x^k)_i \gets (\bm x^k_+)_i \cdot  \frac{(\bm d^k)_i}{ (\bm d^k_+)_i}$ for all $i \in [n]$.
	\item\label{gen:st6} Set $\bm a^k\propto\bm d^k_+-\bm d_k$, rescaled so that $\ip{\bm a^k,\bm d^k_+}=n$.
	\item\label{gen:st7} Define $\bm d^{k+1}$ to be $(1/a_1,\,\dotsc,\,1/a_n)$, the maximizer of $\mathcal L(\bm d)$ subject to $\ip{\bm a^k,\bm d}=n$.
	\item\label{gen:st8} Return $(\bm a^k, \bm d^k_+, \bm x^k_+)$ if $\bm d^{k+1}$ is ``close enough'' to $\bm d^k_*$, otherwise repeat.
\end{enumerate}

The algorithm has polynomially many iterations, since similar to the case when agents have linear disutilities, if it does not terminate in iteration $k$, then the potential $\mathcal L$ increases by at least $\Omega(\varepsilon^2/n^2)$. And $\mathcal L$ is upper bounded. By arguing that every iteration can be done in polynomial time in Section \ref{sec:gen-alg-bounds}, we get an FPTAS in Theorem \ref{gen:mainthm}. 

\subsection{Organization}
We give a brief road map of the rest of the paper. In what follows, we first discuss some related work on CE in Section~\ref{sec:relWork} and state some fundamental results from~\cite{BogomolnaiaMSY17} that we use crucially for our algorithm design in Section~\ref{sec:prelim} . Thereafter, we present the FPTAS when agents have linear disutilities in Section~\ref{sec:CE} so that the reader gets a good idea of the meta-level algorithm. Finally, in Section~\ref{gen:CE} we discuss the FPTAS when agents have general 1-homogeneous disutilities. In section~\ref{sec:extensions}, we discuss the extensions of our results to the setting 
when the items to be divided contain both goods and bads (\emph{mixed manna}) with linear valuations, and when
agents have unequal income needs (\emph{CE in Fisher model}).

\section{Related Work}\label{sec:relWork}
Competitive equilibrium (CE) has been a fundamental concept in several economic models since the time of L{\'e}on Walras~\cite{Walras74} in the 19th century. In this paper, we primarily focus on CEEI, which is a special case of CE in Fisher markets, which again is a special case of CE in exchange markets (also referred to as Arrow-Debreu markets). The existence of CE under some mild assumption was proved in the exchange setting by Arrow and Debreu~\cite{AD} and independently by Mackenzie~\cite{Mckenzie54, Mckenzie59}. However, the proofs of existence used fixed point theorems and were non-constructive. In the last few decades, there has been substantial contribution from the computer science community in coming up with constructive algorithms to determine a CE. As mentioned in the introduction, there has been a long line of convex programs, interior point and  combinatorial polynomial time algorithms for determining CE with goods in both Fisher and the exchange setting~\cite{ColeDGJMVY17, DevanurGV16, nenakov83,  DevanurPSV08, Orlin10, Vegh12, DuanM15, DuanGM16, GargV19, CheungCD13}. There are also hardness results known when agents have more general utility functions~\cite{chen2017complexity, ChenDDT09, CThard, Rubinstein18}. The existence and computational complexity of CE and its relaxations have been studied in discrete settings (with indivisible objects) as well~\cite{FeldmanGL16}. 

The study of CE with chores/ bads has not received similar extensive investigation. One plausible reason could be that this does not capture a natural market and  such a setting is interesting only from a fair division perspective. Nevertheless, the CE with bads exhibits far less structure than the CE with goods as explained in the introduction. There are polynomial time enumerative algorithms known only when there are constant number of agents or chores~\cite{BranzeiS19, GargM20}. Quite recently, ~\cite{ChaudhuryGMM21} gave an LCP formulation for determining CEEI with \emph{mixed manna} (goods and bads) when the utility functions are separable piecewise-linear and concave (SPLC) which includes linear.

	\section{Preliminaries}
\label{sec:prelim}
Recall the chore division problem formalized in Section~\ref{sec:model} above:
We seek to divide $m$ divisible chores among $n$ agents with convex, 1-homogeneous disutility functions $\d_1,\,\dotsc,\,\d_n$, through the mechanism of {\em competitive equilibrium with equal income (CEEI)}. In this section we state a characterization of CEEI and certain properties of the disutility space that are crucial for our results.

In the case of dividing goods, the seminal work of Eisenberg and Gale~\cite{EG} shows that any allocation that maximizes the Nash welfare --- or equivalently the geometric mean of the utilities --- is at a CEEI. 
Since the Nash welfare maximization is a convex program, an approximate CEEI can be determined by an ellipsoid algorithm. 
Unfortunately, in the case of dividing bads, the set of equilibria could be non-convex and therefore one cannot hope for convex program formulation that captures equilibria~\cite{BogomolnaiaMSY17}. 
However, a recent result by  Bogomolnaia et al.~\cite{BogomolnaiaMSY17} show a similar, but non-convex formulation for an exact CEEI (Definition~\ref{def:CE}, with $\varepsilon=0$) with chores. 
In particular,~\cite{BogomolnaiaMSY17} show that the conditions of an exact CE hold if and only if the disutility profile is a critical point for the Nash welfare on the boundary of the feasible region. Formally:

\begin{theorem}[\cite{BogomolnaiaMSY17}]
	\label{bogomolnaiathm}
	Let $\mathcal F$ and $\mathcal D$ be the feasible space of allocations and disutility profiles as defined in~\eqref{def:F}. 
	For some $\bm d\in \mathbb R^n$, denote the Nash social welfare as $\NSW(\bm d):=\prod_{i=1}^n d_i$.
	Then $\bm d$ can be achieved by a CEEI if and only if the following conditions all hold: a) $\bm d\in \mathcal D$, b) $\NSW(\bm d)>0$, and c) $\bm d$ satisfies the KKT conditions for the problem of minimizing $\NSW$ on $\mathcal D$.
	Equivalently, $\bm d$ is on the lower-boundary of $\mathcal D$, but not on the boundary of $\mathbb R^n_{\geq 0}$, and the gradient $\nabla \NSW(\bm  d)$ is parallel to some supporting hyperplane normal for $\mathcal D$ at the point $\bm d$.
\end{theorem}

	Note that when dis-utilities are linear functions, $\mathcal D$ is a linear polytope, though it need not have an efficient representation. 
	When dis-utilities are general, 1-homogeneous, convex functions, the set $\mathcal D$ need not be convex. 
	However, we next show that $\mathcal D+\mathbb R_{\geq 0}^n$ is convex, and we will therefore use it as our feasible region in the analysis; see Appendix~\ref{app:tech-results} for the proof.

\def\ClaimDPlusIsConvex{$\mathcal D+\mathbb R_{\geq 0}^n$ is convex, when the disutility functions $\d_1,\,\dotsc,\,\d_n$ are convex.}

\begin{claim}\ClaimDPlusIsConvex
 \label{claim:convex}
 \end{claim}

In the following sections, we extend Theorem~\ref{bogomolnaiathm} to map approximate KKT points to approximate CEEI (Definition \ref{def:CE}), and then design an algorithm to find an approximate KKT point.

	\section{Polynomial-Time Algorithm for \texorpdfstring{\eCEEI}{e-CEEI} under Linear Disutilities}\label{sec:CE}
In this section we present an algorithm to find an \eCEEI in time polynomial in $\frac{1}{\varepsilon}$ and the size of the input instance, when agents have linear disutility functions. Recall that, the linear function of agent $i$ is represented by $\d_i(\bm x_i)=\sum_{j=1}^m D_{ij} x_{ij}$, or equivalently $\d_i(\bm x_i)=\ip{ \bm \d_i, \bm x_i }$ where $\bm \d_i = (\d_{i1}, \d_{i2}, \dots , \d_{im})$.

Our algorithm will ensure a stronger notion of approximation where all the chores are exactly allocated, {\em i.e.,} condition $3$ in Definition \ref{def:CE} is satisfied exactly. For this, the algorithm finds a $\gamma$-KKT point as defined in Definition \ref{def:apx-KKT-lin}. 
Let us first discuss how such a KKT point gives a stronger approximate CEEI in the next section, thereby extending Theorem~\ref{bogomolnaiathm}.

 \subsection{Approximate KKT Suffices to get Approximate CEEI}\label{definition-of-vector-division}
	We begin with some notation: as we often use element-wise inverse of a vector, 
	for any two $n$-dimensional vectors $\bm x=(x_1,\dotsc,x_n)$ and $\bm y=(y_1,\dotsc, y_n)$,
	we denote 
	\[\bm x/\bm y := (x_1/y_1,\dotsc,x_n/y_n)\ .\]	
	
	Recall that we are interested in finding local minima for the logarithm of the Nash social welfare	
	\begin{equation}
		\mathcal L(\bm d):=\log(\NSW(\bm d)) = \sum_{i=1}^n \log(d_i)\ .
	\end{equation}
	Observe that $\nabla \mathcal L (\bm d) = \bm 1/\bm d$. 
	From Definition \ref{def:apx-KKT-lin}, recall the $\gamma$-KKT point, $\gamma \ge 1$, for minimizing $\mathcal L$ on $\mathcal D$: point $\bm d$ on the boundary of $(\mathcal D+\mathbb R^n_{\geq 0})$, such that it has $\{\bm y\ |\ \bm a^\top \cdot \bm y\geq n\}$ as a supporting hyperplane for $\mathcal D+\mathbb R^n_{\geq 0}$, where $\bm a\in \mathbb R^n$ approximates $\nabla \mathcal L (\bm d)$ coordinate-wise, i.e., $\forall i, \gamma^{-1}\le \frac{a_i}{1/d_i} \le \gamma$. 
	
	We emulate here the proof of Bogomolnaia et al.~\cite{BogomolnaiaMSY17} to show that approximate KKT points give approximate CEEI.

	As stated in the overview, we wish to show the following.
	\begin{theorem}\label{corr:approx}
	\linearTheoremKKTGivesCE
	\end{theorem}
	\begin{proof}
Let $\gamma=(1+\epsilon)$, then it suffices to show that $\gamma$-KKT gives $(1-\gamma^{-2})$-CEEI since $2\epsilon > (1-\gamma^{-2})$ for $\epsilon>0$. 	Recall we have defined $\overrightarrow{\d}(\bm x):=\left(\d_1(\bm x_1),\,\dotsc,\,\d_n(\bm x_n)\right)$, and sets $\mathcal F$ and $\mathcal D$ are as in~\eqref{def:F}, 	namely, the set of feasible allocations and the set of feasible disutility profiles, in general.

	\paragraph{Defining and Computing the Allocation and Prices.}
	Let $\bm d$ be the disutility profile of the approximate KKT point.
	Since $\mathcal D+\mathbb R^n_{\geq 0}\subseteq \{\bm y\ |\ \bm a^\top \bm y\geq \ip{\bm a,\bm d}\}$ and the entries of $\bm a$ are positive, then $\bm d\in \mathcal D$, by minimality.
	Now, consider any allocation $\bm z$ in $\mathcal F$, such that $\overrightarrow{\d}(\bm z)=\bm d$.
	
	For the second part of the statement of the theorem, we must show that $\bm z$ can be computed, 
	as this will be the allocation of the approximate CEEI.
	In fact, it suffices to find an allocation vector $\bm x$ which simultaneously satisfies the non-negativity constraints of $\mathcal F$, and the linear equality constraints of $\mathcal F$ along with $\overrightarrow{\d}(\bm x) = \bm d$. 
	This can be solved by linear programming techniques in polynomial time. 
		
We wish now to compute the prices at the allocation, for which we will need separating hyperplanes.
To this end, define the set $S_\lambda:=\{\bm x\in \mathbb R^{nm}\ |\ \langle{\bm a,\overrightarrow{\d}(\bm x)}\rangle\leq \lambda\}$. As the disutility functions are convex and continuous, we can conclude that the set  $S_{\lambda}$ is closed, convex, and non-empty for all $\lambda>0$, since $S_\lambda\ni \bm 0$.
When disutilities are linear, $S_{\lambda}$ is in fact a closed half-space, since 
\begin{align*}
	\ip{\bm a,\overrightarrow{\d}(\bm x)}\leq n &\iff \textstyle\sum_{i=1}^n a_i\sum_{j=1}^m \d_{ij} x_{ij} \leq n\ .
\end{align*}

Now, because $\ip{\bm a, \bm y} \geq \langle{ \bm a, \overrightarrow{\d}( \bm z)}\rangle$ for all  $y \in \mathcal{D}$, we can conclude that the $S_{\lambda}$ does not intersect $\mathcal F$ for any $\lambda<\langle{ \bm a,\overrightarrow{\d}(\bm z)}\rangle$. 
Denote $S^* := S_{\langle{\bm a,\overrightarrow{\d}(\bm z)}\rangle}$. The set $S^*$ must be only tangent to $\mathcal F$, since the $\d_i$'s are continuous, but $\bm z\in \mathcal F\cap S^*$. See Figure~\ref{DandFspace} for an illustration.
	Thus, there exists a half-space $H_{\bm c}:=\{\bm x\ |\ \ip{\bm c, \bm x} \geq b\}$ which separates the two sets, \textit{i.e.}\ $\mathcal F\subseteq H_{\bm c}$, and $S^*\subseteq \operatorname{cl}(H_{\bm c}^\complement)$. Also, note that we must have $\ip{ \bm c,\bm z}=b$. 
	Note that when disutilities are linear, we have $c_{ij} = a_i\d_{ij}$, and $b=n$ as the hyperplane separating $\mathcal F$ and $S^*$ is $\ip{a, \vd (\bm x)} = n$.
	
	Finally, we can define the prices at the allocation. 
	Let $p_j:=\min_i c_{ij}$, and let $\bm p:=(p_1,\dotsc,p_m)$. See  Figure~\ref{DandFspace} for an illustration of the supporting hyperplanes $\ip{\bm a, \bm y} = \langle{\bm a, \overrightarrow{\d}( \bm z)}\rangle$ in $\mathcal{D}$ and $\ip{\bm c, \bm x} = \ip{\bm c, \bm z}$ in $\mathcal{F}$.
	
	\begin{figure}[t]
		\centering 
		\begin{tikzpicture}
\draw (0,5)--(0,-2);
\draw  (-2,0)--(5,0);

\draw[fill=blue!5] (0,4)--(1,2)--(4,0)--(2,3)--(0,4);
\node at (1.5,2.5) {$\mathcal{D}$};

\draw[red, thick] (-2,5)--(5,-2) node[pos=0.5,sloped,below] {$\scriptstyle{\ip{\bm a,\bm y} = \ip{\bm a,\vd( \bm z)}}$};

\draw [white, thick, fill=red!5] plot [smooth] coordinates { (-1.3+8,0) (0+8,0.97) (1+8,1.5) (2+8,1.5) (3+8,0) (3.3+8,-1.5) (-1.3+8,-1.5) (-1.3+8,0)};
\draw (0+8,7)--(0+8,-2);
\draw  (-2+8,0)--(6+8,0);

\draw[fill=green!5] (0+8,6)--(0.5+8,3)--(2+8,1.5)--(6+8,0)--(4+8,5.5)--(0+8,6);
\node at (2.5+8,3.5) {$\mathcal{F}$};
\node at (1.3+8,0.5) {$S^* = S_{\ip{\bm a, \vd( \bm z)}}$};

\draw [red, thick] plot [smooth] coordinates { (-1.3+8,0) (0+8,0.97) (1+8,1.5) (2+8,1.5) (3+8,0) (3.6+8,-1.5)};
\draw[red!5] (2.9+8,0)--(3.6+8,-1.5) node[pos=0.75,sloped,below] {\textcolor{red}{$\scriptstyle{\ip{\bm a, \bm d( \bm x)} = \ip{\bm a, \vd( \bm z)}}$}};
\draw[blue, thick] (-1+8,3.75)--(5.333+8,-1) node[pos=0.85,sloped,below] {$\scriptstyle{\ip{\bm c, \bm x} = \ip{\bm c, \bm z}}$};

\filldraw (1,2) circle (2pt);
\filldraw (2+8,1.5) circle (2pt);

\node at (1.25,2.15) {$\scriptstyle{\vd( \bm z)}$};
\node at (2.15 + 8,1.65) {$\scriptstyle{\bm z}$};

\end{tikzpicture}
		\caption{
		\small \em Illustration of the supporting hyperplanes: 
		$\bm z \in \mathcal F$ is a point such that $(\protect \vd (\bm z), \bm a)$ satisfies the approximate KKT conditions in Definition~\ref{def:apx-KKT}. 
		Thus, we have a supporting hyperplane  $\ip{\bm a, \bm y} = \ip{\bm a, \protect \vd ( \bm z)} = n$ of $\mathcal{D}$
		such that $\gamma^{-1} \leq a_i \cdot \d_i(z_i) \leq \gamma$ (left). 
		The figure on the right describes the set $S^* = S_{\ip{\bm a, \protect \vd( \bm z)}}$ 
		and the hyperplane $\ip{\bm c, \bm x} = \ip{\bm c, \bm z}$ that separates $\mathcal{F}$ from $S^*$. 
		Note that $\bm z \in \mathcal F \cap S^*$ and the curve $\ip{ \bm a, \protect \vd (\bm x)} = \ip{\bm a, \protect \vd( \bm z) }$ coincides with the hyperplane $\ip{\bm c, \bm x} = \ip{\bm c , \bm z}$ when the disutility functions are linear.} 
		\label{DandFspace}
	\end{figure}
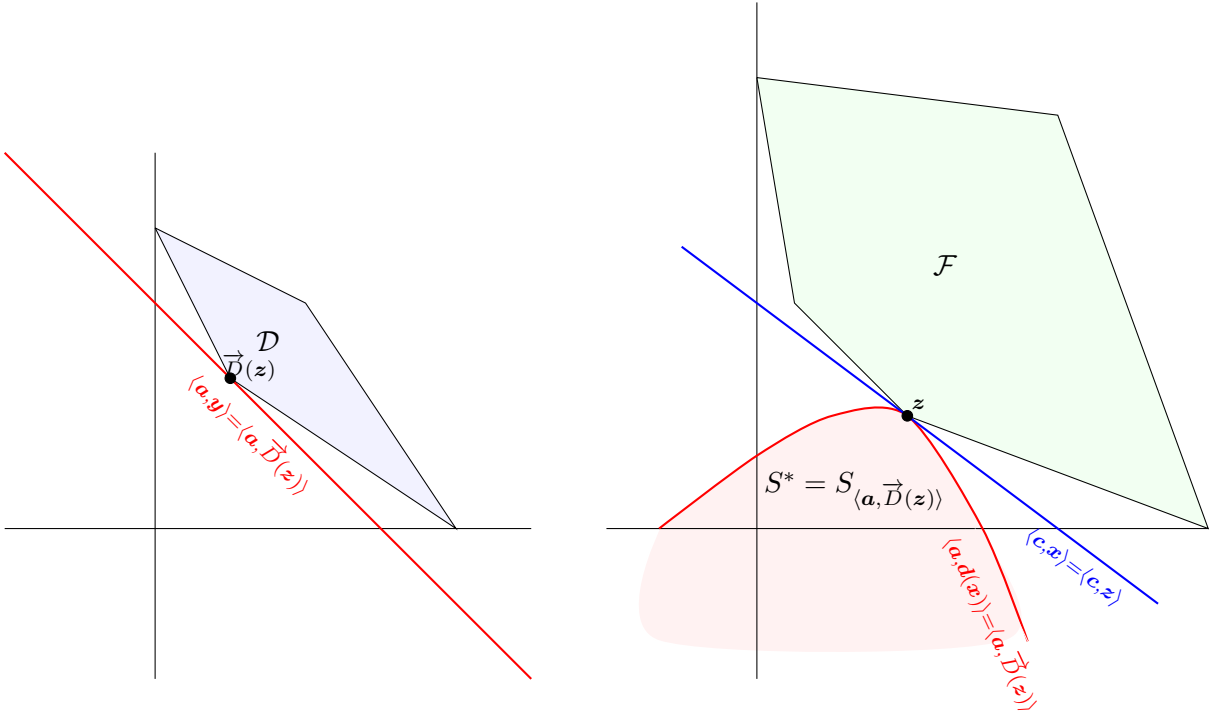

	It remains then to show that the allocation $\bm z$ and the price vector $\bm p$ satisfy the conditions in Definition~\ref{def:CE} where the last two are satisfied without any error, since we have argued already that they can be computed efficiently.

	\paragraph{Satisfying Condition~\ref{P1} in Definition~\ref{def:CE}.}
	  We want to show that for all agents $i$ and $i'$, we have $\gamma^{-2} \cdot  \langle \bm z_i, \bm p \rangle \leq \langle \bm z_{i'}, \bm p \rangle$. But first we make some simple but crucial observations about the price vector $\bm p$.
	  
	  \begin{claim}
	  	\label{technical1}
	  	We have $\sum_{j \in [m]} p_j = \ip{\bm c, \bm z} = b$.
	  \end{claim}
	   
	   \begin{proof}
	   	$\ip{\bm c, \bm x} \geq \ip{\bm c, \bm z} = b$ for all $\bm x \in \mathcal F$ by definition. Also, since $\bm z \in \mathcal F$, we can claim that $b = \min_{\bm x \in \mathcal F} \ip{\bm c, \bm x}$. Observe that $\min_{x \in \mathcal F} \ip{\bm c, \bm x}$ is obtained by assigning each chore fully to the agent that has the smallest $c_{ij}$ value for it. Therefore, we have that $\min_{\bm x \in \mathcal F} \ip{\bm c, \bm x} = \sum_{j \in [m]} \min_{i \in [n]} c_{ij} = \sum_{j \in [m]} p_j$ (by the definition of $p_j$).
	   \end{proof}
   
      Now, consider the half-space $H_{\bm p} = \{\bm x \in \mathbb{R}^{nm}_{\geq 0} \mid \sum_{i\in [n], j \in [m]} p_j \cdot x_{ij} \geq b  \}$. We first observe that this half-space is entirely contained in $H_{\bm c}$.
      
      \begin{claim}
      	\label{technical2}
      	We have $H_{\bm p} \subseteq H_{\bm c}$.
      \end{claim}
  
       \begin{proof}
       	Consider any point $\bm x \in H_{\bm p}$. We have $b \leq \sum_{i \in [n]} \sum_{j \in [m]} x_{ij} \cdot p_j$. Since $p_j \leq c_{ij}$ for all $i \in [n]$, we have that $\sum_{i \in [n]} \sum_{j \in [m]} x_{ij} p_j \leq \sum_{i \in [n]} \sum_{j \in [m]} x_{ij} \cdot c_{ij}$, implying that $\sum_{i \in [n]} \sum_{j \in [m]} x_{ij} \cdot c_{ij} \geq b$, i.e., $\ip{ \bm c, \bm x } \geq b$. Therefore $\bm x \in H_{\bm c}$. 
       \end{proof}
   
       Finally, note that every point $\bm x \in \mathcal F$ is also contained in $H_{\bm p}$.
       
       \begin{claim}
       	\label{technical3}
       	Consider any $\bm x \in \mathcal F$. Then $\bm x \in H_{\bm p}$.
       \end{claim}
   
        \begin{proof}
        	Consider any $\bm x \in \mathcal F$. We have 
        	\begin{align*}
        	 \sum_{i \in [n],j \in [m]} x_{ij} \cdot p_j &= \sum_{j \in [m]} p_j \cdot \sum_{i \in [n]}x_{ij}\\
        	                                             &= \sum_{j \in [m]} p_j &(\text{$\sum_{i \in [n]} x_{ij} = 1$ as $\bm x \in \mathcal F$})\\
        	                                             &= b &(\text{by Claim~\ref{technical1}})     	                                              
        	\end{align*}
        	Therefore $\bm x \in H_{\bm p}$.
        \end{proof}
	  Now, we are ready to show that $\gamma^{-2} \cdot  \langle \bm z_i, \bm p \rangle \leq \langle \bm z_{i'}, \bm p \rangle$. 
	  Assume otherwise and say we have  $\gamma^{-2} \cdot  \langle \bm z_i, \bm p \rangle > \langle \bm z_{i'}, \bm p \rangle$. 
	  Then we could replace the allocation as follows: Construct $\hat{\bm z}$ by setting $\hat{\bm z}_{i'}=\tfrac12\bm z_{i'}$, and $\hat{\bm z}_{i}=\left(1+\tfrac{\ip{\bm z_{i'}, \bm p}}{2\ip{\bm z_{i}, \bm p}}\right)\bm z_{i}$. 
	  Since $\bm z \in \mathcal F$, we have $\sum_{i \in [n],j \in [m]} p_j z_{ij} = b$. Also note that 
	  \[
	  	b = \textstyle\sum_{i \in [n],j \in [m]}p_jz_{ij}=\sum_{i \in [n],j \in [m]}p_j\hat z_{ij}\ ,\]
	  since the payment subtracted from agent $i'$ is equal to the payment added to agent $i$ and so $\hat{\bm z}\in H_p$.

	  By Claim~\ref{technical2}, we have that  $ \hat{\bm z} \in H_{\bm c}$. 
	  Recall that $H_{\bm c}$ is a separating half-space between $S^*$ and $\mathcal{F}$, \textit{i.e.}, $\mathcal{F} \subseteq H_{\bm c}$ and $S^* \subseteq \operatorname{cl}(H_{\bm c}^{\complement})$, 
	  implying that for every point $\bm x \in H_{\bm c}$ we have 
	  $\langle{\bm a, \overrightarrow{\d}(\bm x)}\rangle \geq \langle{\bm a, \overrightarrow{\d}( \bm z)}\rangle$.  Since $\hat{\bm {z}} \in H_{\bm c}$, we have  $\langle{\bm a,\overrightarrow{\d}(\hat{\bm z})}\rangle\geq \langle{\bm a,\overrightarrow{\d}(\bm z)}\rangle$. However,
		\begin{align*}
			\ip{\bm a,\overrightarrow{\d}(\hat{\bm z})} - \ip{\bm a,\overrightarrow{\d}(\bm z)}
			& =-\frac12 a_{i'} \d_{i'}(\bm z_{i'}) + \frac{\ip{\bm z_{i'}, \bm p}}{2\ip{\bm z_{i}, \bm p}} a_{i}\d_{i}(\bm z_{i}) \\
			& \leq -\frac 12 \gamma^{-1} + \frac{\ip{\bm z_{i'}, \bm p}}{2\ip{\bm z_{i}, \bm p}}\cdot \gamma\\
			& <-\tfrac 12 \gamma^{-1}+\tfrac 12 \gamma^{-1}=0\ , 
			&\text{(as ${\ip{\bm z_{i'}, \bm p}}/{\ip{\bm z_{i}, \bm p}} < \gamma^{-2}$)}
		\end{align*}
		which is a contradiction.
	  The first inequality is due to the definition of $\gamma$-approximate KKT, which dictates that $\gamma^{-1} \leq a_{\ell} \cdot \d_{\ell}( \bm z_{\ell}) \leq \gamma$ for all $\ell \in [n]$.

	\paragraph{Satisfying Condition~\ref{P2} in Definition~\ref{def:CE} Exactly (i.e., Condition~\ref{EP2}).} 
	We want to show that for all $i \in [n]$, we have $\d_i(\bm z_i) \leq \d_i(\bm y)$ for all $y$ such that $\ip{\bm y, \bm p} \geq \ip{\bm z_i, \bm p}$. 
	Let us assume that there exists a $\bm y$ such that $\d_i(\bm z_i) > \d_i(\bm y)$ and $\ip{\bm y, \bm p} \geq \ip{\bm z_i, \bm p}$. 
	We define a new allocation $\bm z' = (\bm z_1, \bm z_2, \dots, \bm z_{i-1}, \bm y, \bm z_{i+1}, \dots ,\bm z_n)$. 
	First note that $\sum_{i \in [n], j\in [m]} z'_{ij} \cdot p_j \geq \sum_{i \in [n], j\in [m]} z_{ij} \cdot p_j = b$ as $\ip{\bm y, \bm p} \geq \ip{\bm z_i, \bm p}$. 
	Therefore $\bm z' \in H_{\bm p}$. 
	By Claim~\ref{technical2}, we have that $\bm z' \in H_{\bm c}$. 
	Recall that $H_{\bm c}$ is a separating half-space between $S^*$ and $\mathcal{F}$, \textit{i.e.}, $\mathcal{F} \subseteq H_{\bm c}$ and $S^* \subseteq \operatorname{cl}(H_{\bm c}^{\complement})$, 
	implying that for every point $\bm x \in H_{\bm c}$ we have $\langle{\bm a, \overrightarrow{\d}(\bm x)}\rangle \geq \langle{\bm a, \overrightarrow{\d}( \bm z)}\rangle$.  
	Since $\bm z' \in H_{\bm c}$, we have  
	$\langle{\bm a,\overrightarrow{\d}(\bm z')}\rangle\geq \langle{\bm a,\overrightarrow{\d}(\bm z)}\rangle$.
	However, since $\d_i( \bm y) < \d_i(\bm z_i)$ and $a_i \geq \gamma^{-1}/ \d_i(\bm z_i) > 0$ (by the definition of approximate KKT point),  we have that  $\langle{\bm a,\overrightarrow{\d}(\bm z')}\rangle <  \langle{\bm a,\overrightarrow{\d} (\bm z)}\rangle$, which is a contradiction.

	\paragraph{Satisfying Condition~\ref{P3} in Definition \ref{def:CE} Exactly (i.e., Condition~\ref{EP3}).}
	 Since $\bm z \in \mathcal F$, we have that $\sum_{i \in [n]}z_{ij}=1$ for all $i \in [n]$. \qedhere
  \end{proof}

	This concludes the proof that an  approximate-CEEI can be determined from approximate-KKT points in polynomial time. In the next subsection, we outline a polynomial time algorithm that determines an approximate-KKT point.

\subsection{Algorithm, and Convergence Guarantees}\label{sec:lin-kkt}
   We show that approximate-KKT points can be found in polynomial time. 
   We begin with an overview of the procedure, and later show how the steps are implemented.
   The idea is to perform an exterior-point procedure outside of the feasible region, which produces a sequence of guesses for approximate KKT points, while increasing along the objective.   
    Due to the nature of the objective function, we alternate between finding supporting hyperplanes, and finding $\NSW$-maximizing points on these hyperplanes, until we find a point whose gradient is approximately in line with the supporting hyperplane. 
    
    To be precise, our algorithm starts from a point $\bm d^{0}$ very close to $\bm 0$. Note that this point lies below  $\mathcal{D}$. Then, we find the nearest point $\bm d^{0}_{*}$ in $\mathcal{D}+\mathbb R_{\geq 0}^n$ to $\bm d^0$.    
    We will address how to find this nearest point, and explain how to robustly handle approximation errors in finding this nearest point. In doing so, it will be helpful to find nearest points in the convex region $\mathcal D+\mathbb R_{\geq 0}^n$, but keeping in mind that the true optimum $\bm d^0_*$ lies in $\mathcal D$: 
    to see this, note that $\bm d^{0}_{*}$ has to lie on the lower envelope of $\mathcal{D}$, and since it is the closest point in $\mathcal{D}$ to $\bm d^{0}$, it follows that $(\bm d^{0}_{*} - \bm d^{0})$ is normal to a supporting hyperplane of $\mathcal{D}$ at $\bm d^{0}_{*}$. 
Furthermore, we show that $\bm d^{0}_{*}$ Pareto-dominates $\bm d^{0}$, thereby implying that the Nash welfare at $\bm d^{0}_{*}$ is larger than the Nash welfare at $\bm d^{0}_{*}$. 

    Let $\ip{\bm a, \bm y} = n$ be the supporting hyperplane of $\mathcal{D}$ at $\bm d^{0}_{*}$, where $\bm a \propto (\bm d^{0}_{*} - \bm d^{0})$.
	Let $\bm d^{1}$ be a point on this hyperplane with maximum Nash welfare. Observe that at $\bm d^{1}$, we should have $\nabla \mathcal L$ proportional to $\bm a$, \textit{i.e.}, $\bm a = \bm 1 / \bm d^{1}$, implying that $\bm d^{1} = \bm 1 / \bm a$. 
	Since $\ip{\bm a, \bm y} = n$ is a supporting hyperplane of $\mathcal D$ at $\bm d^{0}_{*}$ (a point on the lower envelope of $\mathcal D$), we have that $\bm d^{1}$ also lies below the lower envelop of $\mathcal{D}$. 
	We prove that if the distance between $\bm d^{0}_{*}$ and $\bm d^{1}$ is small, then $\bm d^{0}_{*}$ is our approximate KKT-point, otherwise we have a new point $\bm d^1$ below $\mathcal{D}$, which has significantly higher Nash welfare than $\bm d^0$. 
	We run the exact same steps from $\bm d^{1}$. We argue that such a procedure should eventually give us an approximate KKT point as there is significant increase in Nash welfare with every iteration of the algorithm whenever no approximate KKT point is found.  
	The full description of the algorithm  is given in Algorithm~\ref{alg:KKT}.
	
   In what follows, define $\logd(\bm x,\bm y):= \sum_{i} \left|\log( x_i/y_i)\right|$.	Notice that if $\logd(\bm x,\bm y)\leq \varepsilon$, then $(1+\varepsilon)^{-1}\leq x_i/y_i\leq (1+\varepsilon)$ for all $i$, since $\log(1+a)\leq a$ for all $a>-1$. We will find a point which is a $(1+\varepsilon)$-approximate KKT point following Algorithm~\ref{alg:KKT}. 

\begin{algorithm}[h]
\caption{Finding Approximate KKT}
 \begin{algorithmic}[1]
 \State \label{alg:init} Let $\bm d^0$ be any infeasible, strictly positive, disutility profile, near $\bm 0$\;
 \While{\textbf{true}}
	\State Set $\bm d^k_*$ to be the nearest dominating point in $\mathcal D$ to $\bm d^k$, \textit{i.e.} \label{alg:gradflow}
		\[
			\argmin\left\{\Vert \bm y - \bm d^k\Vert_2^2 \ \middle|\  \bm y\in {\mathcal D + \mathbb{R}^n_{\geq 0}},\, \bm y\geq \bm d^k \right\}\;
		\]

	\State \textbf{Set} $\bm a^k \gets (\bm d^k_*-\bm d^k)$, the direction from $\bm d^k$ to $\mathcal D$\;\label{alg:ak1}
	\State \textbf{Rescale} $\bm a^k$ so that $\ip{\bm a^k,\bm d^k_*}= n$\; \label{alg:ak2}
	
	\State  \textbf{Set} $\bm d^{k+1} \gets \bm 1/\bm a^k$\; \label{alg:hyperplane-move}
	
	\If{$\logd(\bm d^{k+1},\bm d^k_*)<\varepsilon$}\label{alg:stopping-hyp-cond}
	  \State \textbf{Return} $(\bm d^{k}_*,\bm a^k)$ \label{alg:stopping-hyp}
	\EndIf
  \EndWhile
\end{algorithmic}
\label{alg:KKT}
\end{algorithm}
	
	\paragraph{Correctness.}
	We begin by proving here that the algorithm truly returns an approximate KKT point and we will later show that $(i)$ it will terminate in polynomially many iterations, $(ii)$ each iteration can be implemented in polynomial time. To this end, we will need the following technical results, about the steps of the algorithm. 
	
	\def\lemmaAboutAlgoStepsLinear{Regardless of the geometry of $\mathcal D$, so long as $\mathcal D+\mathbb R_{\geq 0}^n$ is convex, we have that for each iteration $k\geq 0$ of Algorithm~\ref{alg:KKT}:\begin{enumerate}\itemsep0pt
			\item The hyperplane defined as $\{\bm y\in \mathbb R^n|\ip{\bm a^k,\bm y}\geq \ip{\bm a^k,\bm d^k_*}\}$ is supporting for $\mathcal D+\mathbb R_{\geq 0}^n$, at $\bm d^k_*$. 
			\item If $\bm d^{k}$ has strictly positive entries and does not lie in $\mathcal D+\mathbb R^n_{\geq 0}$, then $\bm d^k_*$, $\bm a^k$, and $\bm d^{k+1}$ have strictly positive entries, and $\bm d^k_*\in \mathcal D$.
		\end{enumerate}}

	\begin{lemma}\label{lem:algo-steps-linear}\lemmaAboutAlgoStepsLinear \end{lemma}
	
	We show these results in Appendix~\ref{app:tech-results}, as the proofs are mostly technical. 
	Informally, these hold due to the geometry of the feasible region, and ensure that each iterate is well-defined, and economically meaningful.	
	To complete the proof of correctness, we
	show that we can efficiently find a starting point $\bm d^0$ which is strictly positive in every entry, and is infeasible. 
	Thus, Lemma~\ref{lem:algo-steps-linear} will inductively show that every point is positive and well-defined.
	
	\begin{claim}\label{claim:d0}
The point $\bm d^0 = \frac{m\delta}{2n} \bm 1$ where $\delta = \min_{ij} \d_{ij}$ is a strictly positive infeasible disutility profile. 
\end{claim}
\begin{proof}
Since $\d_{ij}\ge \delta$ for all $i$ and $j$,
any feasible dis-utility profile must assign disutility at least $m\delta/n$ to some agent. 
Therefore, it is impossible for every agent to have disutility $\tfrac12m\delta/n$ at a feasible point.
\end{proof}
	
	We now show that in the stopping condition, Algorithm~\ref{alg:KKT} returns an approximate KKT point. 
	Intuitively, this holds because the $\logd$ function in the stopping condition is designed to correctly captures the multiplicative error needed in the definition of approximate KKT.

	\begin{lemma}\label{lem:kkt-lin}
		Algorithm~\ref{alg:KKT} returns a $(1+\varepsilon)$-KKT point for minimizing $\mathcal L$ on $\mathcal D$.
	\end{lemma}
	\begin{proof}
		Suppose the algorithm terminates and returns $(\bm d^k_{*}, \bm a^k)$ on line~\ref{alg:stopping-hyp}. Note that we have $\bm d^{k+1} = \bm 1 / \bm a^k$ and $\logd(\bm d^{k+1},\bm d^{k}_{*})<\varepsilon$, implying that $\logd(\bm 1/ \bm a^k, \bm d^{k}_{*}) < \varepsilon$. Then, we have $(1+\varepsilon)^{-1}\leq a^k_i \cdot (d^k_{*})_i \leq 1+\varepsilon$ for all $i$. Also by Lemma~\ref{lem:algo-steps-linear}, we have that $\ip{ \bm a^k, \bm y} = n$ is a supporting hyperplane of $\mathcal{D}$ passing through $\bm d^k_*$. 
		Therefore, the point $\bm d^{k}_{*}$ is a $(1+\varepsilon)$-KKT point as in Definition~\ref{def:apx-KKT-lin}.	
	\end{proof}
	
		In the rest of this section, we will argue that the number of iterations must be polynomial, and that each iteration can be solved in polynomial time, which will allow us to conclude the correctness and efficiency of the algorithm.

 \paragraph{Polynomially Many Iterations.}
	We show that in polynomially many iterations the algorithm finds an approximate KKT point. In particular, we show that (a) the log-NSW $\mathcal L$ is always increasing throughout Algorithm~\ref{alg:KKT}, and (b) it increases additively by $\textup{poly}(n, 1/ \varepsilon)$ every time $\logd (\bm d^{k+1}, \bm d^{k}_{*}) \geq \varepsilon$. 
	Bounding the range of $\mathcal L$ over the course of the iteration will then give our desired bound.

	\begin{claim}\label{claim:inc1}
		Steps~\ref{alg:gradflow}.\ and~\ref{alg:hyperplane-move}.\ always increase $\mathcal L$, 
		the log-product of disutilities.
		Formally, $\mathcal L(\bm d^{k+1})\geq \mathcal L(\bm d^k_*)\geq \mathcal L(\bm d^k)$ for all $k\geq 0$.
	\end{claim}
	\begin{proof}
		By Lemma~\ref{lem:algo-steps-linear}, $\bm d^k_*\geq \bm d^k$, coordinate-wise.
		Thus, since $\mathcal L$ is monotone increasing in each coordinate direction, 
		$\mathcal L(\bm d^k_*)\geq \mathcal L(\bm d^k)$. 
		
		We prove that Step~\ref{alg:hyperplane-move} is an improvement by showing that $\bm d^{k+1}$ is the maximizing point on the hyperplane $\ip{\bm a^k,\bm y} = n$, and therefore $\mathcal L(\bm d^{k+1})\geq \mathcal L(\bm d^k_*)$.

		Since $\mathcal L$ is a concave function, it is maximized on this hyperplane when $\nabla\mathcal L$ is proportional to $\bm a^k$, \textit{i.e.}\ when $a^k_i = c/d_i$ for some $c>0$, for all $i$. Since we need $\ip{\bm a^{k}, \bm d }= n$, it suffices to set $c=1$. Thus, $\bm d^{k+1}$ is the $\mathcal L$-maximizing point on the supporting hyperplane which contains $\bm d^k_*$, and so this move is an $\mathcal L$-improvement.
	\end{proof}
	
	Using the above claims, next we show that $\mathcal L$ increases significantly in each iteration of our algorithm.

	\begin{lemma}\label{lem:inc2}
		If Algorithm~\ref{alg:KKT} does not return at step~\ref{alg:stopping-hyp}, then the logarithm of the Nash social welfare increases by at least $\tfrac 1{16}(\varepsilon/n)^2$, i.e., $\mathcal L(\bm d^{k+1})-\mathcal L(\bm d^k)\ge \frac 1{16}(\varepsilon/n)^2$.
	\end{lemma}
\begin{proof}
		Since $\mathcal L(\bm d^k_*) > \mathcal L(\bm d^k)$ by Claim \ref{claim:inc1}, it suffices to show that if $\logd(\bm d^{k+1},\bm d^k_*)>\varepsilon$, then $\mathcal L(\bm d^{k+1})-\mathcal L(\bm d^k_*)$ is large.
		Let $A=\operatorname{diag}(\bm a^k)$, and note that $\ip{\bm 1, A\bm d} = \ip{\bm a^k,\bm d}$, and furthermore, $A\bm d^{k+1} = \bm 1$. Let $\bm \Delta = A\bm d^k_*-\bm 1$, and notice that \[
		\ip{\bm 1,\bm \Delta} = \ip{\bm 1, A(\bm d^k_* - \bm d^{k+1})} = 0
	\]
	Note that $\bm d^k_* = (\bm 1+\bm \Delta)/\bm a^k$, where we take the quotient componentwise as is defined at the start of Section~\ref{definition-of-vector-division}. With $\bm d^{k+1} = \bm 1/\bm a^k$, this gives $\logd(\bm d^k_*,\bm d^{k+1}) = \logd((\bm 1+\bm \Delta),\bm 1)$.
	Therefore, we know that $\sum_{i=1}^n \left|\log(1+\Delta_i)\right| > \varepsilon$. We also get
	\begin{align*}
		\mathcal L(\bm d^{k+1})-\mathcal L(\bm d^k_*) &= \sum_{i=1}^n \log(1/a^k_i)-\log((1+\Delta_i)/a^k_i)
		= -\sum_{i=1}^n \log(1+\Delta_i)
	\end{align*}
	Define:	
	\[
		F(z) := \begin{cases}
			\tfrac14z^2&\text{ if } -1< z\leq 1\\
			\tfrac12z-\tfrac 14&\text{ if }z \geq 1\\
			+\infty&\text{ otherwise}
		\end{cases}
	\]
	At $z=0$, we have that $-z+F(z) = 0 = \log(1+z)$ and $\tfrac{\mathrm d}{\mathrm dz}(-z+F(z))=-1=\tfrac{\mathrm d}{\mathrm dz}(-\log(1+z))$.
	By comparing derivatives for the other values of $z> -1$, we can show that $-\log(1+z)\geq -z+F(z)$ for all $z$.
	Thus,
	\begin{align*}
		\mathcal L(\bm d^{k+1})-\mathcal L(\bm d^k_*)\ &=\ -\sum_{i=1}^n \log(1+\Delta_i)\ \geq\ \sum_{i=1}^n -\Delta_i + \sum_{i=1}^n F(\Delta_i)\ =\ \sum_{i=1}^n F(\Delta_i)
	\end{align*}
		
	Now, since we have $\sum_{i=1}^n \left|\log(1+\Delta_i)\right| > \varepsilon$, there must be some $i$ such that $|\log(1+\Delta_i)|>\varepsilon/n$. If $\Delta_i>0$, then $\Delta_i\geq \log(1+\Delta_i)\geq \varepsilon/n$.
Conversely, if $\Delta_i<0$, we being by noting that for $|z|<0.5$, we have $-\log(1+z)\leq -z + z^2$ for reasons similar to the above. Thus, we get
\[
	\varepsilon/n < -\log(1+\Delta_i) \leq -\Delta_i+\Delta_i^2
\]
We must have $\Delta_i>-1$, since the argument can't be negative, so we have
$2 |\Delta_i|> \Delta_i^2 - \Delta_i>\varepsilon/n$, or $\Delta_i < -\tfrac{1}{2}\varepsilon/n$.
Noting that $F(z)\geq 0$ for all $z$, we can then conclude
\[
	\mathcal L(\bm d^{k+1})-\mathcal L(\bm d^k_*) \geq \textstyle\sum_i F(\Delta_i) \geq \max_i F(\Delta_i) \geq 
	\tfrac 1{16} \varepsilon^2/n^2
\]
as desired.
	\end{proof}

Finally, to bound the number of iterations Algorithm \ref{alg:KKT} would take we need to bound the log-NSW value at the starting point, namely $\mathcal L(\bm d^0)$, where $\bm d^0:=\bm 1\cdot \tfrac m{2n}\min_{i,j}\d_{ij}$, as in Claim~\ref{claim:d0}.
We show the following.

\begin{lemma}\label{lem:noiter-lin}
Starting at $\bm d^0:=\bm 1\cdot \tfrac m{2n}\min_{i,j}\d_{ij}$,
	Algorithm \ref{alg:KKT} finds a $(1+\varepsilon)$-KKT point in \[
		O\left( \frac{n^3}{\varepsilon^2}\cdot \log\left (\frac{n\cdot \max_{i,j}\d_{ij}}{\min_{i,j}\d_{ij}}\right) \right)
		\] many iterations.
\end{lemma}
\begin{proof}
	If we can bound the range of the log-NSW objective, then the proof follows using Lemmas \ref{lem:kkt-lin} and \ref{lem:inc2}. 
	Let $M$ be such that $\d_i(\bm x_i)\leq M$ for every agent $i$, at every feasible $\bm x\in \mathcal F$. 
	Note that $M\leq m\cdot \max_{i,j}\d_{ij}$. 
	
	Then we have that for any feasible $\bm x$, $\mathcal L(\bm d(\bm x))\leq n\log M$.
	Since each round of the above algorithm that doesn't terminate increases the log-NSW by at least $\tfrac 1{16} (\varepsilon/n)^2$, then the total number of rounds possible is at most
	\begin{align*}
		16 \cdot \frac{n^2}{\varepsilon^2} \cdot (n\log(M)-\mathcal L(\bm d^0))&\leq 
		\frac{16 n^3}{\varepsilon^2}\cdot\left(
			\log(m\cdot \max_{i,j}\d_{ij}) - \log(\tfrac {m}{2n}\min_{i,j}\d_{ij})
		\right)\ ,
	\end{align*}
which gives the desired bound.
\end{proof}

Now that we have shown there are polynomially many iterations in our algorithm, it suffices to show that each iteration can be implemented in polynomial time to establish that Algorithm~\ref{alg:KKT} is indeed polynomial time.	

\paragraph{Implementing Each Iteration in Polynomial Time.}
To show that each iteration can be implemented in polynomial time, it suffices to show that the \emph{nearest neighbour search} (step~\ref{alg:gradflow} in Algoritm~\ref{alg:KKT}) can be implemented in polynomial time.

\begin{lemma}\label{lem:iter-poly}
Each iteration of Algorithm~\ref{alg:KKT} can be computed exactly in time polynomial in $n$, $m$, and the description complexity of the $\d_{ij}$'s.
\end{lemma}
\begin{proof}
Let $\overrightarrow{\d}(\bm x):=\left(\d_1(\bm x_1),\,\dotsc,\,\d_n(\bm x_n)\right)$ as defined previously.
Recall that disutility functions are linear, with $\d_i(\bm x_i) := \sum_{j=1}^m \d_{ij}x_{ij}$.

Let $\widehat{\d}$ be the $n\times nm$ block-diagonal matrix such that $\widehat{\d}\bm x= \overrightarrow{\d}(\bm x)$.
To find the nearest-feasible disutility profiles,
we will find the allocation $\bm x$ which minimizes the following convex quadratic program:
\[
	\min_{\bm x\in \mathcal F}\ \left\Vert \overrightarrow{\d}(\bm x) - \bm d^k\right\Vert_2^2 = 
	\min_{\bm x\in \mathcal F}\  \bm x^\top\left(\widehat{\d}^\top\widehat{\d}\right)\bm x - 2(\bm d^k)^\top\widehat{\d}\bm x + (\bm d^k)^\top \bm d^k\ .
\]
It was shown by Khachiyan et al.~\cite{quadprog} that this program can be solved exactly, with running time polynomial in the description complexity of the system.
Thus, so long as $\widehat{\d}$ and $\bm d^k$ have rational entries with polynomial description complexity (polynomial-sized numerators and denominators), the problem can be solved exactly in polynomial time, and the solution will have small description complexity. 

The matrix $\widehat\d$ consists of the $\d_{ij}$'s and our running time is assumed to depend on their description complexity. 
\end{proof}

\paragraph{Final Result.}
We now have all the ingredients to conclude that an approximate CEEI (Definition \ref{def:CE}) can be computed in polynomial time. 
Lemma~\ref{lem:noiter-lin} bounds the number of iterations as a polynomial in $n$, $1/\epsilon$, and the description complexity of the instance,
Claim~\ref{claim:d0} shows how to find a good starting point,
Lemma~\ref{lem:iter-poly} shows that each iteration can be computed in polynomial time, with the same arguments,
and Theorem~\ref{corr:approx} shows how to compute a \eCEEI in polynomial time given the output of Algorithm~\ref{alg:KKT}.
Thus, we conclude that Algorithm~\ref{alg:KKT} is an FPTAS for finding $3\varepsilon$-CEEI.

\begin{theorem}\label{MAINTHMLINEAR}
	Given linear disutility values $\d_{11},\,\dotsc,\,\d_{nm}$, 
	Algorithm~\ref{alg:KKT}, along with~\ref{corr:approx}, finds an \eCEEI in time polynomial in $n$, $m$, $1/\varepsilon$, $\log(\tfrac{\max \d_{ij}}{\min \d_{ij}})$, and the description complexity of the $\d_{ij}$'s.
\end{theorem}

	\section{1-Homogeneous Disutilities: Computing \texorpdfstring{\eCEEI}{e-CEEI} in Polynomial-Time }
\label{gen:CE}
In this section, we show how to extend the results of the previous section when agents' disutility functions are general 1-homogeneous and convex. Access to the disutility functions are through value oracle. For ease of notation, throughout this section, we refer to the $i^{\mathit{th}}$ coordinate of a disutility vector $\bm d$ as $\bm d_i$ (or equivalently $(\bm d)_i$). Similarly, given an allocation $\bm x$, we refer to agent $i$'s bundle as $\bm x_i$ (or equivalently $(\bm x)_i$)and the amount of chore $j$ allocated to agent $i$ as $\bm x_{ij}$ (or equivalently $(\bm x)_{ij}$).  
We first discuss the two main roadblocks in generalizing the approach in Section~\ref{sec:CE}. 
The convex program $\min_{\bm x\in \mathcal F}\ \Vert \overrightarrow{\d}(\bm x) - \bm d^k\Vert_2^2$ for finding the nearest neighbour is not necessarily convex when agents have general $1$-homogeneous and convex disutilities. 
We design an alternative formulation that returns the nearest feasible point, and is convex. 
However, the domain of the new convex program is not defined by a set of linear inequalities and as such one can only find approximate nearest neighbours, \textit{e.g.}, via interior point methods \cite{conv-opt}. 
In turn, the supporting hyperplanes $\ip{\bm a, \bm y} = n$ in Algorithm~\ref{alg:KKT} are now approximate supporting hyperplanes. 
To allow this extra error, we extend the notion of approximate KKT to that given in Definition~\ref{def:apx-KKT}, namely
\begin{enumerate}
	\item $(\bm z)_{ij} \geq 0$ for all $i \in [n]$, and $j \in [m]$, and $\lambda^{-1} \leq \sum_{i \in [n]} (\bm z)_{ij} \leq \lambda$ for all $j \in [m]$,
	\item $\bm d  = \vd(\bm z)$ and $\ip{\bm a, \bm y} \geq \ip{\bm a, \bm d} - \delta = n -\delta $ for  all $\bm y \in \mathcal{D} + \mathbb{R}^n_{\geq 0}$, and 
	\item for each $i \in [n]$, we have $\gamma^{-1} \leq (\bm a)_i \cdot (\bm d)_i \leq \gamma$. 
\end{enumerate}

In the previous section, with linear disutilities we have $\lambda=1$ and $\delta=0$ in the above definition, and this was crucially used to map approximate KKT to stronger approximate CEEI. We show in Section~\ref{gen:apprx-KKT-apprx-CEEI}  that the claim follows even with $\lambda>1, \delta>0$.

For all of these to work, the disutility functions have to be {\em well-behaved}. To this end, we make the following assumptions about the rate of growth of the disutility functions.
\begin{assumption}\label{as:lip}
We assume that the disutility functions have Lipschitz-style lower- and upper-bounds.
Formally, for some constant $L>0$, 
we assume that for all $i\in [n]$ and for all $j\in[m]$, we have $| D_i( \bm x + \delta \cdot \bm e_j) - D_i (\bm x)| \geq \delta / L $;
furthermore, for all $i \in [n]$, and all $\bm x, \bm y \in \mathcal{D}+\mathbb{R}^n_{\geq0}$, we assume $\lvert D_i(\bm x) - D_i(\bm y) \rvert \leq L \cdot \lvert \lvert \bm x - \bm y \rvert \rvert_2$. 
\end{assumption}

The running time of our algorithm will be polynomial in $\log (L)$. However, we believe that we can also handle cases with a weaker lower-Lipschitz condition: for all $i \in [n]$, $|D_i(\bm x+ \delta \bm e_j) - D_i( \bm x)| \geq \frac{1}{L} \cdot \mathit{min}\big( \delta, \delta^k \big)$ for $k \in \textup{poly}(n,m)$. Towards the end of this section, we briefly mention what changes would be required to Algorithm~\ref{alg2:KKT} to make it work with the weaker assumption.  For simplicity, we stick to Assumption~\ref{as:lip} for the rest of this section.

Analogously to the linear case, we begin by showing in Section~\ref{gen:apprx-KKT-apprx-CEEI} that approximate KKT points will constructively yield approximate CEEI, 
and show in Section~\ref{sec:gen-KKT-alg} a refinement of Algorithm~\ref{alg:KKT} to find these approximate KKT points in the general setting.
Finally, in Section~\ref{sec:gen-alg-bounds}, we bound the number of iterations of this new algorithm, and show how to compute each iteration efficiently.

\subsection{\texorpdfstring{$(\lambda, \gamma, \delta)$}{(l,g,d)}-KKT Gives Approximate CEEI.} \label{gen:apprx-KKT-apprx-CEEI}
In this Subsection, we show the following strengthening of Theorem~\ref{corr:approx}.

\begin{theorem}\label{thm:gen-approx}
	Let $(\bm a, \bm d, \bm x)$ be a $(\lambda,\gamma,\delta)$-KKT point for the problem of minimizing $\mathcal L(\bm y)$ subject to $\bm y \in \mathcal D$, and $\mathcal L(\bm y) > -\infty$. 
	Then there exists payments $\bm p=(p_1,\,\dotsc,\,p_m)$ such that $(\bm x,\bm p)$ form a \eCEEI, as in Definition~\ref{def:CE}, where $\varepsilon = \max\{3(\gamma-1)+5\delta,\ \ \lambda-1\}$.
	
\end{theorem}

The whole of this subsection constitutes the proof of the above theorem.
Let $(\bm a, \bm d, \bm z)$ be a $(\lambda, \gamma, \delta)$-KKT point, as in Definition~\ref{def:apx-KKT}.
Since $\bm a_i \geq 0$ for all $i \in [n]$, let $\bm d^*$ be any point in $\mathcal D + \mathbb{R}^n_{\geq 0}$ such that 
$\{\bm y\in \mathbb R^n|\ip{\bm a, \bm y} \geq \ip{\bm a, \bm d^*}\}$ is a supporting hyperplane for 
$\mathcal D + \mathbb{R}^n_{\geq 0}$.
Since $\bm d^* \in \mathcal D + \mathbb{R}^n_{\geq 0}$ and $\ip{\bm a, \bm y} \geq \ip{\bm a , \bm d} - \delta $ for all $\bm y \in \mathcal D + \mathbb{R}^n_{\geq 0}$, we have $\ip{\bm a, \bm d^*} \geq \ip{\bm a, \bm d} - \delta$, or equivalently, $\ip{\bm a,\bm d} - \ip{\bm a,\bm d^*}\leq \delta$.

From here on, our proof emulates the proof of Theorem~\ref{corr:approx}, and consequently the proof of Bogomolnaia et al.~\cite{BogomolnaiaMSY17}.  
Recall that we have defined 
\[
\mathcal F' := \left\{ \bm y \in \mathbb{R}^{nm}_{\geq 0} \, \middle| \, \textstyle\sum_{i \in [n]} \bm y_{ij} \geq 1 \text{ for all } j \in [m] \right\}\ ,
\]
the pre-image of $\mathcal D+\mathbb R^n_{\geq 0}$ under $\vd$.
As before, let $S_{\lambda} = \{ \bm x \in R^{mn} \mid \ip{\bm a, \vd (\bm x)} \leq \lambda\}$. 
$S_{\lambda}$ is non-empty, closed and convex for all $\lambda>0$ as the disutilities are convex, and $\bm 0$ is feasible. 
Since $\ip{\bm a, \bm y} \geq \ip{\bm a, \bm d^*}$ for all $\bm y \in \mathcal D + \mathbb{R}^n_{\geq 0}$, 
we have $S_{\lambda} \cap \mathcal{F'} = \emptyset$ for all $\lambda < \ip{\bm a, \bm d^* }$. 
Thus, $S^* = S_{\ip{\bm a, \bm d^*}}$ is tangent to $\mathcal{F'}$ at the point $\bm z^*$ and $\vd( \bm  z^*) = \bm d^*$. 
Therefore, there exists a supporting hyperplane $\ip{\bm c, \bm x} = b$ of $\mathcal{F'}$ at $\bm z^*$, separating $\mathcal{F'}$ from $S^*$. 
We define the vector $\bm p \in \mathbb{R}^m_{\geq 0}$ such that  $\bm p_j = \mathit{min}_{i \in [n]}(\bm c)_{ij}$. 
Define $H_{\bm c} = \{ \bm x \in \mathbb{R}^{nm} \mid \ip{\bm c, \bm x} \geq b \}$ and 
$H_{\bm p} = \{ \bm x \in \mathbb{R}^{nm} \mid \sum_{j \in [m]} \sum_{i \in [n]}\bm p_j \bm x_{ij} \geq b \}$. 
As before, we have $\mathcal{F} \subseteq H_{\bm c}$.

\paragraph{Satisfying Condition~\ref{P1} in Definition~\ref{def:CE}.}
We want to show that for all agents $i$ and $i'$, we have $(\gamma + 2\delta)^{-2} \cdot  \langle \bm z_i, \bm p \rangle \leq \langle \bm z_{i'}, \bm p \rangle$. 
Note that for $\gamma$ sufficiently close to 1, and $\delta$ sufficiently small, $(\gamma+2\delta)^2 \leq 1+3(\gamma-1)+5\delta$.
We begin with the following observations.

\begin{claim}
	\label{gen:technical1}
	We have $\sum_{j \in [m]} \bm p_j = \ip{\bm c, \bm z^*} = b$.
\end{claim}
\begin{proof}
	Since $H_{\bm c}$ is a supporting hyperplane of $\mathcal F'$ at $\bm z^*$, we have 
	 $b = \min_{\bm x \in \mathcal F'} \ip{\bm c, \bm x}$. 
	 Observe that this minimum is obtained by assigning each chore fully to the agent that has the smallest $\bm c_{ij}$ value for it. Therefore, we have that
	 \[ 
	 	\min_{\bm x \in \mathcal F} \ip{\bm c, \bm x} = \textstyle\sum_{j \in [m]} \min_{i \in [n]} \bm c_{ij} = \sum_{j \in [m]} \bm p_j\ .\qedhere
	 	\]
\end{proof}

\begin{claim}
	\label{gen:technical2}
	$H_{\bm p} \subseteq H_{\bm c}$.
\end{claim}

\begin{proof}
	Consider any point $\bm x \in H_{\bm p}$. 
	We have $b \leq \sum_{i \in [n]} \sum_{j \in [m]} \bm p_j \bm x_{ij}$. 
	Since $p_j \leq c_{ij}$ for all $i \in [n]$ and $j\in [m]$,
	\[\textstyle
		b \leq \sum_{i \in [n]} \sum_{j \in [m]} \bm  p_j \bm x_{ij} \leq \sum_{i \in [n]} \sum_{j \in [m]} \bm c_{ij} \bm x_{ij} = \ip{\bm c,\bm x}\ ,
	\] 
	as desired.
\end{proof}

\begin{claim}
	\label{gen:technical3}
	$\mathcal F' \subseteq H_{\bm p}$, and $\mathcal F\subseteq \partial H_{\bm p}$, the boundary.
\end{claim}

\begin{proof}
	Consider any $\bm x \in \mathcal F$.
	Note that this is the original feasible region with equality.
	 We have 
	 \[\textstyle
	 	\sum_{i\in [n]}\sum_{j\in [m]} \bm p_j \bm x_{ij} = \sum_{j\in [m]}\bm  p_j = b\ .
	 \]
	The first equality holds since we assume $\sum_i \bm x_{ij}=1$ in $\mathcal F$,
	and the second holds by Claim~\ref{gen:technical1}.
	 If instead $\bm x\in \mathcal F'$, then the first equality becomes an inequality, concluding the proof.
\end{proof}

With these three claims, we can show the first condition for CEEI.
Assume for a contradiction that $(\gamma + 2\delta)^{-2} \cdot  \langle \bm z_i, \bm p \rangle > \langle \bm z_{i'}, \bm p \rangle$ for some $i$, $i'$.
Then we could replace the allocation as follows: Construct $\hat{\bm z}$ by setting $\hat{\bm z}_{i'}=\tfrac12\bm z_{i'}$, and $\hat{\bm z}_{i}=\left(1+\tfrac{\ip{\bm z_{i'}, \bm p}}{2\ip{\bm z_{i}, \bm p}}\right)\bm z_{i}$. 
Since $\bm z \in \mathcal F$, we have $\sum_{i \in [n],j \in [m]} \bm p_j \bm z_{ij} = b$. Also note that 
\[
b = \textstyle\sum_{i \in [n],j \in [m]} \bm p_j \hat{\bm z}_{ij}=\sum_{i \in [n],j \in [m]} \bm p_j \bm z_{ij}\ ,\]
since the payment subtracted from agent $i'$ is equal to the payment added to agent $i$ and so $\hat{\bm z}\in H_{\bm p}$.

By Claim~\ref{gen:technical2}, we have that  $ \hat{\bm z} \in H_{\bm c}$. 
Recall that $H_{\bm c}$ is a separating half-space between $S^*$ and $\mathcal{F'}$, \textit{i.e.}, $\mathcal{F'} \subseteq H_{\bm c}$ and $S^* \subseteq \operatorname{cl}(H_{\bm c}^{\complement})$, 
implying that for every point $\bm x \in H_{\bm c}$ we have 
$\langle{\bm a, \overrightarrow{\d}(\bm x)}\rangle \geq \langle{\bm a, \bm d^*}\rangle \geq \ip{\bm a, \vd (\bm z)} - \delta$.  Since $\hat{\bm {z}} \in H_{\bm c}$, we have  $\langle{\bm a,\overrightarrow{\d}(\hat{\bm z})}\rangle \geq \ip{\bm a, \vd (\bm z)} - \delta$. However,

\begin{align*}
	\langle \bm a, \vd(\hat{ \bm z}) \rangle - \langle \bm a, \vd (\bm z) \rangle &= -\frac{\bm a_{i'} \cdot D_{i'}(\bm z_{i'})}{2} +  \bm a_i \cdot D_i(\bm z_i) \cdot \frac{ \langle \bm p , \bm z_{i'} \rangle}{2\langle \bm p, \bm z_{i} \rangle}\\ 
	&\leq -\frac{\gamma^{-1}}{2} +  \gamma \cdot \frac{ \langle \bm p , \bm z_{i'} \rangle}{2\langle \bm p, \bm z_{i} \rangle} & (\gamma^{-1} \leq a_i \d_i(\bm z_i)\leq \gamma)\\
	&\leq -\frac{\gamma^{-1}}{2} + \frac{\gamma}{2(\gamma + 2\delta)^2}= 
-\frac{\gamma^{-1}}{2} \left( 1 - \frac{1}{(1 + \frac{2\delta}{\gamma})^2}\right) 
	\end{align*}
	
	Now, we have $1-\frac{1}{(1+p)^2} > 2p-3p^2$ for $p\geq 0$, and therefore
	\[
		-\frac{\gamma^{-1}}{2} \left( 1 - \frac{1}{(1 + \frac{2\delta}{\gamma})^2}\right)
		\ < \ 
		-\tfrac12 \gamma^{-1} ( 4\delta/\gamma - 8\delta^2/\gamma^2)
		\ \leq \ 
		-\tfrac12 \gamma^{-1} (3 \delta/\gamma)
		\ \leq \ -\delta\ ,
	\]
for $\delta$ sufficiently close to 0 and $\gamma$ sufficiently close to 1.
This implies that we have $\ip{\bm a, \vd(\hat{\bm z})} < \ip{\bm a, \vd(\bm z)} - \delta = \ip{\bm a, \bm d} - \delta$, which is a contradiction.

\paragraph{Satisfying Condition~\ref{P2} in Definition~\ref{def:CE}.} 
We want to show that for all $i \in [n]$, we have $(1-2\delta) \cdot D_i(\bm z_i) \leq D_i(\bm y)$ for all $\bm y$ such that $\ip{\bm y, \bm p} \geq \ip{\bm z_i, \bm p}$. 
Let us assume that there exists a $\bm y$ such that $(1-2\delta) \cdot  D_i(\bm z_i) > D_i(\bm y)$ and $\ip{\bm y, \bm p} \geq \ip{\bm z_i, \bm p}$. 
We define a new allocation $\bm z' = (\bm z_1, \bm z_2, \dots, \bm z_{i-1}, \bm y, \bm z_{i+1}, \dots ,\bm z_n)$. 
First note that $\sum_{i \in [n], j\in [m]} z'_{ij} \cdot p_j \geq \sum_{i \in [n], j\in [m]} z_{ij} \cdot p_j = b$ as $\ip{\bm y, \bm p} \geq \ip{\bm z_i, \bm p}$. 
Therefore $\bm z' \in H_{\bm p}$. 
By Claim~\ref{gen:technical2}, we have that $\bm z' \in H_{\bm c}$. 
Recall that $H_{\bm c}$ is a separating half-space between $S^*$ and $\mathcal{F}'$, \textit{i.e.}\  $\mathcal{F'} \subseteq H_{\bm c}$ and $S^* \subseteq \operatorname{cl}(H_{\bm c}^{\complement})$.
Therefore 
$\langle{\bm a,\overrightarrow{\d}(\bm z')}\rangle\geq \ip{\bm a, \bm d^*} \geq \langle{\bm a,\overrightarrow{\d}(\bm z)}\rangle - \delta$.
However, 
\begin{align*}
\langle \bm a, \vd( \bm z') \rangle - \langle \bm a, \vd( \bm z) \rangle &= \bm a_i \cdot (\d_i(\bm z'_i) - \d_i(\bm z_i))\\
									 &< -2\delta \bm a_i \d_i(\bm z_i) 
									 \leq - 2 \gamma ^{-1} \delta \leq -\delta\ .
\end{align*}
This implies that $\langle{\bm a,\overrightarrow{\d}(\bm z')}\rangle< \langle{\bm a,\overrightarrow{\d}(\bm z)}\rangle - \delta $, which is a contradiction.

\paragraph{Satisfying Condition ~\ref{P3} in Definition \ref{def:CE}.}
By definition of $(\lambda, \gamma, \delta)$-KKT point, we have $(\bm z)_{ij} \geq 0$ for all $i \in [n]$, and $j \in [m]$ and also $\lambda^{-1} \leq \sum_{i \in [n]} (\bm z)_{ij} \leq \lambda $ for all $j \in [m]$.

\bigskip

\subsection{Exterior-Point Methods for \texorpdfstring{$(\lambda, \gamma, \delta)$}{(l,g,d)}-KKT Points}\label{sec:gen-KKT-alg}
We introduce here a refinement of Algorithm~\ref{alg:KKT} which allows us to handle the extra errors, and finds 
$(\lambda, \gamma, \delta)$-KKT points in polynomial time. 
The algorithm takes as input three error terms, $\varepsilon_1$, $\varepsilon_2$, and $\varepsilon_3$, 
and we will later see how to set these to attain polynomial running time. We encourage the reader to go through Subsection~\ref{gen:overview} to get an overview of the entire algorithm and the challenges it handles compared to the setting where all agents have linear disutilities.

Recall that  $\mathcal F' := \mathcal F+\mathbb R_{\geq 0}^{nm} = \{ \bm x \in \mathbb{R}^{nm}_{\geq 0} \mid \sum_{i \in [n]} x_{ij} \geq 1 ~~\forall i \in [n] \}$, as above, and note that $\mathcal D+ \mathbb{R}^n_{\geq 0} = \{ \bm d \in \mathbb{R}^n \mid \bm d = \vd (\bm y) \text{ for some } \bm y \in \mathcal F'  \}$. Our new  algorithm is as follows.
\begin{algorithm}[H]
	\caption{Finding Approximate KKT for 1-Homogeneous Disutilities}
	\begin{algorithmic}[1]
		\State \label{alg2:init} Let $(\bm d^0) \gets \textsc{initialize}()$, be any infeasible, strictly positive, disutility profile near $\bm 0$, and $k=1$
		\While{\textbf{true}}
		\State  \textbf{Set} $(\bm x^k_{+}, \bm d^k_{+}) \gets \textsc{nearest-point}(\bm d^k, \varepsilon_1)$ \label{alg2:nearestpoint}
		\State \textbf{Set} $(\bm x^k_{+})_{ij} \gets (\bm x^k_+)_{ij} + L \cdot (2 \cdot \varepsilon_1 )$ for all $i \in [n]$ and $j \in [m]$ and $\bm d^k_+ \gets \overrightarrow{\d}(\bm x^k_+) $\label{alg2:pd}
	
		\If{$||\bm d^k_{+} - \bm d^{k}||_2 \leq \varepsilon_2$} 
		    \State \textbf{Set} $(\bm x^k_{+}, \bm d^k_{+}) \gets \textsc{adjust-coordinates} ({\bm x^k_+}, {\bm d^k_+}, \bm d^k)$. \label{alg2:ac}
			\State \textbf{Return} $(\bm a^{k-1}, \bm d^{k}_{+}, \bm x^k_{+})$ \label{alg2:returnclosepoint}
	    \EndIf 
		\State \textbf{Set} $\bm a^k = \bm d^k_{+} - \bm d^{k}$.
		\State \textbf{Rescale} $\bm a^k$ so that $\ip{\bm a^k,\bm d^k_{+}}= n$ \label{alg2:definehyperplane}.
		
		\State  \textbf{Set} $\bm d^{k+1} \gets \bm 1/\bm a^k$ \label{alg2:hyperplane-move}
		
		\If{$\logd(\bm d^{k+1},\bm d^k_{+})<\varepsilon_3$}
		\State \textbf{Return} $(\bm a^k, \bm d^{k}_{+},\bm x^k_{+})$ \label{alg2:stopping-hyp}
		\EndIf
		\State  \textbf{Set} $k\gets k+1$
		\EndWhile
	\end{algorithmic}
	\label{alg2:KKT}
\end{algorithm}

\begin{algorithm}[H]
	\caption{$\textsc{initialize} ()$}
	\begin{algorithmic}[1]
	\State \textbf{Fix} an arbitrary $a \in [n]$ and $c \in [m]$.
	\State \textbf{Set} $(\bm x^0)_{ac} \gets 1 / 2nL^2$ and $(\bm x)_{ij} = 0$ for all other $i \in [n]$ and $j \in [m]$.
	\State \textbf{Set} $\bm d^0 \gets \vd (\bm x^0)$.
	\State \textbf{Return} $\bm d^0$
	\end{algorithmic}
	\label{alg:initialize}
\end{algorithm}

\begin{algorithm}[H]
	\caption{$\textsc{nearest-point}({\bm d^k}, \varepsilon_1)$}
	Returns the point $(\bm x^k_+, \bm d^k_+)$ such that $\bm d^k_+ = \vd (\bm x^k_+)$ and $|| \bm x^k_+ -\bm x^k_* ||_2 \leq \varepsilon_1$ and $|| \bm d^k_+ - \bm d^k_*||_2 \leq \varepsilon_1$ where $\bm x^k_* \in \mathcal F' $ such that $\bm d^k_* = \vd (\bm x^k_*)$ is the nearest point to $\bm d^k$ in $\mathcal D +\mathbb{R}^n_{\geq 0}$. Also every coordinate of $\bm x^k_+$ and $\bm d^k_+$ is an integral multiple of $\varepsilon_1 / 2^{\textup{poly}(n,m)}$ for some $\textup{poly}(n,m)$.
	The details of this algorithm will be presented in Section~\ref{sec:gen-alg-bounds}.
	\label{alg:adjust-cordinates}	
\end{algorithm}

\begin{algorithm}[H]
	\caption{$\textsc{adjust-coordinates}({\bm x^k_+}, {\bm d^k_+}, \bm d^k)$}
	\begin{algorithmic}[1]	
		\State \textbf{Set} $y_{ij} = (\bm x^k)_{ij} \cdot \frac{(\bm d^k)_i}{(\bm d^{k}_{+})_i}$ for all $i \in [n]$, $j \in [m]$.
		\State \textbf{Return} $(\bm y, \overrightarrow{\d}(\bm y))$.   
	\end{algorithmic}
	\label{alg:adjust-cordinates}	
\end{algorithm}

We begin by showing that the above algorithm will correctly return a $(\lambda, \gamma, \delta)$-KKT point for the appropriate values of $\varepsilon_1,\,\varepsilon_2,\,\varepsilon_3$, which will be chosen later. 
For now, the reader should think of the $\varepsilon$'s as being related as follows: $\varepsilon_3 \gg \varepsilon_2 \gg \varepsilon_1$. 
In particular, we have $\varepsilon_3 \geq 2^{\textup{poly}(n,m)} \cdot \varepsilon_2$ and $\varepsilon_2 \gg n^3m^3L^3 \varepsilon_1$.

The proof of correctness for Algorithm~\ref{alg2:KKT} will be significantly more involved than that of Algorithm~\ref{alg:KKT}. 
Notably, the error in the computation of the nearest point $\bm d^k$ on line~\ref{alg2:nearestpoint} will introduce  many sources of additive error, which will need to be handled along with Assumption~\ref{as:lip} above, to ensure that we can recover multiplicative error guarantees.

We will need to argue that an approximate equilibrium can be found whether the algorithm returns in either of Steps~\ref{alg2:returnclosepoint} or~\ref{alg2:stopping-hyp}. 
The latter was the stopping condition for the original Algorithm~\ref{alg:KKT}, and its proof will be simpler. 

Let $\bm d^k_*$ be the true nearest point to $\bm d^k$ in $\mathcal D  + \mathbb{R}^n_{\geq 0}$ 
and let $\bm x^k_{*} \in \mathcal F'$ be such that $\vd(\bm x^k_{*}) = \bm d^k_*$.

\subsubsection{Stopping on Line~\ref{alg2:stopping-hyp} of Algorithm~\ref{alg2:KKT}}
We first show that the $\bm a^k$ vectors are indeed normals to approximately supporting hyperplanes.
\begin{lemma}
	\label{supporting-hyperplane}
	In each iteration $k$ of Algorithm~\ref{alg2:KKT}, after step~\ref{alg2:definehyperplane}, we have
	$\ip{\bm a^k, \bm y} \geq  \ip{\bm a^k, \bm d^k_+} - \delta$ for all $y \in \mathcal D + \mathbb{R}^n_{\geq 0}$ with $\delta = 9n^5mL^2  \frac{\varepsilon_1}{\varepsilon_2^2}$.
\end{lemma}

Before proving this results, we first need some technical claims.

\begin{observation}\label{cl:supphypcl}
	In each iteration $k$ of Algorithm~\ref{alg2:KKT}, after step~\ref{alg2:definehyperplane}, we have
	$\bm d^k_{*} \in \mathcal{D}$, $\bm x^k_{*} \in \mathcal{F}$, $||\bm d^k_*||_2 \leq nL$.
\end{observation}
\begin{proof}
	 Note that since the algorithm has constructed the vector $\bm a^k$, in iteration $k$, we have $\bm d^k \notin \mathcal D + \mathbb{R}^n_{\geq 0}$.
  Otherwise, by the choice of $\varepsilon$'s, the algorithm would have terminated since $\bm d^k_+$ would have been too close to $\bm d^k=\bm d^k_*$.
  Therefore, we have $|| \bm d^k_+ - \bm d^k ||_2 \ll \varepsilon_2$ and our algorithm will terminate before step~\ref{alg2:definehyperplane}. 
  Therefore, $\bm d^k$ lies outside $\mathcal D + \mathbb{R}^n_{\geq 0}$, 
  \textit{i.e.}, below the lower envelop of $\mathcal D + \mathbb{R}^n_{\geq 0}$. 
  By Lemma~\ref{lem:algo-steps-linear}, $\bm d^k_*\in\mathcal{D}$. 
  This implies that $\bm d^k_*$ has a pre-image $\bm x^k_* \in \mathcal F$ under $\vd$.
  As a result, we can bound $||\bm d^k_*||_2 \leq \sqrt{n} \cdot \mathit{max}_{i \in [n]}  D_i(\bm 1)$. 
  Note that by Assumption~\ref{as:lip}, we have $D_i (\bm 1) = D_i(\bm 1) -D_i(\bm 0)  \leq L \cdot || \bm 1 - \bm 0||_2 \leq L \sqrt{n}$ for all $i \in [n]$. Therefore, we have $||\bm d^k_*||_2 \leq nL$. 
\end{proof}

\begin{observation}
	\label{pareto-domination}
	After step~\ref{alg2:pd} of Algorithm~\ref{alg2:KKT}, we have $(\bm x^k_+)_{ij} \geq (\bm x^k_*)_{ij}$ for all $i \in [n]$, $j \in [m]$ and that $(\bm d^k_+)_i \geq (\bm d^k_*)_i \geq (\bm d^k)_i$ for all $i \in [n]$, implying that $\bm x^k_+ \in \mathcal F'$ and $\bm d^k_+ \in \mathcal D + \mathbb{R}^n_{\geq 0}$
\end{observation}

\begin{proof}
	 By construction, after step~\ref{alg2:nearestpoint}, we have $|(\bm x^k_*)_{ij} - (\bm x^k_+)_{ij} |\leq \varepsilon_1$ for all $i \in [n]$ and $j \in [m]$ and $|(\bm d^k_*)_i - (\bm d^k_+)_i| \leq \varepsilon_1$ for all $i \in [n]$. 
	 Then in step~\ref{alg2:pd}, we increase all $(\bm x^k_+)_{ij}$ by $2L\varepsilon_1$ and thus we have $(\bm x^k_+)_{ij} \geq (\bm x^k_*)_{ij}$. Since our disutility functions have lower-bounded partial derivatives, we have for each $i \in [n]$, the disutility of agent $i$ increases by at least $1/ L \cdot 2L \varepsilon_1 = 2\varepsilon_1$ and thus we have $(\bm d^k_+)_i \geq (\bm d^k_*)_i$ for all $i \in [n]$. 
	 We complete the proof by showing that $(\bm d^k_*)_i \geq (\bm d^k)_i$ for all $i \in [n]$. To this end, first observe that if $\bm d^k \in \mathcal D + \mathbb{R}^n_{\geq 0}$, then $\bm d^k_* = \bm d^k$ and the claim holds trivially. If $\bm d^k \notin \mathcal D + \mathbb{R}^n_{\geq 0}$, then by the same argument in Lemma~\ref{lem:algo-steps-linear}, we can prove $(\bm d^k_{*})_i \geq (\bm d^k)_i$ for all $i \in [n]$.  
\end{proof}

We also show that $\bm d^k_+$ and $\bm x^k_+$ are in the $O(nmL^2\varepsilon_1)$ neighbourhood of $\bm d^k_*$ and $\bm x^k_*$ respectively.

\begin{observation}
	\label{appx-nearest-point-dist}
    After step~\ref{alg2:pd} of Algorithm~\ref{alg2:KKT},  we have 
	 \begin{itemize}
	 	\item $||\bm x^k_{+} - \bm x^k_{*}||_2 \leq 3nmL \varepsilon_1$, and 
	 	\item  $||\bm d^k_{+} - \bm d^k_{*}||_2 \leq 3nmL^2 \varepsilon_1$.
	 \end{itemize}
\end{observation} 

\begin{proof}
	After step~\ref{alg2:nearestpoint}, we have $||\bm x^k_+ - \bm x^k_*||_2 \leq \varepsilon_1$. 
	Then, in step~\ref{alg2:pd}, we increased each $(\bm x^k_+)_{ij}$ by $2L \varepsilon_1$. 
	Therefore, after step~\ref{alg2:pd}, we have $||\bm x^k_+ - \bm x^k_*||_2 \leq \varepsilon_1 + 2nmL \varepsilon_1 \leq 3nmL \varepsilon_1$.	
	Observe that,
	\begin{align*}
	 ||\bm d^k_{+} - \bm d^k_{*}||^2_2 &= \sum_{i \in [n]} |D_i ((\bm x^k_+)_i) - D_{i}((\bm x^k_*)_i)|^2\\
   \text{(by Assumption~\ref{as:lip})}\qquad &\leq  L^2 \cdot \sum_{i \in [n]} |(\bm x^k_+)_i - (\bm x^k_*)_i|^2 
	                                   \ \leq\  L^2 \cdot ||\bm x^k_+ -\bm x^k_*||^2_2
	                                   \ \leq\  L^2 \cdot (3nmL \varepsilon)^2.
	\end{align*} 
    This implies that 	$||\bm d^k_{+} - \bm d^k_{*}||_2 \leq 3nmL^2 \varepsilon_1$.
\end{proof}

We can now prove the lemma.

\begin{proof}[Proof of Lemma~\ref{supporting-hyperplane}.]
	Recall, by Lemma~\ref{lem:algo-steps-linear}, the direction $\bm d^k_*-\bm d^k$ is normal to a supporting hyperplane for $\mathcal D+\mathbb R^n_{\geq 0}$ at $\bm d^k_*$.
	Thus, for any $\bm y \in \mathcal D + \mathbb{R}^n_{\geq 0}$, 
	we can conclude from Observation~\ref{pareto-domination} above, that
  \begin{align}
  \ip{\bm d^k_+ - \bm d^k, \bm y} &\geq \ip{\bm d^k_* - \bm d^k, \bm y} 
                                  \geq \ip{\bm d^k_* - \bm d^k, \bm d^k_*} \label{eq}
  \end{align}
    Now, we wish to bound $\ip{\bm d^k_* - \bm d^k, \bm d^k_*}$ so that we can control the hyperplane error.
    We have
	\begin{align*}
	\ip{\bm d^k_+ - \bm d^k, \bm d^k_+} - \ip{\bm d^k_* - \bm d^k, \bm d^k_*}  &=
	\ip{\bm d^k_+ - \bm d^k, \bm d^k_+} - \ip{\bm d^k_+ - \bm d^k, \bm d^k_*} + \ip{\bm d^k_+ - \bm d^k, \bm d^k_*} - \ip{\bm d^k_* - \bm d^k, \bm d^k_*} \\
	&= \ip{\bm d^k_+ - \bm d^k, \bm d^k_+ - \bm d^k_*} + \ip{\bm d^k_+ - \bm d^k_*, \bm d^k_*}\\
	&\leq \Vert \bm d^k_+ - \bm d^k_* \Vert_2\cdot\left( 
		\Vert \bm d^k_+ - \bm d^k \Vert_2	+ \Vert \bm d^k_*\Vert_2
	\right)\\
	&\leq 3nmL\varepsilon_1\cdot\left(
		(nL+3nmL\varepsilon_1) + 0 + nL 
	\right)\leq 9n^2mL^3\varepsilon_1\ ,
\end{align*}

  where the last row follows from the fact that $||\bm d^k_* - \bm d^k_+||_2 \leq 3nmL^2 \varepsilon_1$ (Observation~\ref{appx-nearest-point-dist}) and that $|| \bm d^k_*||_2 \leq nL$. 
Thus,
	\begin{align*}
		\ip{\bm d^k_+ - \bm d^k,\bm d^k_+}&\leq \ip{\bm d^k_*-\bm d^k,\bm d^k_*}+9n^2mL^3\epsilon_1\\
			&\leq \ip{\bm d^k_+ - \bm d^k,\bm y}+9n^2mL^3\epsilon_1\quad \forall\ \bm y\in \mathcal D+\mathbb R^n_{\geq 0}
			&\text{by~\eqref{eq}}
	\end{align*}

Since $\bm a^k = (n / \ip{\bm d^k_+ - \bm d^k, \bm d^k_+}) \cdot (\bm d^k_+ - \bm d^k)$, we have for for all $y \in \mathcal D  + \mathbb{R}^n_{\geq 0}$,
  \begin{align}
  \label{apprx-hyperplane}
  \ip{\bm a^k, \bm y} &\geq \ip{\bm a^k, \bm d^k_+} - \frac{n } {\ip{\bm d^k_+ - \bm d^k, \bm d^k_+}} \cdot 9n^2mL^3\varepsilon_1                   
  \end{align}
  Observe that since we are step~\ref{alg2:definehyperplane} of Algorithm~\ref{alg2:KKT}, our algorithm did not terminate in step~\ref{alg2:returnclosepoint}, and we have $||\bm d^k_+ - \bm d^k||_2 > \varepsilon_2$. Again, since $(\bm d^k_+)_i \geq (\bm d^k_*)_i \geq (\bm d^k)_i$ for all $i \in [n]$ (by Observation~\ref{pareto-domination}), there exists an $i' \in [n]$, such that $(\bm d^k_+)_i \geq (\bm d^k_+)_i -(\bm d^k)_i \geq \varepsilon_2/n$, implying that $\ip{\bm d^k_+ -\bm d^k, \bm d^k_+} \geq \varepsilon^2/n^2$. Substituting this lower bound in~\ref{apprx-hyperplane}, we have,
  \begin{align*}
  \ip{\bm a^k, \bm y} &\geq \ip{\bm a^k, \bm d^k_+} -  9n^5mL^3  (\varepsilon_1/\varepsilon_2^2). \qedhere                  
  \end{align*} 
\end{proof}

With the bound of Lemma~\ref{supporting-hyperplane}, we can now show that if the algorithm stops on Step~\ref{alg2:stopping-hyp}, \textit{i.e.}\ if $\bm d^{k+1}$ is too close to $\bm d^k_+$, then we have an approximate KKT point.

\begin{lemma}
	\label{appx-KKT-point-1}
	Let $(\bm a^k, \bm x^k_{+}, \bm d^k_{+})$ be the point returned by Algorithm~\ref{alg2:KKT} in Step~\ref{alg2:stopping-hyp}. Then $(\bm a^k, \bm x^k_{+}, \bm d^k_{+})$ is a $(\lambda, \gamma, \delta)$-KKT point with $\lambda = 1 + n^2L \varepsilon_1$, $\gamma = 1+\varepsilon_3 $ and $\delta = 9n^5mL^3  \frac{\varepsilon_1}{\varepsilon_2^2}$. 
\end{lemma}
	\begin{proof}
	We first show that for each $i \in [n]$, we have $ \gamma^{-1} \leq (\bm a^k)_i \cdot (\bm d^k_{+})_i \leq \gamma$. 
	Since our algorithm returns this point in step~\ref{alg2:stopping-hyp}, 
	we have 
	$\logd (\bm d^k_{+}, \bm 1 /\bm a^{k}) \leq \varepsilon_3$. 
	This implies that for each $i \in [n]$, we have $(1+\varepsilon_3)^{-1} \leq (\bm d^k_{+})_i \cdot (\bm a^k)_i \leq (1 + \varepsilon_3)$. 
	Also, by Lemma~\ref{supporting-hyperplane}, we have $\ip{\bm a^k, \bm y} \geq \ip{\bm a^k, \bm d^k_+} -\delta = n -\delta $ for all $\bm y \in \mathcal D + \mathbb{R}^n_{\geq 0}$ and $\delta = 9n^5mL^3 \varepsilon_1 / (\varepsilon_2^2)$.   
	
	It remains to show that $(\bm x^k_+)_{ij} \geq 0$ for all $i \in [n]$, $j \in [m]$ and  $\lambda^{-1} \leq \sum_{i \in [n]} (\bm x^k_+)_{ij} \leq \lambda$. 
	We have that after step~\ref{alg2:pd} of Algorithm~\ref{alg2:KKT}, $(\bm x^k_+)_{ij} \geq (\bm x^k_*)_{ij} \geq 0$ for all $i,j$.
	This is because by construction, after step~\ref{alg2:nearestpoint} of Algorithm~\ref{alg2:KKT}, 
	$(\bm x^k_+)_{ij} \geq (\bm x^k_*)_{ij} - \varepsilon_1$. 
	Furthermore after step~\ref{alg2:pd}, $(\bm x^k_+)_{ij}$ is increased by an additive factor of $2L\varepsilon_1$ and thus it becomes larger than $(\bm x^k_*)_{ij}$. 
	Since $\bm d^k_* \in \mathcal D + \mathbb{R}^n_{\geq 0}$ and $\bm x^k_* \in \mathcal F'$, we have  
	$(\bm x^k_{*})_{ij}  \geq 0$, further implying that $(\bm x^k_+)_{ij} \geq 0$ as well.	
	
	Now, as our algorithm is returning in step~\ref{alg2:stopping-hyp} and not step~\ref{alg2:returnclosepoint}, we do not change $\bm x^k_+$ after step~\ref{alg2:pd}. Therefore, we have $(\bm x^k_+)_{ij} \geq 0$ for all $i \in [n]$ and $j \in [m]$ after step~\ref{alg2:stopping-hyp} also.

	We now show $\lambda^{-1} \leq \sum_{i \in [n]} (\bm x^k_+)_{ij} \leq \lambda$. To this end, first observe that after step~\ref{alg2:pd},   $(\bm x^k_*)_{ij} \leq (\bm x^k_+)_{ij} \leq (\bm x^k_*)_{ij} + \varepsilon_1 + 2L\varepsilon_1 \leq (\bm x^k_*)_{ij} + 3L\varepsilon_1$ for all $i \in [n]$ and $j \in [m]$. Furthermore, by Observation~\ref{cl:supphypcl}, we have that $\bm d^k_* \in \mathcal D$ and $\bm x^k_* \in \mathcal F$. Therefore, we have $\sum_{i \in [n]} (\bm x^k_*)_{ij} = 1$ for all $j \in [m]$. Since, $(\bm x^k_*)_{ij} \leq (\bm x^k_+)_{ij} \leq (\bm x^k_*)_{ij} + 3L\varepsilon_1$, we have $(\lambda)^{-1} \leq 1 - 3nL\varepsilon_1 \leq \sum_{i \in [n]}(\bm x^k_+)_{ij} \leq 1 + 3nL\varepsilon_1 \leq \lambda$.
\end{proof}

\subsubsection{Stopping on Line~\ref{alg2:returnclosepoint} of Algorithm~\ref{alg2:KKT}}
We have shown that if Algorithm~\ref{alg2:KKT} stops in step~\ref{alg2:stopping-hyp}, then we have an approximate KKT point.
It remains to show that this holds if we stop in step~\ref{alg2:returnclosepoint}. We start by addressing a small subtlety. The normal vector  $\bm a^{k-1}$ is only well defined from $k \geq 2$. Therefore, we need to show that Algorithm~\ref{alg2:KKT} never returns at Line~\ref{alg2:returnclosepoint} at $k=1$ (the first iteration). This follows from the fact that the distance between $\bm d^0$ and any point in $\mathcal D + \mathbb{R}^n_{\geq 0} \gg \varepsilon_2$. To see this, note that for any disutility vector $\bm d \in \mathcal D + \mathbb{R}^n_{\geq 0}$, there is one agent who gets at least $1/n$ fraction of some chore and as a result his disutility will be at least $1/ nL$ (by Assumption~\ref{as:lip}). However, the disutility of any agent in $\bm d^0$ is at most $L \cdot ||(\bm x^0)_i - \bm 0||_2 \leq 1/ 2nL$ (by Assumption~\ref{as:lip} and Algorithm~\ref{alg:initialize}). Therefore $|| \bm d - \bm d^0||_2 \geq || \bm d - \bm d^0||_{\infty} \geq 1 / 2nL \gg \varepsilon_2$. We now focus on the main proof.

Note that this stopping condition is an additive error, and we will need to be more careful. The bulk of the proof will lie in showing that the allocation $\bm x^k_+$ is neither an over-allocation nor an under-allocation of any of the chores. The remaining conditions will be relatively straightforward, as they are a consequence of Lemma~\ref{supporting-hyperplane} on the previous iteration.

\begin{lemma}
	\label{appx-KKT-point-2-gammadelta}
	Let $(\bm a^{k-1}, \bm x^k_{+}, \bm d^k_{+})$ be the point returned by Algorithm~\ref{alg2:KKT} in Step~\ref{alg2:returnclosepoint}. Then $(\bm a^{k-1}, \bm x^k_{+}, \bm d^k_{+})$ is a $(\lambda, \gamma, \delta)$-approximate KKT point with $\gamma = 1+\varepsilon_3 $, $\delta = 9n^5mL^3  \frac{\varepsilon_1}{\varepsilon_2^2}$, and some $\lambda$.
\end{lemma}

\begin{proof}
We have that (i) $\bm d^k$ lies on the hyperplane $\ip{\bm a^{k-1}, \bm y} = n$ where $\bm a^{k-1} = \bm 1/ \bm d^k$, by construction, and (ii) $\ip{\bm a^{k-1}, \bm y} = n$ is a $ \delta$-approximate supporting hyperplane of $\mathcal D + \mathbb{R}^n_{\geq 0}$, by Lemma~\ref{supporting-hyperplane}). 	
	Therefore, to show that $(\bm x^k_+, \bm d^k_+)$ is a $(\lambda, \gamma, \delta)$-approximate KKT point with no $\lambda$ bound, it suffices to show that $(\bm x^k_+)_{ij} \geq 0$ for all $i \in [n]$ and $j \in [m]$.		

	We have already shown in the proof of Lemma~\ref{appx-KKT-point-1} that after line~\ref{alg2:pd} of Algorithm~\ref{alg2:KKT}, $(\bm x^k_+)_{ij} \geq 0$ for all $i \in [n]$ and $j \in [m]$.
	The returned allocation, after the application of \textsc{adjust-coordinates}, is a positively re-scaling of this vector.
	Thus, we satisfy the conditions of a $(\lambda,\gamma,\delta)$-KKT point assuming no $\lambda$ bounds.
\end{proof}

It remains then to show that $\sum_{i=1}^n (\bm x^k_+)_{ij}$ is not too far from 1 in any direction.
Notice that our supporting hyperplane is approximately supporting the set $\mathcal D+\mathbb R^n_{\geq 0}$.
Thus, it is relatively straightforward to argue, as we do here, that under-allocations are unlikely, 
but any arbitrary over-allocation will need to be controlled. 
We begin by ruling out under-allocations.

Recall, if the algorithm stops on line~\ref{alg2:returnclosepoint}, then it will have applied \textsc{adjust-coordinates} to the allocation. 
In what follows, let $\bm z^k_+$ denote the value of $\bm x^k_+$ {\em before} the application of \textsc{adjust-coordinates}, \textit{i.e.} the value of $\bm x^k_+$ on line~\ref{alg2:pd}, and let $\bm x^k_+$ be the returned allocation.
Formally $(\bm x^k_+,\vd(\bm x^k_+)) = \textsc{adjust-coordinates}(\bm z^k_+, \vd ( \bm z^k_+), \bm d^k)$.

\begin{claim}
		\label{adjustcordinatesclaim}
		If we stop on line~\ref{alg2:returnclosepoint}, then $\bm d^k = \bm d^k_+$.
	\end{claim}
    
    \begin{proof}
    	Note that $(\bm x^k_+)_i = (\bm z^k_+)_i \cdot \tfrac{(\bm d^k)_i}{D_i((\bm z^k_+)_i)}$. As the disutility functions are $1$-homogeneous, we have for each $i \in [n]$,  $(\bm d^k_+)_i = D_i((\bm x^k_+)_i) = D_i ((\bm z^k_+)_i) \cdot \tfrac{(\bm d^k)_i}{D_i((\bm z^k_+)_i)} = (\bm d^k)_i$. 
    \end{proof}

\begin{lemma}\label{appx-KKT-point-2-lambda-under}
	Let $(\bm a^{k-1}, \bm x^k_{+}, \bm d^k_{+})$ be the point returned by Algorithm~\ref{alg2:KKT} in Step~\ref{alg2:returnclosepoint}. Then we have that $\sum_{i \in [n]} (\bm x^k_+)_{ij} \geq 1 - nL \varepsilon_2$ for all items $j\in [m]$.
\end{lemma}
\begin{proof}
We begin by showing that for all $i \in [n]$ and $j \in [m]$, $(\bm x^k_+)_{ij} \geq (\bm z^k_+)_{ij} - L\varepsilon_2$.
Note that by Observation~\ref{pareto-domination}, 
we have $\d_i((\bm z^k_+)_i) \geq (\bm d^k_*)_i \geq (\bm d^k)_i = \d_i((\bm x^k_+)_i)$ for all $i \in [n]$.
Therefore, $\tfrac{(\bm d^k)_i}{\d_i((\bm z^k_+)_i)} \leq 1$ for all $i \in [n]$, 
implying that no agent increases their consumption of any chore, \textit{i.e.}, 
$(\bm x^k_+)_{ij} \leq (\bm z^k_+)_{ij}$ for all $i \in [n]$ and $j \in [m]$. 
Now, assume that there exist an $i \in [n]$ and $j \in [m]$ such that $(\bm x^k_+)_{ij} < (\bm z^k_+)_{ij} - L\varepsilon_2$.  
Since the disutility functions have lower-bounded partial derivatives (Assumption~\ref{as:lip}), 
and agent $i$ does not increase consumption of any other chore from $(\bm z^k_+)_i$ to $(\bm x^k_+)_i$, 
we have $ \d_i((\bm x^k_+)_i) < \d_i((\bm z^k_+)_i) - L\varepsilon_2 / L$, or equivalently $(\bm d^k)_i < \d_i((\bm z^k_+)_i) - \varepsilon_2$, contradicting the fact that $|| \bm d^k - \vd ( \bm z^k_+) ||_2 \leq \varepsilon_2$.  Therefore, $(\bm x^k_+)_{ij} \geq (\bm z^k_+)_{ij} - L\varepsilon_2$.   

Now,
by Observation~\ref{pareto-domination}, we have that $\bm z^k_+ \in \mathcal{F'}$, implying that $\sum_{i \in [n]} (\bm z^k_+)_{ij} \geq 1$ for all $j \in [m]$. 
Therefore, if we have $(\bm x^k_+)_{ij} \geq (\bm z^k_+)_{ij} - L\varepsilon_2$, then we have $\sum_{i \in [n]} (\bm x^k_+)_{ij} \geq 1 -nL\varepsilon_2$ for all $j \in [m]$. 
\end{proof}

	\paragraph{No chores are significantly over-allocated.} We now show $\sum_{i \in [n]} (\bm x^k_+)_{ij} \leq 1 + \beta$ for all $j \in [m]$, where $\beta = 2mn^2L^3(\alpha + \varepsilon_2)$ and $\alpha = 48n^7L^5 \varepsilon_1/ \varepsilon_2^3$. The reason behind the exact choice of the upper bound will become explicit by the end of Claims~\ref{disutilitydecrease} and~\ref{contradictiony}. 
	We start by making some observations on $\bm d^k$ and $\bm a^{k-1}$. 
	Recall, $\bm d^k$ lies on the hyperplane $\ip{\bm a^{k-1}, \bm y} = n$, where 
	\[\bm a^{k-1} = \frac{n}{\ip{\bm d^{k-1}_+ - \bm d^{k-1}, \bm d^{k-1}_+}} (\bm d^{k-1}_+ - \bm d^{k-1})\ .\] 
	We start by showing that there is at least one coordinate where $\bm a^{k-1}$ is not small, w.r.t~$\varepsilon_2$.
	
	\begin{claim}
		\label{lower-boundona}
		There exists an $i_0 \in [n]$ such that $(\bm a^{k-1})_{i_0} \geq \varepsilon_2 / (4n^2L^2)$.
	\end{claim}

    \begin{proof} 
     In the $(k-1)$-st iteration, 
     Algorithm~\ref{alg2:KKT} did not return at step~\ref{alg2:returnclosepoint}, 
     and so $\Vert \bm d^{k-1}_+ - \bm d^{k-1}\Vert_2 > \varepsilon_2$. 
     There must be some $i \in [n]$ such that $|(\bm d^{k-1}_+)_i - (\bm d^{k-1})_i| > \varepsilon_2/n$. 
     By Observation~\ref{pareto-domination}, $(\bm d^{k-1}_+)_i \geq (\bm d^{k-1}_*)_i \geq (\bm d^{k-1})_i$, 
     and thus $(\bm d^{k-1}_+)_i - (\bm d^{k-1})_i \geq \varepsilon_2 / n$.
    Now,
    \begin{align*}
     \ip{\bm d^{k-1}_+ - \bm d^{k-1}, \bm d^{k-1}_+} &\leq ||\bm d^{k-1}_+||_2^2\\ 
                                                 & \leq (||\bm d^{k-1}_*||_2 + 3nmL^2\varepsilon_1)^2 &\text{(by Observation~\ref{appx-nearest-point-dist})} 
    \end{align*}
     Again, since the algorithm did not return at step~\ref{alg2:returnclosepoint} in the $(k-1)$-st iteration, we have $\Vert\bm d^{k-1}_{*}\Vert_2 \leq nL$ by Observation~\ref{supporting-hyperplane}. 
     This implies that $\ip{\bm d^{k-1}_+ - \bm d^{k-1}, \bm d^{k-1}_+} \leq (2nL)^2 = 4n^2L^2$. 
          
     Now note that $(\bm a^{k-1})_i = \tfrac{n}{\ip{\bm d^{k-1}_+ - \bm d^{k-1}, \bm d^{k-1}_+}} ((\bm d^{k-1}_+)_i - (\bm d^k)_i)  \geq \tfrac{n}{4n^2L^2} \cdot \tfrac{\varepsilon_2}{n} = \varepsilon_2 / (4n^2L^2)$.   	
    \end{proof}

	We now define a new allocation $\bm y$ from $\bm x^k_+$ such that $\bm y \in \mathcal F'$ and consequently $\vd (\bm y) \in \mathcal D + \mathbb{R}^n_{\geq 0}$. 
	This would imply that $\ip{\bm a^{k-1}, \vd(\bm y)} \geq \ip{\bm a^{k-1}, \bm d^{k}_+} -\delta$.
	By Claim~\ref{adjustcordinatesclaim}, this equals $\ip{\bm a^{k-1}, \bm d^{k}} -\delta  = n -\delta$.

	We will show in the following that if any chore is significantly over-allocated in $\bm x^k_+$, then $\ip{\bm a^{k-1}, \vd(\bm y)} < n -\delta$ which is a contradiction. 
    Let us therefore assume that there is a $j' \in [m]$ which is over-allocated, \textit{i.e}, 
    $\sum_{\ell \in [n]} (\bm x^k_+)_{\ell j'} \geq 1 + \beta$. 
    This implies that there is some $i' \in [n]$ such that $(\bm x^k_+)_{i'j'} > 1/n + \beta / n$. 
    Furthermore, for all $j$, we denote by $r_{j}$, the excess amount of chore $j$ left undone in $\bm x^k_+$, 
    \textit{i.e.}, $r_{j} = \max\{0,\ 1 - \sum_{\ell \in [n]} (\bm x^k_+)_{\ell j}\}$.

	As a corollary of Lemma~\ref{appx-KKT-point-2-lambda-under}, we can bound $r_{j'}$.
    \begin{claim}
    	\label{boundonr}
    	For all $j \in [m]$, we have $r_j \leq nL\varepsilon_2$.
    \end{claim}	
	
	To define the allocation $\bm y$, we distinguish two cases.
	Recall that the agent $i_0$ has $(\bm a^{k-1})_{i_0} \geq \varepsilon_2 / (4n^2L^2)$, by Claim~\ref{lower-boundona}.
	Furthermore, recall that we have chosen $\alpha:= 48n^7L^5 \varepsilon_1/ \varepsilon_2^3$ for our desired $\lambda$ bound.
	
	\begin{description}
		\item[Case 1.] For some item $j$, $(\bm x^k_+)_{i_0j} \geq \alpha $.
		In this case, agent $i_0$ consumes a non-negligent amount of some chore $j$ w.r.t. the $\varepsilon$'s.
		Define $\bm y=(\bm y_1,\,\dotsc,\,\bm y_n)$ as follows:
		\[
			\bm y_\ell :=\begin{cases}
				(\bm x^k_+)_{i_0} - \alpha \cdot \bm e_j &\text{ if $\ell=i_0$,}\\
				(\bm x^k_+)_{i'}  + \alpha \cdot \bm e_j + \sum_{q \in [m] \setminus \{j,j'\}} r_{q} \cdot \bm e_{q}  - (\beta / n) \cdot \bm e_{j'} &\text{ if $\ell=i'$,}\\
				(\bm x^k_+)_\ell &\text{ otherwise.}
			\end{cases}
		\]		
		This has the following effects:
		(i) we decrease $i_0$'s consumption of $j$ by $\alpha$ units and increase $i'$'s consumption of $j$ by $\alpha$ units (so the total consumption of $j$ remains unchanged), 
		then (ii) increase the consumption of every under-consumed chore for agent $i'$ until their total consumption becomes 1 and finally 
		(iii) decrease the $i'$'s consumption of $j'$ (the overallocated chore) by $\beta /n$ units. 
		  
		\item[Case 2.] If instead, $(\bm x^k_+)_{i_0j} < \alpha $ for all items $j$,
		define $\bm y=(\bm y_1,\,\dotsc,\,\bm y_n)$ as follows:
		\[
			\bm y_\ell :=\begin{cases}
				\bm 0 &\text{ if $\ell=i_0$,}\\
				(\bm x^k_+)_{i'} + (\bm x^k_+)_{i_0} + \sum_{q \in [m] \setminus \{j'\}} r_{q} \cdot \bm e_{q} - (\beta / n) \cdot \bm e_{j'} &\text{ if $\ell=i'$,}\\
				(\bm x^k_+)_\ell &\text{ otherwise.}
			\end{cases}
		\]		
		(i) We decrease $i_0$'s consumption of each chore to zero and increase $i'$'s consumption of item $q$ by $(\bm x^k_+)_{iq}$ units (so the total consumption of each chore $q$ remains unchanged), 
		then (ii) we increase the consumption of every under-consumed chore for agent $i'$ until their total consumption becomes 1 and finally 
		(iii) we decrease $i'$'s consumption of $j'$ (the overallocated chore) by $\beta /n$ units. 
	\end{description}

	 We first show that $\bm y \in \mathcal F'$.
	 
	 \begin{claim}
	 	\label{yinF'}
	 	We have $\bm y \in \mathcal F'$.
	 \end{claim}
	 
	 \begin{proof}
	  In both Case 1 and 2, we have increased agent $i'$'s consumption of each the under-allocated chores $\ell$ ($r_\ell>0$) by $r_\ell$, with the exception of $j'$ and $j$ if applicable.
	  Furthermore, the decrease in consumption of any chore for agent $i_0$ is matched by an increase for $i'$, before subtracting the $j'$ term.  
	   Therefore, $\sum_{i \in [n]} (\bm y)_{i \ell} \geq 1$ for all $\ell \in [m] \setminus j'$.  
	   Finally, the consumption of chore $j'$ is decreased by $\beta / n$, but since the total consumption of $j'$ in $\bm x^k_+$ is at least $1 + \beta$, 
	   the total consumption of $j'$ in $\bm y$ is at least $1$. 
	   Thus $\bm y \in \mathcal F'$ and $ \vd(\bm y) \in \mathcal{D} + \mathbb{R}^n_{\geq 0}$. 
   \end{proof}

	We next argue that the disutility values of all agents have not increased. 
	
	\begin{claim}
		\label{disutilitydecrease}
		For all $\ell \in [n]$, we have $D_{\ell} ((\bm y)_{\ell}) \leq D_{\ell}((\bm x^k_+)_{\ell})$.
	\end{claim}
    
    \begin{proof}
	Note that we have only reduced the consumption of chores for all agents except $i'$. Therefore, the disutility values for all agents in $[n] \setminus \{ i' \}$ decreases. It suffices to show that agent $i'$'s disutility also decreases. We now argue that $D_{i'} ((\bm y)_{i'}) < D_{i'} ( (\bm x^k_+)_{i'})$.
	
	Let $\bm \Delta$ be defined to equal $\alpha\bm e_j$ in Case 1, and $(\bm x^k_+)_{i'}$ in Case 2. 
	In both cases, we are subtracting $\bm \Delta$ from $(\bm x^k_+)_{i0}$ and adding it to $(\bm x^k_+)_{i'}$.
	We recall that 1-homogeneous and convex functions are sub-additive: 
	$\d(\bm p+\bm q) = 2\d(\tfrac12\bm p+\tfrac 12\bm q)\leq 2\tfrac12(\d(\bm p)+\d(\bm q)) = \d(\bm p)+\d(\bm q)$.
	Therefore, 
		\begin{align*}
	\d_{i'}&(\bm y_{i'})\\
		&=  \d_{i'} \left((\bm x^k_+)_{i'} + \bm \Delta + \textstyle\sum_{\ell \in [m] \setminus \{j,j'\}} r_{\ell} \cdot \bm e_{\ell} - (\beta / n) \cdot \bm e_{j'}\right) \\
		&\leq  \d_{i'} \left((\bm x^k_+)_{i'} - (\beta / n) \cdot \bm e_{j'}\right) + 
				\d_{i'}(\bm \Delta) + \sum_{\ell \in [m] \setminus \{j,j'\}} \d_{i'}(r_{\ell} \cdot \bm e_{\ell})\\
		&\leq  \underbrace{\d_{i'} \left((\bm x^k_+)_{i'}\right) - \tfrac 1L (\beta/n)}_{\text{Assumption~\ref{as:lip}}} + \d_{i'}(\bm \Delta) + \sum_{\ell \in [m] \setminus \{j,j'\}} \underbrace{\d_{i'}(nL\varepsilon_2 \cdot \bm e_{\ell})}_{\text{Claim~\ref{boundonr}}}
		\end{align*}
		
		Now, in both cases, $\Vert\bm \Delta\Vert_\infty \leq \alpha$, and so by Assumption~\ref{as:lip}, $\d_{i'}(\bm \Delta)\leq m\cdot\alpha\cdot L$, and $\d_{i'}(\varepsilon_2\cdot \bm e_{\ell})\leq L\varepsilon_{2}$. 
		Finally, recall $\beta := 2mn^2L^3(\alpha + \varepsilon_2)$, and so we have
		\begin{align*}
	\d_{i'}(\bm y_{i'})&\leq \d_{i'} \left((\bm x^k_+)_{i'}\right) - 2mnL^2(\alpha+\varepsilon_2) + mL\cdot \alpha + mnL^2(\varepsilon_2) < \d_{i'} \left((\bm x^k_+)_{i'}\right)\qedhere
		\end{align*}
   \end{proof}
   
   With these two claims, we can now show that $\ip{\bm a^{k-1}, \vd(\bm y)} < n -\delta$ despite $\bm y \in \mathcal F'$, a contradiction.
   	\begin{claim}
   		\label{contradictiony}
   		We have $\ip{\bm a^{k-1}, \vd(\bm y)} < n -\delta$.
   	\end{claim}
   \begin{proof}
	The disutility of all agents decreases from $\bm x^k_+$ to $\bm y$ and since every entry in $\bm a^{k-1}$ is strictly positive (Lemma~\ref{lem:algo-steps-linear} and Observation~\ref{pareto-domination}),
	 we have \begin{equation}\label{eqhere}
	 \ip{ \bm a^{k-1}, \vd(\bm y)} - \ip{\bm a^{k-1}, \vd(\bm x^k_+)}\ \leq\ (\bm a^{k-1})_{i_0} \cdot \left(\d_{i_0}(\bm y_{i_0}) - \d_{i_0}((\bm x^k_+)_{i_0})\right)\ .	 
	 \end{equation}
	 We wish to show that this difference is smaller than $-\delta$.
	We will distinguish the two cases.
	
	\begin{description}
		\item[Case 1.] In this case, we have 
		 \begin{align*}
		  \d_{i_0} (\bm y_{i_0}) - \d_{i_0}((\bm x^k_+)_{i_0}) &= \d_{i_0}((\bm x^k_+)_{i_0} - \alpha \bm e_{j}) - D_{i_0}((\bm x^k_+)_{i_0})\leq -\alpha / L\ ,
		 \end{align*} 
		 by Assumption~\ref{as:lip}.
		 By Claim~\ref{lower-boundona}, $(\bm a^{k-1})_{i_0} \geq \varepsilon_2/ (4n^2L^2)$. 
  		Therefore, the right hand side of~\eqref{eqhere} is at most 		 
		 $-\alpha \varepsilon_2/ 4n^2L^3$. 
		\item[Case 2.] In this case, we have
		  \begin{align*}
		     \d_{i_0} (\bm y_{i_0}) - \d_{i_0}((\bm x^k_+)_{i_0}) &= 0 -\d_{i_0}((\bm x^k_+)_{i_0}) 
		                                          = -(\bm d^k_+)_{i_0}
		                                          = - (\bm d^k)_{i_0}\ ,
		  \end{align*} 
		  by Claim~\ref{adjustcordinatesclaim}.
		  Note that $(\bm d^k)_{i_0} = 1 / (\bm a^{k-1})_{i_0}$, and so 
		  the right hand side of~\eqref{eqhere} is equal to $-1$, which is at most $-\alpha\varepsilon_2/4n^2L^3$.
	\end{description}
	Substituting $\alpha=48n^7mL^6 \varepsilon_1/ \varepsilon_2^3$, we have $\ip{ \bm a^{k-1}, \vd(\bm y)} - \ip{\bm a^{k-1}, \vd(\bm x^k_+)} \leq -12n^5mL^3 (\varepsilon_1/ \varepsilon_2^2) < - 9n^5mL^3 (\varepsilon_1/ \varepsilon_2^2) = -\delta$. 
\end{proof}
Therefore, we have $\bm y \in \mathcal F'$, but $\ip{\bm a^{k-1}, \vd(\bm y)} < n -\delta$, which is a contradiction. This implies that no chores are significantly over-allocated.
Thus, we have proven the following:

\begin{lemma}\label{appx-KKT-point-2-lambda-over}
For all $j \in [m]$, we have $\sum_{i \in [n]} (\bm x^k_{+})_{ij} \leq 1 + \beta$, where $\beta = 2mn^2L^3 (\alpha + \varepsilon_2)$ and $\alpha = 48n^7mL^6 \varepsilon_1/ \varepsilon_2^3$.
\end{lemma}

We now have everything we need to prove the following:
\begin{theorem}
	\label{correctness-main-lemma}
	Algorithm~\ref{alg2:KKT} returns a $(\lambda, \gamma, \delta)$-KKT point with $\lambda =1 + 3mn^2L^3(\alpha + \varepsilon_2)$, where $\alpha = 48n^7mL^6 \varepsilon_1 / \varepsilon_2^3$, $\gamma = 1+\varepsilon_3 $ and $\delta = 9n^5mL^3  \frac{\varepsilon_1}{\varepsilon_2^2}$.
\end{theorem}
\begin{proof}
	If the algorithm returns at step~\ref{alg2:stopping-hyp}, then $(\bm a^k, \bm x^k_{+}, \bm d^k_{+})$ is a $(1 + n^2L\varepsilon_1, 1 + \varepsilon_3, 6n^2L^2 \varepsilon_1 / \varepsilon_2^2)$-KKT point, by Lemma~\ref{appx-KKT-point-1}. 
	If the algorithm returns at step~\ref{alg2:returnclosepoint}, then $(\bm a^{k-1}, \bm x^k_{+}, \bm d^k_{+})$ is a $(1 + 3mn^2L^3(\alpha + \varepsilon_2), 1 + \varepsilon_3, 9n^5mL^3 \varepsilon_1 / \varepsilon_2^2)$-KKT point where $\alpha = 48n^7mL^6 \varepsilon_1 / \varepsilon_2^3$, by Lemmas~\ref{appx-KKT-point-2-gammadelta}, \ref{appx-KKT-point-2-lambda-under}, and~\ref{appx-KKT-point-2-lambda-over}. 
	As $n^3 m^3 L^3 \varepsilon_1 \ll \varepsilon_2$, 
	we have $ 1 + 3mn^2L^3(\alpha + \varepsilon_2) >1 + n^2L \varepsilon_1$  and thus the point returned by Algorithm~\ref{alg2:KKT} is $(1 + 3mn^2L^3(\alpha + \varepsilon_2), 1 + \varepsilon_3, 9n^5mL^3 \varepsilon_1 / \varepsilon_2^2)$-KKT.   
\end{proof}

\subsection{Polynomially Bounding the Number and Running Time of Iterations}\label{sec:gen-alg-bounds}

In this section, we show that in polynomially many iterations, Algorithm~\ref{alg2:KKT} finds the $(1 + 3mn^2L^3(\alpha + \varepsilon_2), 1 + \varepsilon_3, 9n^5mL^3 \varepsilon_1 / \varepsilon_2^2)$-KKT point of Lemma~\ref{correctness-main-lemma}. 
The proof follows exactly the proof in the setting with linear disutilities (the proof of Lemma~\ref{lem:inc2}). 
One can argue that (a) the log-NSW $\mathcal L$ is always increasing throughout Algorithm~\ref{alg:KKT}, and (b) it increases additively by $\varepsilon_3 / 16n^2$ every time $\logd (\bm d^{k+1}, \bm d^{k}_{+}) \geq \varepsilon_3$. Therefore, the total number of iterations of the algorithm is bounded by $\textup{poly}(n,m, 1/ \varepsilon_3)$.

\begin{lemma}
 \label{polynomialmanyiterations}
  After $\textup{poly}(n,m, 1/ \varepsilon_3)$ iterations, Algorithm~\ref{alg2:KKT} returns a $(1 + 3mn^2L^3(\alpha + \varepsilon_2), 1 + \varepsilon_3, 9n^5mL^3 \varepsilon_1 / \varepsilon_2^2)$-KKT point, where $\alpha = 48n^7mL^6 \varepsilon_1 / \varepsilon_2^3$. 
\end{lemma}	

It remains therefore to show that each iteration can be efficiently computed. 
The main difficulty lies in finding the approximate nearest point $\bm d^k_+$, or rather $\bm x^k_+$, \textit{i.e.}\ to explain the algorithm for \textsc{nearest-point}, from line~\ref{alg2:nearestpoint} of Algorithm~\ref{alg2:KKT}.
The remaining steps of the algorithm and of the \textsc{adjust-coordinates} procedure are polynomial.

Recall that given a scalar $\varepsilon_1$ and a point $\bm d^k$, the subroutine $\textsc{nearest-point}(\cdot , \cdot)$ returns a point $\bm x^k_+$ such that $|| \bm x^k_+ - \bm x^k_{*}||_2 \leq \varepsilon_1$ and $||\vd (\bm x^k_+) -  \bm d^k_* ||_2 \leq \varepsilon_1$, where $\bm d^k_*$  is the nearest point in $\mathcal D + \mathbb{R}^n_{\geq 0}$ to $\bm d^k$ and $\bm x^k_* \in \mathcal F'$ is a pre-image of 	$\bm d^k_*$. Also note that we have an additional requirement on the nearest point that each coordinate should be an integral multiple of $\varepsilon_1/ (\textup{poly}(n,m))$. However, this can be implemented by rounding the nearest approximate point that we find in polynomial time without increasing $|| \bm d^k_+ - \bm d^k ||_2$ and $||\bm x^k_+ - \bm x^k||_2$ significantly (there will only be additive errors of $\varepsilon_1/ (\textup{poly}(n,m))$). Therefore, the main bottleneck is in finding the approximate nearest point. We focus mainly on this now.

We will implement this with the following convex program, which returns simultaneously $\bm x^k_*$ and $\bm d^k_*$ as its solution.

\begin{equation}
\label{convex-pgm-nearest-point}
\begin{array}{ll@{}ll}
\text{minimize}  & \displaystyle \sum\limits_{i \in [n]}^{ } (\beta_i)^2 \\
\text{subject to}& \displaystyle  \sum_{i \in [n]} (\bm z)_{ij} \geq 1, & &\forall j \in [m]\\ 
&   (\bm  z)_{ij} \geq  0,  & &\forall i \in [n], \forall j \in [m]\\
& D_i(\bm z_i) - (\bm d^k)_i - \beta_i \leq 0, & &\forall i \in [n],
\end{array}
\end{equation}

The objective function is clearly convex, 
and we show in Claim~\ref{convexityofprogram} in Appendix~\ref{app:tech-results} that the constraints are convex as well.

We now prove that Program~\ref{convex-pgm-nearest-point} is correct, \textit{i.e.}\ its solution gives the nearest point in $\mathcal D + \mathbb{R}^n_{\geq 0}$ to $\bm d^k$.

\begin{lemma}
	\label{convex-program-correctness}
	Let $(\bm z^*, \bm \beta^*)$ be an exact solution to the convex program~\ref{convex-pgm-nearest-point}. Then $\bm z^* \in \mathcal F'$ and $ \vd (\bm z^*)$ is a nearest point in $\mathcal D +\mathbb{R}^n_{\geq 0}$ to $\bm d^k$.
\end{lemma}
\begin{proof}
	Since $(\bm z^*, \bm \beta^*)$ satisfies the feasibility constraints, we have $(\bm z^*)_{ij} \geq 0$ for all $i \in [n]$, $j \in [m]$ and $\sum_{i \in [n]} (\bm z^*)_{ij} \geq 1$ for all $j \in [m]$, implying that $\bm z^* \in \mathcal F'$. Now, it remains to show that $\vd (\bm z^*) = \bm d^k_*$.
	
	Let $\OPT$ be the minimum value of the objective function in~\eqref{convex-pgm-nearest-point} achieved by any feasible solution. We wish to show that $\OPT = ||\bm d^k_* - \bm d^k||_2^2$. First note that $\OPT \leq ||\bm d^k_* - \bm d^k ||_2^2$, since it is feasible to set $\bm z \gets \bm x^k_*$ and 
	\begin{align*}
	&\beta_i \gets D_i(\bm z_i) - (\bm d^k)_i = D_i((\bm x^k_*)_i) - (\bm d^k)_i = (\bm d^k_*)_i - (\bm d^k)_i\\
	&\implies \sum_{i\in[n]}(\beta_i)^2 = \Vert\bm d^k_*-\bm d^k\Vert_2^2\ .
	\end{align*}
Now, we show that $\OPT \geq  || \bm d^k_* - \bm d^k ||_2^2$.  
Note that if $\bm d^k \in \mathcal D + \mathbb{R}^n_{\geq 0}$, then $\bm d^k_* =\bm d^k$ and thus $\OPT$ will be trivially larger than $|| \bm d^k_* - \bm d^k ||_2^2 = 0$. 
So we only focus on the case when $\bm d^k \notin \mathcal D + \mathbb{R}^n_{\geq 0}$. 
Suppose for a contradiction that $\OPT < ||\bm d^k_* - \bm d^k||_2^2 $ and there exists a feasible $(\bm z, \bm \beta)$ such that $\sum_{i \in [n]}(\beta_i)^2 < ||\bm d^k_* - \bm d^k||_2^2$.  
This implies that $\bm d^k + \bm \beta \notin \mathcal{D} + \mathbb{R}^n_{\geq 0}$. 
Therefore, there exists a hyperplane $\ip{\bm a, \bm y} = n$ that separates $\mathcal D + \mathbb{R}^n_{\geq 0}$ and the point $\bm d^k + \bm \beta$.
Since $\bm z$ lies in $\mathcal D+\mathbb R^n_{\geq 0}$, we have then by definition, 
\begin{align*}
	\ip{\bm a, (\bm d^k+\bm \beta)} &< n \leq \ip{\bm a, \vd(\bm z)}\\
	\implies 0&<\ip{\bm a, \vd(\bm z)-\bm d^k-\bm\beta}
	\ .
\end{align*}
By the last constraint in program~\eqref{convex-pgm-nearest-point}, 
each coordinate of $\vd(\bm z)-\bm d^k-\bm\beta$ is non-positive, and so for this inequality to hold, $\bm a$ must have a strictly negative entry. 
However, since $\mathcal D + \mathbb R^n_{\geq 0}$ is unbounded in the positive directions, 
no hyperplane of the form $\ip{\bm a,\cdot}\geq n$ can be supporting unless $\bm a\geq \bm 0$, a contradiction.

	Therefore, at the optimum, $||\bm \beta^*||_2^2 = ||\bm d^k_* - \bm d^k||_2^2$. 
	Since $(\bm z^*, \bm \beta^*)$ is a feasible point we have $D_i ((\bm z^*)_i) - (\bm d^k)_i \leq (\bm \beta^*)_i$ for all $i \in [n]$, implying that $||\vd (\bm z^*) - \bm d^k ||_2^2 \leq ||\bm d^k_* - \bm d^k||_2^2$. Since $\bm d^k_*$ is a closest point in $\mathcal D + \mathbb{R}^n_{\geq 0}$ to $\bm d^k$, and $\bm z^* \in \mathcal D + \mathbb{R}^n_{\geq 0}$, we have $||\vd (\bm z^*) - \bm d^k ||_2^2 =  ||\bm d^k_* - \bm d^k||_2^2$, implying that $\vd (\bm z^*)$ is also a closest point in $\mathcal D + \mathbb{R}^n_{\geq 0}$ to $\bm d^k$. 
\end{proof}

Let $(\bm z^*, \bm \beta^*)$ be the optimum solution to the program~\ref{convex-pgm-nearest-point}. We just need to find an approximate solution to the convex program $(\bm z^+, \bm \beta^+)$ such that $||(\bm z^+, \bm \beta^+) - (\bm z^*, \bm \beta^*)||_2 \leq \varepsilon_1/L$, as this would give us the desired bounds in the allocation as well as the resulting disutility vector (follows from the Lipschitz condition of the disutility functions mentioned in Assumption~\ref{as:lip}). This can be determined in polynomial time by interior point algorithms. This brings us to the main lemma of this section.

\begin{lemma}
	\label{polynomialtimeiteration}
	Given a scalar $\varepsilon_1$ and a point $\bm d^k$, the subroutine $\textup{NEAREST-POINT}(\cdot , \cdot)$ returns a point $\bm x^k_+$ such that $|| \bm x^k_+ - \bm x^k_{*}||_2 \leq \varepsilon_1$ and $||\vd (\bm x^k_+) -  \bm d^k_* ||_2 \leq \varepsilon_1$, where $\bm d^k_*$  is the nearest point in $\mathcal D + \mathbb{R}^n_{\geq 0}$ to $\bm d^k$ and $\bm x^k_* \in \mathcal F'$ is a pre-image of 	$\bm d^k_*$, in time $\textup{poly}(n,m, \log (L / \varepsilon_1))$. Additionally, each coordinate of $\bm d^k_+$ and $\bm x^k_+$ is an integral multiple of $\varepsilon_1/ (\textup{poly}(n,m))$. 
\end{lemma}
\begin{proof}
	Many interior-point methods exist to solve this program, including the ellipsoid method, which can efficiently find a near-optimal point in $\textup{poly}(n,m,\log(1/\varepsilon_1))$ time, given efficient separation oracles~\cite{conv-opt}.
	The constraint region needs to be bounded for these methods to work, but we can use
	the correctness of Algorithm~\ref{alg2:KKT} (Theorem~\ref{correctness-main-lemma}) to upper-bound the allocations, and Assumption~\ref{as:lip} to bound the feasible $\beta$ values.
	
	We can easily determine which constraint is violated in the program~\eqref{convex-pgm-nearest-point}, 
	but we need to return a hyperplane if the violated constraint is one of the $\d_i(\bm z_i)-(\bm d^k)_i-\beta_i\geq 0$ constraints.
	We have assumed access to the partial derivatives of the $\d_i$'s, and we therefore have access to the gradient of these constraints.
	A ``good enough'' candidate point at which to take the supporting hyperplane can be found by approximating the nearest feasible point using unconstrained optimization with barrier functions~\cite{conv-opt}, 
	which gives separation oracles to implement a step of the ellipsoid method.
\end{proof}

 To finish arguing that each iteration of Algorithm~\ref{alg2:KKT} can be implemented in polynomial time, we need to show that bit size of $\bm d^{k+1}$ computed at the end of each iteration does not grow too much. In particular, there should be at most an additive polynomial increase in the bit-size of $\bm d^{k+1}$ from $\bm d^k$. To this end, note that $\bm d^{k+1} = \frac{\bm 1}{\bm a^k} = \frac{\ip{\bm d^k_+ - \bm d^k, \bm d^k_+}} {n} \cdot (\bm d^k_+ - \bm d^k)$. Note that such an operation can only cause a polynomial additive increase the bit size of $\bm d^{k+1}$ from $\bm d^k$, as the bit-size of $\bm d^k_+$ is polynomially bounded (recall that we enforce every coordinate of $\bm d^k_+$ to be an integral multiple of $\varepsilon_1/ (\textup{poly}(n,m))$) . Since the total number of iterations is $\textup{poly}(n,m,1 /\varepsilon_3, \log (L))$, we have that the bit-size of $\bm d^{k+1}$ at the end of each iteration is $\textup{poly}(n, m , 1/ \varepsilon_3, \log( L / \varepsilon_1))$. Therefore, the running time of each iteration is polynomial. We can now bound the running time of the algorithm.

\begin{theorem}
	\label{apprx-KKT-generalmain}
	 In $\textup{poly}(n,m,1 / \varepsilon_3, \log (L / \varepsilon_1))$ time, we can determine a $(1 + 3mn^2L^3(\alpha + \varepsilon_2), 1 + \varepsilon_3, 9n^5mL^3 \varepsilon_1 / \varepsilon_2^2)$-KKT point where $\alpha = 48n^7mL^6 \varepsilon_1 / \varepsilon_2^3$.
\end{theorem}

\begin{proof}
	Lemma~\ref{polynomialmanyiterations} states that in $\textup{poly}(n,m, 1 / \varepsilon_3)$ many iterations, Algorithm~\ref{alg2:KKT} determines a $(1 + 3mn^2L^3(\alpha + \varepsilon_2), 1 + \varepsilon_3, 9n^5mL^3 \varepsilon_1 / \varepsilon_2^2)$-KKT point where $\alpha = 48n^7mL^6 \varepsilon_1 / \varepsilon_2^3$. Lemma~\ref{polynomialtimeiteration}, states that each iteration of Algorithm~\ref{alg2:KKT} can be implemented in $\textup{poly}(n,m, \log (L / \varepsilon_1))$ time. Thus, Algorithm~\ref{alg2:KKT} finds the approximate KKT point in $\textup{poly}(n,m, 1/ \varepsilon_3, \log (L / \varepsilon_1))$ time. 
\end{proof}

By setting appropriate values for $\varepsilon_1$, $\varepsilon_2$ and $\varepsilon_3$, we have the following corollary.

\begin{corollary}
	\label{approx-KKT-generalcorollary}
	In $\textup{poly}(n,m, 1/\varepsilon, \log (L))$ time, we can determine a $(1 + \frac{\varepsilon}{2^{\textup{poly}(n,m)}}, 1 + \varepsilon, \frac{\varepsilon}{2^{\textup{poly}(n,m)}})$-KKT point.
\end{corollary}

\begin{proof}
	We set $\varepsilon_3 = \varepsilon$. Thereafter, we set $\varepsilon_1 = \varepsilon / (2^{f(n) \cdot f(m)} \cdot f(L))$ where $f(\cdot)$ is a polynomial of sufficiently large constant degree. 
	The lower-bound on the degree will become clear by the end of this proof. 
	Then, set $\varepsilon_2 = n^4m^3L^3 (\varepsilon_1)^{\tfrac{1}{6}}$. 
	Note that we have $\varepsilon_2 \gg n^3m^3L^3 \varepsilon_1$ and $\varepsilon_3 \geq \frac{\varepsilon_2}{2^{\textup{poly}(n,m)}}$ as required by all the proofs so far. 
	
	We can show the following bounds in terms of $\varepsilon$:
	\begin{itemize}
		\item  $ 1 + 3mn^2L^3 \cdot (\alpha + \varepsilon_2) \leq 1 + (\varepsilon^{{1}/{6}})/ (2^{\textup{poly}(n,m)})$: We first note that $\alpha \leq \varepsilon_1^{1/2}$. We have 
		\begin{align*}
		 \alpha = 48n^7mL^6 \cdot \frac{\varepsilon_1}{\varepsilon_2^3}&= 48n^7mL^6 \cdot \frac{\varepsilon_1}{n^{12} \cdot m^9 L^9 \varepsilon_1^{{1}/{2}}}\\
		        &=\frac{48}{n^5m^8L^3} \cdot \varepsilon_1^{{1}/{2}}\leq \varepsilon_1^{1/2}\ .
		\end{align*}
		Substituting the upper-bound on $\alpha$ in $1 + 3mn^2L^3 \cdot (\alpha + \varepsilon_2)$, we have   
		\begin{align*}
		 1 + 3mn^2L^3 \cdot (\alpha + \varepsilon_2) &\leq 1 + 3mn^2L^3( (\varepsilon_1)^{1/2} + n^4m^3L^3 (\varepsilon_1)^{1/6})\\
		                                             &\leq 1 + 3mn^2L^3 \cdot (2n^4m^3L^3 (\varepsilon_1)^{1/6})\\
		                                             &\leq 1 + 6n^6m^4L^6 \cdot (\varepsilon_1)^{1/6}\\
		                                             &\leq 1 + \frac{\varepsilon^{{1}/{6}}}{2^{\textup{poly}(n,m)}} &\text{(for a sufficiently large $f(\cdot)$)}. 
		\end{align*} 
		\item $1 + \varepsilon_3 = 1 + \varepsilon$.
		\item $9n^5mL^3 \cdot \varepsilon_1/ \varepsilon_2^2 \leq (\varepsilon^{2/3}) / (2^{\textup{poly}(n,m)})$: We have 
		 \begin{align*}
		 9n^5mL^3 \cdot \frac{\varepsilon_1}{\varepsilon_2^2} = 9 n^5mL^3 \cdot \frac{\varepsilon_1}{n^8m^6L^6 \varepsilon_1^{\tfrac{1}{3}}}
		    &= \frac{9}{n^3m^5L^3} \cdot \varepsilon_1^{2/3}\leq \frac{\varepsilon^{{2}/{3}}}{2^{\textup{poly}(n,m)}}.
		 \end{align*}
	\end{itemize}
   Substituting these bounds in the statement of Theorem~\ref{apprx-KKT-generalmain}, we can conclude that Algorithm~\ref{alg2:KKT} returns a $(1 + \frac{\varepsilon^{{1}/{6}}}{2^{\textup{poly}(n,m)}}, 1 + \varepsilon, \frac{\varepsilon^{{2}/{3}}}{2^{\textup{poly}(n,m)}})$-KKT point for our choice of $\varepsilon_1$, $\varepsilon_2$ and $\varepsilon_3$. Note that the running time now is 
   \begin{align*}
     \textup{poly}(n,m,1/ \varepsilon_3, \log (L / \varepsilon_1)) &= \textup{poly}\bigg(n,m, 1 / \varepsilon, \log \big(L \cdot f(L) \cdot 2^{(f(n) \cdot f(m))} \cdot  1/ \varepsilon\big) \bigg)\\
                                                                   &= \textup{poly}(n,m, 1/ \varepsilon, f(n) \cdot f(m), \log (L \cdot  f(L)))\\
                                                                   &= \textup{poly}(n,m, 1/ \varepsilon, \log (L)) &\!\!\!\!\!\!\!\!\!\!\!\!\!\!\!\!\!\!\!\!\!\!\!\!\!\!\!\!\!\!\!\!\text{(as $f(\cdot)$ is polynomial)}
   \end{align*}
\end{proof}

\subsection{Finding Approximately Competitive Equilibria in Polynomial Time}

We can now combine the results of the previous sections to show that \eCEEI can be found in polynomial time.
\begin{theorem}
	\label{gen:mainthm}
	Given black-box access to 1-homogeneous and convex disutilities $\d_{1},\,\dotsc,\,\d_{n}$ 
	all satisfying Assumption~\ref{as:lip} with constant $L$,
	and also to their partial derivatives, 
	Algorithm~\ref{alg2:KKT}, finds an $\varepsilon^{1 / 6}$-CEEI in time polynomial in $n$, $m$, $1/\varepsilon$, and $\log(L)$.
\end{theorem}	

\begin{proof}
	From Theorem~\ref{thm:gen-approx} we know that a $(\lambda, \gamma, \delta)$-KKT point gives a $\mathit{max}(3 ( \gamma-1) + 5 \delta, \lambda -1)$- CEEI. Substituting the values of $\lambda$, $\gamma$ and $\delta$ as in Corollary~\ref{approx-KKT-generalcorollary}, we have that in $\textup{poly}(n,m,1/ \varepsilon, \log(L))$-time,   Algorithm~\ref{alg2:KKT}, returns a $\mathit{max}(3 \varepsilon + \frac{5\varepsilon^{2 / 3}}{2^{\textup{poly}(n,m)}}, \frac{\varepsilon^{1/ 6}}{2^{\textup{poly}(n,m)}})$-CEEI, which is a $\varepsilon^{1/6}$-CEEI. 
\end{proof}

We remark that although our overall approximation of CEEI is inverse-polynomial, we satisfy the condition~\ref{P2} and~\ref{P3} in Definition~\ref{def:CE} with an inverse-exponential $\varepsilon$.

\begin{remark}
	\label{extendingtoCES}
	We briefly remark how to make the algorithm work with a weaker assumption than the one in Assumption~\ref{as:lip}, i.e., for all $i \in [n]$, we have $|D_i(\bm x+ \delta \cdot \bm e_j) - D_i( \bm x)| \geq \frac{1}{L} \cdot \mathit{min}\big( \delta, \delta^k \big)$ for $k \in \textup{poly}(n,m)$. In this case, the re-adjustments done in steps~\ref{alg2:pd} such that $\bm d^k_+$ Pareto-dominates $\bm d^k$ should increase $\bm x^k_+$ by $2L \varepsilon_1^{1 / k}$ instead of $2L\varepsilon_1$. Similarly, all the bounds for $\lambda$ and $\delta$ in Theorem~\ref{apprx-KKT-generalmain} that are functions of $\varepsilon_1$ and $\varepsilon_2$ will be functions of $\varepsilon_1^{1 / k}$ and $\varepsilon_2^{1 / k}$. Since, we have the flexibility to choose inverse exponential values for both $\varepsilon_1$ and $\varepsilon_2$ to get an $1 / (\textup{poly}(n,m))$-CEEI, we can tolerate $k \in \textup{poly}(n,m)$.    
\end{remark}

\section{Extending to Mixed Linear Disutilities and Un-equal Incomes.}\label{sec:extensions}
We briefly argue in this section that our algorithms extend to the setting of {\em mixed manna} with linear disutilities, where some agents likes some items and dislike others, and to un-equal income where agents have different importance/weights modeled as different income requirements. In the latter case the resulting allocation is known as {\em competitive equilibrium (CE)} allocation.

\subsection{Mixed Manna with Linear Disutilities}
We say the instance is of {\em mixed manna} if each agent may value some items positively (goods), and other items negatively (chores). 
Formally, linear {\em utility} functions in this setting are represented by $U_i(\bm x_i)=\sum_{j} U_{ij} x_{ij}$ for agent $i$, where $U_{ij}$ is positive if $j$ is a good for agent $i$, and $U_{ij}$ is negative if $j$ is a chore for $i$, interpreted $U_{ij} = -\d_{ij}$.

In this case, \cite{BogomolnaiaMSY17} characterized the instanced into three categories, intuitively goods-heavy (``positive''), chores-heavy (''negative''), and null. To formalize this, let us first divide the agents into two categories:
\[
N_+ = \{i\ |\ \exists j, U_{ij}>0\}\ ,\quad \quad N_-=\{i\ |\ \forall j, U_{ij}\le 0\}\ .
\]

We have the following three cases:

\begin{description}
\item[(positive)] There exists a feasible allocation $\bm x\in \mathcal F$ such that all agents in $N_+$ get strictly positive utilities, and those in $N_-$ get zero utility.
Formally, $U_i(\bm x_i)>0,\ \forall i\in N_+$ and $U_i(\bm x_i)=0,\ \forall i \in N_-$. 
\item[(null)] The positive case is not possible, but there exists a feasible allocation $\bm x\in \mathcal F$ such that all agents get zero utility, {\em i.e.,} $U_i(\bm x_i)=0,\ \forall i$.
\item[(negative)] Every feasible allocation gives strictly negative utility to some agent.
\end{description}

Given an instance it is easy to check which category it belongs to using linear programming. 
Bogomolnaia et al.~\cite{BogomolnaiaMSY17} characterize the set of CEEI as follows:

\begin{theorem}\cite{BogomolnaiaMSY17}\label{thm:bmsy-mixed}
Given an instance $I=(U_1,\dots,U_n)$ with mixed manna,\begin{enumerate}\itemsep0pt
	\item If $I$ is a positive instance, then every agent {\em spends} at most one unit of money, and an allocation is a CEEI if and only if it maximizes the product $\prod_{i\in N_+} U_i(\bm x_i)$, and agents in $N_-$ get zero utility.
	\item If $I$ is a null instance, then an allocation is a CEEI if and only if all agents get zero utility.
	\item If $I$ is a negative instance, then an allocation is a CEEI if and only if it is a local minimum (KKT point) of the $\prod_{i} |U_i(\bm x_i)|$ with $U_i(\bm x_i)<0$ for all $i$.
\end{enumerate}
\end{theorem}

Using the above characterization, we note that the positive case can be solved using the Eisenberg-Gale~\cite{EG} convex program, 
by maximizing $\sum_{i\in N_+} \log(U_i(\bm x_i))$, subject to feasibility and assigning zero to $N_-$.
		\footnote{Since CEEI must be Pareto-optimal, 
			if for some item $j$, $U_{ij}>0$ and $U_{i'j}\le 0$, 
			then item $j$ will never be allocated to agent $i'$. 
			Therefore, without loss of generality we may eliminate the $x_{i'j}$ variable. 
			After all such modifications, we can divide items into two sets: 
			goods $G$ that are non-negatively valued by all the agents, 
			and chores (bads) $B$ that are non-positively valued by all the agents. 
			Then the feasibility constraints should be $\sum_i x_{ij} \le 1,\ \forall j\in G$ and $\sum_i x_{ij} \ge 1,\ \forall j\in B$ 
			to ensure correct sign for the prices that comes from the dual variables.}
In the null case, any allocation that gives zero utility to all the agent is a CEEI and can be computed through a linear feasibility program. 
The negative case is similar to the purely-chores setting that we establish in this paper. 
The rest of this section explains how to extend our algorithms to handle this case.

To follow the notation used in the previous sections, 
we will return to using disutility functions $\d_i(\bm x_i)  = - U_i(\bm x_i)$, by setting $\d_{ij}=-U_{ij}$ for all $i$ and $j$. 
Now as per Theorem \ref{thm:bmsy-mixed}, in the negative case we want to find a KKT point of $\prod_{i} \bm d_i$ subject to $\bm d \in \mathcal D$ and $d_i>0$, $\forall i$. 
This is exactly what we compute in Section~\ref{sec:CE}, 
but we must now allow for the feasible region $\mathcal D$ to contain disutility profiles with negative entries. 

We show here that our theorems do not require that all feasible disutility profiles be positive, but instead need only the returned allocation to have positive disutilities.
We handle this without any modification by noting that Lemma~\ref{lem:algo-steps-linear} was written to not require $\mathcal D$ to lie in $\mathbb R_{\geq 0}^n$. 
Thus, so long as we can find an initial (infeasible) disutility profile $\bm d^0$ whose entries are all positive and sufficiently large, every other detail will go through with positive disutility.

When disutilities are linear, we must modify Claim~\ref{claim:d0} to account for the fact that some $\d_{ij}$'s may be negative. 
Unfortunately, na\"ive constructions like that of Claim~\ref{claim:d0} will not work anymore.
However, $\bm d^0=\bm 0$ is a valid starting point, but it simply means that we cannot use the objective function at the starting point to bound the number of iterations.

We show that one can always find an infeasible $\bm d^0$ with polynomial bit complexity as follows: 
We have that $\mathcal D$ is a linear polytope and $\bm 0 \notin \mathcal D$. 
For a $\delta>0$, check if $\delta \bm 1 \in \mathcal D$ or not. 
If not then set $\bm d^0 = \delta \bm 1$ and then $- \mathcal L(\bm d^0) =n \log(1/\delta)$. 
Otherwise, the line joining $\bm 0$ and $\delta \bm 1$ must intersect $\mathcal D$, 
and this intersection has to be a point with polynomial bit complexity, so long as $\delta$ is sufficiently large. 
Therefore, we can do a binary search on the line between $\bm 0$ and $\delta \bm 1$ to find an infeasible $\bm d^0$ of polynomial bit complexity.

The upper bound on $\mathcal L(d)$ from Lemma \ref{lem:noiter-lin} works as is, and thereby we get that the number of iterations remain polynomial.

\subsection{Un-equal Incomes: Competitive Equilibrium (CE)}
The model with unequal incomes is formally defined as follows: 
each agent has a disutility function $\d_i$ as before, and an income level $\eta_i>0$. 
We will show that our algorithm extends to this setting. And the running time remains polynomial as far as $\nicefrac{\max_i \eta_i}{\min_i \eta_i}$ is polynomially bounded.

A competitive equilibrium is an allocation $\bm x=(\bm x_1,\,\dotsc,\,\bm x_n)\in \mathcal F$, and a payment vector $\bm p$ such that $\ip{\bm p,\bm x_i}=\eta_i$ for all $i$, 
and every agent minimizes their disutility subject to 
$\ip{\bm p,\bm x_i}\geq \eta_i$. 
Accordingly the only change in Definition \ref{def:CE} of \eCEEI is in Condition \ref{P1}: we now want that all agents' $\eta$-rescaled incomes are approximately same, {\em i.e.,} $(1-\varepsilon) \cdot  \langle \bm z_i, \bm p \rangle/\eta_i \leq \langle \bm z_{i'}, \bm p \rangle/\eta_{i'}$ for any $i,i'$.

To take the weights into account the objective functions changes to minimizing $\prod_i d_i^{\eta_i}$, and accordingly define $\mathcal L(\bm d;\bm \eta):=\sum_{i=1}^n \eta_i\log(d_i)$. 
And then the definition of approximate KKT point will need to be modified: 
instead of requiring that $\gamma^{-1}\leq a_i d_i\leq \gamma$, we instead ask that 
$\gamma^{-1}\leq a_i d_i/\eta_i\leq \gamma$.

Now, in the proof of Theorems~\ref{corr:approx} and~\ref{thm:gen-approx}, we make the following changes. 
When showing condition~\ref{P1} of Definition~\ref{def:CE}, 
we now want to show that $\gamma^{-2} \cdot  \langle \bm z_i, \bm p \rangle/\eta_i \leq \langle \bm z_{i'}, \bm p \rangle/\eta_{i'}$ for any $i,i'$. 
	  Assume otherwise and say we have  $\gamma^{-2} \cdot  \langle \bm z_i, \bm p \rangle/\eta_i > \langle \bm z_{i'}, \bm p \rangle/\eta_{i'}$. 
	  As before, replace the allocation as follows: Construct $\hat{\bm z}$ by setting $\hat{\bm z}_{i'}=\tfrac12\bm z_{i'}$, and $\hat{\bm z}_{i}=\left(1+\tfrac{\ip{\bm z_{i'}, \bm p}}{2\ip{\bm z_{i}, \bm p}}\right)\bm z_{i}$.
	  The same proofs extend, but the contradiction is attained as follows:
\begin{align*}
			\ip{\bm a,\overrightarrow{\d}(\hat{\bm z})} - \ip{\bm a,\overrightarrow{\d}(\bm z)}
			& =-\frac{1}{2} a_{i'} \d_{i'}(\bm z_{i'}) + 
				\frac{\ip{\bm z_{i'}, \bm p}}{2\ip{\bm z_{i}, \bm p}} 
				a_{i}\d_{i}(\bm z_{i}) \\
			& =\eta_{i'}\left(-\frac{1}{2\eta_{i'}} a_{i'} \d_{i'}(\bm z_{i'}) + \frac{\eta_{i}\ip{\bm z_{i'}, \bm p}}{2\eta_i\eta_{i'}\ip{\bm z_{i}, \bm p}} a_{i}\d_{i}(\bm z_{i})\right) \\
			& \leq \eta_{i'}\left(-\frac 12 \gamma^{-1} + \frac{\eta_{i}\ip{\bm z_{i'}, \bm p}}{2\eta_{i'}\ip{\bm z_{i}, \bm p}}\cdot \gamma\right)\\
			& <\eta_{i'}\left(-\tfrac 12 \gamma^{-1}+\tfrac 12 \gamma^{-1}\right)=0\ ,
		\end{align*}

Conditions~\ref{P2} and~\ref{P3} work without modification. 

The stopping conditions and update moves of Algorithm~\ref{alg:KKT} must be changed to ensure that the modified definition of KKT points can be met.
The first modification is to replace the definition of $\bm d^{k+1}$ on line~\ref{alg:hyperplane-move} with $(\eta_1 / a_1,\,\dotsc,\, \eta_n /a_n)$, if we rescale $\bm a$ so that $\ip{\bm a^k,\bm d^k_*}=\sum_i\eta_i$.  
This choice will ensure that small changes between $\bm d^k_*$ and $\bm d^{k+1}$ imply the right $\eta$-rescaled KKT conditions.
Note that, now the stopping condition on line \ref{alg:stopping-hyp-cond} will indeed ensure that the algorithm returns $(1+\epsilon)$-KKT point with respect to the new definition. 

Finally, it remains to show that the number of iterations is still polynomial, though it will depend on the $\eta_i$'s.
We can use Lemma~\ref{lem:inc2} to lower-bound the improvement at each step:
We have $\logd(\bm d^k_*,\bm d^{k+1})>\varepsilon$, 
and therefore $\mathcal L(\bm d^{k+1},\bm 1)-\mathcal L(\bm d^k_*,\bm 1)>\Omega(\epsilon^2/n^2)$.
However, we need to use $\mathcal L(\bm d,\bm \eta)$, not $\mathcal L(\bm d,\bm 1)$, since the latter is not a potential function for the modified algorithm. 
We have, however, 
\[
    \mathcal L(\bm d,\bm 1)\cdot \min_i\eta_i\ \leq \ 
    \mathcal L(\bm d,\bm \eta)\ \leq \ 
    \mathcal L(\bm d,\bm 1)\cdot \max_i\eta_i\ .
\]
This allows us to conclude that
\begin{equation}
    \mathcal L(\bm d^{k+1},\bm \eta) - \mathcal L(\bm d^k_*,\bm \eta) \geq 
    \frac{\varepsilon^2 \cdot \min_{i}\eta_i}{16n^2}
\end{equation}

Now, a proof identical to that of Lemma~\ref{lem:noiter-lin} bounds the number of iterations by 
\begin{align*}
	&\frac{\textup{poly}(n,1/\varepsilon)}{\min_i\eta_i}\cdot \left(\max_{\bm d\in \mathcal D}\mathcal L(\bm d;\bm \eta) - \mathcal L(\bm d^0;\bm \eta)\right)\\
	&\leq \frac{\textup{poly}(n,1/\varepsilon)}{\min_i\eta_i}\cdot  \left(\sum_i\eta_i\log(m\cdot\max_{i,j}\d_{ij}) - \sum_i\eta_i\log(\tfrac m{2n}\min_{i,j}\d_{ij}) \right)\\
	&= \frac{\textup{poly}(n,1/\varepsilon)}{\min_i\eta_i}\cdot  \left(\textstyle\sum_{i=1}^n\eta_i\right)\cdot \log\left(
		\frac{2n\cdot \max_{i,j}\d_{ij}}{\min_{i,j}\d_{ij}}
	\right)\ ,
\end{align*}
which is the same running time as the equal-income setting, up to the $\min \eta_i\cdot \sum \eta_i$ term. 
But the $\eta_i$'s are dimensionless, and therefore without loss of generality, we can assume $\max_i\eta_i=1$, which implies $\sum \eta_i\leq n$, and so we can bound the number of iterations by
\begin{equation}
    \nonumber \textup{poly}(n,1/\varepsilon)\cdot  \frac{\max_i \eta_i}{\min_i\eta_i}
    \cdot \log\left(
		\frac{2n\cdot \max_{i,j}\d_{ij}}{\min_{i,j}\d_{ij}}
	\right)\ .
\end{equation}

Thus, we have shown that when disutilities are linear, the results in this paper extend to CE without too much modification. And the running time guarantees are the same as far as the max to min income ratio is polynomially bounded. 
It remains to argue that the same is true for general disutilities. 
The above modifications are in fact the only ones needed for this more general setting: the proofs of Conditions~\ref{P1}, \ref{P2}, and~\ref{P3} in the definition of approximate KKT all go through the same as they did in the linear case, and the analysis of approximately supporting hyperplane does not play a role in the unequal income.

\bibliographystyle{alpha}
\bibliography{compeq}	
	
\appendix
\clearpage
\section{Technical Proofs}\label{app:tech-results}
\paragraph{Claim~\ref{claim:convex}.}	\textit{\ClaimDPlusIsConvex}
	\begin{proof}
		Let $\bm y,\,\bm y'\in \mathcal D+\mathbb R^n_{\geq 0}$, such that $\bm y=\bm d+\bm \Delta$ for some $\bm d\in \mathcal D$ and $\bm \Delta\in \mathbb R^n_{\geq 0}$. Define $\bm d'$ and $\bm \Delta'$ similarly for $\bm y'$.
		Let $0<\lambda<1$, and let $\bar\lambda:=1-\lambda$. 
		Letting $\overrightarrow{\d}(\bm x):=\left(\d_1(\bm x_1),\,\dotsc,\,\d_n(\bm x_n)\right)$,
		let $\bm x,\, \bm x'\in \mathcal F$ be such that $\overrightarrow{\d}(\bm x) = \bm d$, and $\overrightarrow{\d}(\bm x')=\bm d'$.
		Since $\mathcal F$ is a linear polytope, and is therefore convex, 
		$\lambda \bm x + \bar\lambda \bm x'\in \mathcal F$, and so $\overrightarrow{\d}(\lambda \bm x + \bar\lambda \bm x')\in \mathcal D$.
		Since the $\d_i$'s are convex, we have that $\d_i(\lambda \bm x_i+\bar \lambda\bm x'_i)\leq \lambda \d_i(\bm x_i)+\bar\lambda \d_i(\bm x'_i)$ for all $i$.
		Thus, component-wise, we have 
		\[	
			\overrightarrow{\d}(\lambda \bm x+\bar \lambda\bm x')\leq \lambda \overrightarrow{\d}(\bm x)+\bar\lambda \overrightarrow{\d}(\bm x') = 
			\lambda \bm d+\lambda \bm d'
		\]
		Thus, there must exist a $\bm \Delta''\in \mathbb R_{\geq 0}^n$ such that
		\[
			\lambda \bm y+\bar\lambda\bm y' = \lambda \bm \Delta + \bar\lambda\bm \Delta' + \overrightarrow{\d}(\lambda \bm x+\bar\lambda\bm x')+\bm\Delta'' \in \mathcal D+\mathbb R_{\geq 0}^n.\qquad \qedhere
		\]
	\end{proof}
	
	\paragraph{Lemma~\ref{lem:algo-steps-linear}.} {\em\lemmaAboutAlgoStepsLinear}
	\begin{proof}
		Informally, 1.~holds by the KKT conditions of the minimization problem, and the geometry of the $\ell_2$ norm, though we give here a more direct proof. 
		The idea is to assume that some feasible point lies on the wrong side of the supporting hyperplane, and contradict the minimality of the distance of $\bm d^k_*$.
		
		Let $\bm a:= \bm d^k_*-\bm d^k$ (\textit{i.e.}\ $\bm a^k$ before rescaling), and let $\bm y$ be any feasible point in $\mathcal D+\mathbb R^n_{\geq 0}$.
		Define $\bm u:= \bm y-\bm d^k$, which can be decomposed as $\bm u = (1-\alpha) \bm a + \beta \bm v$, with $\ip{\bm v,\bm a} = 0$.
Assume without loss of generality $\Vert \bm a\Vert_2 = \Vert \bm v\Vert_2$.

For any $0\leq \lambda \leq 1$, the vector $\bm d^k + \lambda \bm u + (1-\lambda)\bm a$ is feasible, by convexity. Furthermore, its squared distance from $\bm d^k$ is given by 
\begin{align*}
    \Vert (\lambda \bm y+(1-\lambda)\bm d^k_*)-\bm d^k\Vert^2_2&= \Vert \lambda \bm u + (1-\lambda)\bm a \Vert_2^2 \\&= 
    \Vert (\lambda (1-\alpha) + 1 - \lambda)\bm a + \lambda\beta \bm v \Vert_2^2\\
    &= (1-\lambda\alpha)^2 \Vert \bm a\Vert_2^2 + (\lambda\beta)^2\Vert \bm v\Vert_2^2\\
    &= 1 - 2\lambda\alpha + \lambda^2(\alpha^2+\beta^2) \Vert \bm a\Vert_2^2
\end{align*}

Note that $\ip{\bm y,\bm a} - \ip{\bm d^k_*,\bm a} = ((1-\alpha)-1)\Vert \bm a\Vert_2^2 = -\alpha\Vert \bm a\Vert_2^2$, 
which is negative if and only if $\alpha>0$. 
Thus, assume $\alpha>0$ for a contradiction, and set $\lambda := \alpha / (\alpha^2 + \beta^2)$, a positive number.
Furthermore, note that $\Vert \bm d^k_*-\bm d^k\Vert^2_2 = \Vert \bm a\Vert^2_2$.
Therefore, 
\begin{align*}
    \frac{\Vert (\lambda \bm y+(1-\lambda)\bm d^k_*)-\bm d^k\Vert^2_2}{\Vert \bm d^k_*-\bm d^k\Vert^2_2} &= 
    1 - \frac{2\alpha^2}{\alpha^2+\beta^2} + \frac{\alpha^2}{\alpha^2+\beta^2} < 1
\end{align*}

If $\lambda\leq 1$, this is a contradiction of the minimality of $\bm d^k_*$, since $\lambda \bm y+(1-\lambda)\bm d^k_*$ must be a feasible point.
If $\lambda\geq 1$, we observe that distance from $\bm d^k$ is a convex function, and therefore it must be smaller at $\bm y$, which is a convex combination of $\lambda \bm y+(1-\lambda)\bm d^k_*$ and $\bm d^k_*$, also a contradiction.

		Part 2.\ of the statement holds by the geometry of the region, and part 1. 
		Suppose for a contradiction that for some agent $i$, $(\bm d^k_*)_i<(\bm d^k)_i$.
		Then setting the $i$-th coordinate of $\bm d^k_*$ to $(\bm d^k)_i$ will move it strictly closer to $\bm d^k$. 
		Furthermore, this move is in the $\mathbb R^n_{\geq 0}$ direction, and thus the new point will lie in the set, contradicting the minimality of $\bm d^k_*$.
		Therefore, coordinate-wise, $\bm d^k_*\geq \bm d^k$, ensuring that the former has strictly positive entries. 
		This also proves that the entries of $\bm a^k$ are non-negative. 
		It remains to show they are positive.
		
		Suppose $\bm d^k_*$ agrees with $\bm d^k$ in the $i$-th component, \textit{i.e.}\ $(\bm a^k)_i=0$.
		This cannot happen in every entry, as otherwise $\bm d^k$ would have been feasible.
		By part 1., we have that $\bm a^k$ is normal to a supporting hyperplane for $\mathcal D+\mathbb R^n_{\geq 0}$ at $\bm d^k_*$.
		Furthermore, $\ip{\bm a^k,\bm d^k_*}>0$, since $\bm a^k$ has non-negative entries with at least one positive entry, and the entries of $\bm d^k_*$ are positive. 
		However, letting $\bm e_i$ be the $i$-th standard basis vector, we have that $\d_i(\bm 1)\bm e_i$ lies in $\mathcal D$, but $\ip{\bm a^k,\d_i(\bm 1)\bm e_i}=0<\ip{\bm a^k,\bm d^k_*}$, a contradiction.
		Thus the entries of $\bm a^k$ must all be positive as well. 
		
		This allows us to show that $\bm d^k_*\in \mathcal D$.
		Suppose not, then $\bm d^k_*\in\mathcal D+(\mathbb R^n_{\geq 0} \setminus\{\bm 0\} )$. 
		We have established that coordinate-wise, $\bm d^k_* > \bm d^k$.
		Therefore, there must exist some coordinate $i$ such that we can reduce $(\bm d^k_*)_i$ and remain feasible, contradicting the minimality.
		
		Finally, since we simply define $\bm d^{k+1}$ as being the coordinate-wise inverse of $\bm a^k$, it must also have positive entries.	
	\end{proof}

\begin{claim}
	\label{convexityofprogram}
	Program~\ref{convex-pgm-nearest-point} is a convex program, i.e. it minimizes a convex function over a convex domain.
\end{claim}

\begin{proof}
    Since $\sum_{i \in [n]} \beta_i^2$ is a convex function, our program indeed minimizes a convex function. We now argue that the domain is convex. To this end, first note that the constraints $\sum_{i \in [n]} (\bm z)_{ij} \geq 1$ for all $j \in [m]$,  and $(\bm z)_{ij} \geq 0$ for all $i \in [n]$, $j \in [m]$ are linear constraints and the intersections of all these half-spaces is convex. Now, we show that each constraint $D_i(\bm z_i) - (\bm d^k)_i - (\bm \beta)_i \leq 0$ is also convex. Let $(\bm z^1, \bm \beta^1)$ and $(\bm z^2, \bm {\beta^2})$ be two points such that $D_i(\bm z^1_i) - (\bm d^k)_i - (\bm \beta^1)_i \leq 0$ and $D_i(\bm z^2_i) - (\bm d^k)_i - (\bm \beta^2)_i \leq 0$. Fix any $\lambda \in (0,1)$ and let $ (\bm z', \bm \beta') = \lambda \cdot (\bm z^1, \bm \beta^1) + (1-\lambda) \cdot (\bm z^2, \bm \beta^2)$. To prove convexity of the domain, it suffices to show that $D_i(\bm z'_i) - (\bm d^k)_i - (\bm \beta')_i \leq 0$. To this end, note that 
     \begin{align*}
		D_i(\bm z'_i) - &(\bm d^k)_i - (\bm \beta')_i\\ 
		&= D_i(\lambda \cdot \bm z^1_i + (1-\lambda) \cdot \bm z^2_i) - (\bm d^k)_i - \lambda \cdot (\bm \beta^1)_i - (1-\lambda) \cdot (\bm \beta^2)_i\\
		&\leq  \lambda \cdot D_i(\bm z^1_i) + (1-\lambda) \cdot D_i( \bm z^2_i) - (\bm d^k)_i - \lambda \cdot (\bm \beta^1)_i - (1-\lambda) \cdot (\bm \beta^2)_i &\text{($D_i(\cdot)$ convex)}\\
		&= \lambda \cdot (D_i(\bm z^1_i) - (\bm d^k)_i  - (\bm \beta^1)_i)  + (1-\lambda) \cdot (D_i(\bm z^2_i) - (\bm d^k)_i  - (\bm \beta^2)_i)\\
		&\leq 0 + 0 = 0.   \qedhere                             
\end{align*}
\end{proof}

\section{Approximate CEEI to Approximate EF and PO}\label{app:ef-po}
CEEI allocations are known ensure envy-freeness and Pareto-optimality. 
In this section we show that we can determine an allocation that is approximately envy-free and approximately Pareto-optimal from an \eCEEI.
To this end, let $\bm x_1, \bm x_2, \dots , \bm x_n$ and $\bm p$ be the allocation and the price vector at a $(1-\varepsilon)$-CEEI. Let $\alpha_j = \sum_{i \in [n]}x_{ij}$. 
Note that if $\alpha_j > 1$, then $j$ is over allocated and if $\alpha_j < 1$, then $j$ is under allocated. 
Let $\alpha = \mathit{max}_j \alpha_j$ and $\alpha' = \mathit{min}_j \alpha_j$. 
Note that $1-\varepsilon \leq {\alpha}'\leq \alpha\leq 1 + \varepsilon$. 
We define a new allocation $\bm y = (\bm y_1, \bm y_2, \dots , \bm y_n)$ such that $y_{ij} = x_{ij} / (\alpha_j)$.
Note that for all chores $j \in [m]$, we have $\sum_{i \in [n]}y_{ij} = 1$, \textit{i.e.}\ no chore is over-allocated or under-allocated. 
Observe that the total earning for each agent has not changed significantly: for each $i \in [n]$, we have $(1 / \alpha) \cdot \langle \bm x_i, \bm p \rangle \leq \langle \bm y_i, \bm p \rangle \leq (1 / \alpha') \cdot \langle \bm x_i, \bm p \rangle$. 
Similarly, by 1-homogeneity, the disutilities for the agents have also not changed significantly: for each $i \in [n]$, we have $(1/ \alpha) \cdot \d_i(\bm x_i) \leq \d_i( \bm y_i) \leq (1/\alpha') \cdot \d_i(\bm x_i)$.

\begin{claim}
	Allocation $\bm y$ is $(1-4\varepsilon)$-envy-free, {\em i.e.}\ for all pairs of agents  $i$ and $i'$, we have $(1-4\varepsilon) \cdot \d_i(\bm y_i) \leq  \d_i (\bm y_{i'})$.
\end{claim}

\begin{proof}
	We first show that allocation $\bm x$ is approximately-envy-free. Thereafter, since the disutilities of the agents in allocation $\bm y$ are not significantly different from their disutilities in $\bm x$, the approximate-envy-freeness for $\bm y$ will follow easily. Consider two agents $i$ and $i'$. Let $\bm x'_{i'} = (1-\varepsilon)^{-1} \bm x_{i'}$.  Since $\bm x$ is at a  $(1-\varepsilon)$-CEEI, by condition~\ref{P1} in Definition~\ref{def:CE}, we have $(1-\varepsilon)\langle \bm x_i, \bm p \rangle \leq \langle \bm x_{i'} , \bm p \rangle$, implying that $\langle \bm x_i, \bm p \rangle \leq \langle \bm x'_{i'} , \bm p \rangle$. Therefore, $(1-\varepsilon) \cdot \d_i(\bm x_i) \leq  \d_i(\bm x'_{i'}) = (1-\varepsilon)^{-1} \cdot \d_i(\bm x_{i'})$ (by 1-homogeneity), further implying that $(1-\varepsilon)^2 \cdot \d_i(\bm x_i) \leq  \d_i(\bm x_{i'})$.
	
	Now we show approximate-envy-freeness for $\bm y$: Since  $(1/ \alpha) \cdot \d_i(\bm x_i) \leq \d_i( \bm y_i) \leq (1/\alpha') \cdot \d_i(\bm x_i)$, we have that
	\begin{align*}
	 \d_i(\bm y_{i'}) \ \geq\  \frac{1}{\alpha} \cdot \d_i(\bm x_{i'})
	                 \ \geq\  \frac{(1-\varepsilon)^2}{\alpha} \cdot \d_i(\bm x_i)
	                 \ \geq\  \frac{(1-\varepsilon)^2 \cdot \alpha'}{\alpha} \cdot \d_i(\bm y_i).
	\end{align*}
   Using the fact that $\alpha \leq 1+\varepsilon$ and $\alpha' \geq 1-\varepsilon$ we have that
	\begin{align*}
	\d_i(\bm y_{i'}) \ \geq\ \frac{(1-\varepsilon)^3} { 1+\varepsilon} \cdot \d_i(\bm y_i)
	                \ \geq\ (1-4\varepsilon)  \cdot \d_i(\bm y_i). &\qedhere
	\end{align*} 
\end{proof}

\begin{claim}
	Allocation $\bm y$ is $(1-2\varepsilon)$-Pareto-optimal, i.e., there exists no allocation $\bm y'$ such that $\d_i(\bm y'_i) \leq (1-2\varepsilon) \cdot \d_i(\bm y_i)$ for all $i \in [n]$ with at least one strict inequality, 
	and such that $\bm y'$ does not under- or over-allocate any item.
\end{claim}

\begin{proof}
	Assume otherwise. Let $\bm y'$ be an allocation such that $\d_i(\bm y'_i) \leq (1-\varepsilon) \cdot \d_i(\bm y_i)$ for all $i \in [n]$ and let $i'$ be an agent such that $\d_{i'}(\bm y'_{i'}) < (1-2\varepsilon) \cdot \d_{i'}(\bm y_{i'}) \leq (1-\varepsilon) \cdot \d_{i'}(\bm y_{i'}) $.  Since $\bm y$ is at a $(1-\varepsilon)$-CEEI, by condition~\ref{P2}, in Definition~\ref{def:CE}, we can conclude that $\langle  \bm y'_{i'}, \bm p \rangle < \langle \bm y_{i'}. \bm p \rangle$.  Note that since both $\bm y$ and $\bm y'$ are feasible allocations we have 
	\begin{align*}
	 \sum_{i \in [n]} \langle \bm y'_i, \bm p \rangle 
	                                                  \ =\sum_{j \in [m]} p_j \cdot \sum_{i \in [n]} y'_{ij}
	                                                  \ = \sum_{j \in [m]} p_j
	                                                  \ =\sum_{j \in [m]} p_j \cdot \sum_{i \in [n]} y_{ij}
	                                                  \ =\sum_{i \in [n]} \langle \bm y_i, \bm p \rangle. 
	\end{align*}
	The last two equalities require that $\bm y'$ does not under- or over-allocate any item, as we have assumed.
	The original allocation $\bm y$ satisfies these conditions by construction.
These equalities imply $\sum_{i \in [n]} \langle \bm y'_i - \bm y_i, \bm p \rangle = 0$. Since we have $\langle  \bm y'_{i'}, \bm p \rangle < \langle \bm y_{i'}, \bm p \rangle$, there must be an agent $\ell$ such that $\langle  \bm y'_{\ell}, \bm p \rangle > \langle \bm y_{\ell}, \bm p \rangle$. Again, by condition~\ref{P2} in Definition~\ref{def:CE}, we can conclude that $\d_{\ell}(\bm y'_{\ell}) \geq (1-\varepsilon) \cdot \d_{\ell}(\bm y_{\ell}) > (1-2\varepsilon) \cdot \d_{\ell}(\bm y_{\ell})$, which is a contradiction.	
\end{proof}

\end{document}